\documentclass[onecolumn,journal,noadjust]{IEEEtran}
\usepackage{cite}

\ifCLASSINFOpdf
\else
\fi
\usepackage{amssymb}
\usepackage{dsfont}
\usepackage{float}
\usepackage{graphicx}
\newtheorem{lemma}{Lemma}

\newtheorem{definition}{Definition}
\newtheorem{proposition}{Proposition}
\newtheorem{theorem}{Theorem}
\newtheorem{remark}{Remark}
\newtheorem{corollary}{Corollary}
\usepackage{xcolor}
\usepackage{mathrsfs}
\usepackage{mathtools}
\usepackage{enumitem}
\usepackage{caption}
\usepackage[ colorlinks = true,              linkcolor = blue, urlcolor  =
blue,              citecolor = red,              anchorcolor = green,
]{hyperref}

\newcommand{\ind}{\perp \!\!\! \perp}

\usepackage [autostyle, english = american]{csquotes}
\MakeOuterQuote{"}

\usepackage{amsmath}
\DeclareMathOperator*{\argmax}{arg\,max}
\DeclareMathOperator*{\argmin}{arg\,min}
\begin{document}
%

\title{Lossy Joint Source-Channel Coding over\\Unknown Channels}
%
%
%

\author{%
  \begin{tabular}[t]{@{}c@{\hspace{6em}}c@{\hspace{6em}}c@{}}
    Adeel Mahmood & Harish Viswanathan & Jinfeng Du\\
    Nokia Bell Labs & Nokia Bell Labs & Nokia Bell Labs\\
    Murray Hill, NJ, USA & Murray Hill, NJ, USA & Murray Hill, NJ, USA
  \end{tabular}
}

\maketitle


\begin{abstract}
We analyze the performance of joint
source-channel codes in an unknown-channel framework, where the
true channel is unknown but the source distribution is known. We derive achievability bounds for a family of mismatched-design joint source-channel codes constructed for a design channel $Q_{Y|X}$ and operated
over a possibly different true channel $P_{Y|X}$. Our one-shot achievability bound allows for standard Borel alphabets for the source, reproduction, channel input and channel output. The subsequent block coding result based on the normal approximation applies to
stationary memoryless sources and memoryless, possibly nonstationary
channels under regularity and moment conditions. The achievability
bound is given in terms of the rate-distortion and rate-dispersion
functions, as well as two channel-dependent quantities that we call
the mismatched-design rate and mismatched-design rate-dispersion. We
use a family of Gibbs posteriors parameterized by a single scalar as decoder-side kernels,
and the envelope of the corresponding achievable rates recovers the
generalized mutual information. In the stationary matched setting
covered by our assumptions, our result recovers the achievability part
of Kostina and Verd\'u's 2013 Gaussian approximation result and improves its
third-order term. We also formalize a notion of a second-order universal family of source-channel codes under which there is no first- or second-order asymptotic penalty. We then construct two
channel-blind families of source-channel codes: one that is
second-order universal over a regular class of nonstationary block erasure channels and
another that is second-order universal over stationary Gaussian
channels. Our code construction
uses Poisson functional representations of suitable conditional
probability measures to produce the encoder and decoder outputs.
\end{abstract}

\section{Introduction}

With the advent of AI, joint source–channel coding (JSCC) has attracted significant recent interest \cite{gunduz-wigger, tung-gunduz, 8723589, 9398576}. In AI-based JSCC, encoder models are trained to map source data directly to channel inputs, while decoder models are trained to reconstruct the source data with the highest possible fidelity. Image and video datasets, combined with synthetic additive white Gaussian noise (AWGN) channels, are commonly used for training. Often, the encoder and decoder models are trained at a single nominal channel signal-to-noise ratio (SNR) and then operated over a range of SNR values \cite{8723589,tung-gunduz}. Similarly, \cite{homapaper2} has studied the semantic communication performance by training a source-channel code over block erasure channels with non-uniform erasure probabilities. Since the actual channel is unknown in such training frameworks, we investigate achievable performance bounds for lossy JSCC when the source distribution is assumed to be known but the actual channel statistics are unknown to both the encoder and the decoder. In this unknown-channel framework, we present results for two classes of random joint source-channel codes: the \emph{mismatched-design} class and the \emph{universal-design} class.

Let $P_{Y|X}$ denote the unknown channel encountered at deployment,
which we also call the true channel. Mismatched-design JSCC for unknown channels is driven by the needs of real deployments. On the one hand, for transmission over mobile channels, only an estimate $Q_{Y|X}$ of the true channel $P_{Y|X}$ may be known, so that the JSCC design may be $Q_{Y|X}$-based. On the other hand, replacing separate source and channel coding by JSCC in real communication networks is difficult~\cite{homaspaper1}  because practical networks are inherently modular: application providers responsible for source coding and network providers responsible for channel coding are typically distinct entities connected through standardized packet interfaces. It is impractical to  relay the real-time channel state information to the application provider to enable channel-aware source coding. In this case, neural joint source-channel codes are trained for a fixed channel abstraction $Q_{Y|X}$ while the effective channel $P_{Y|X}$ at the network may be molded to optimize transmission.

Recent works \cite{homapaper2}, \cite{homaspaper1} proposed a new standard interface, called the multi-level reliability interface, between the application layer and the network in order to abstract the network channel into multiple reliability levels for the application provider. The multiple reliability levels motivate a nonstationary channel model for the channel abstraction $Q_{Y|X}$. Specifically, the channel abstraction is chosen to be block erasure channels (BLECs) with different erasure probabilities in \cite{homapaper2}. The motivation for this choice is that in real packet networks, decoding failures are naturally represented as packet erasures, with the decoder explicitly knowing which blocks/packets are missing. Hence, the true channel $P_{Y|X}$ can be modeled as a nonstationary BLEC, with different packet erasure probabilities that are shaped by network-side resource allocation\footnote{This could include orthogonal resource allocation weights or different power allocations in a power-domain superposition scheme \cite{adeelUEP}.} and scheduling. The source coder at the application layer is then trained as a JSCC over the nonstationary block erasure channel abstraction $Q_{Y|X}$, thereby inducing a learned encoded representation comprising blocks (or packets) with nonuniform block importance. The channel coder at the network optimizes the transmission of the message packets based on packet importance weights \cite{adeelUEP} that are also derived from the multi-level reliability interface.  Understanding the achievable performance in the channel-mismatched framework provides valuable insights in optimizing the design and operation: the application uses a JSCC trained for a channel abstraction $Q_{Y|X}$ while the network chooses its parameters and policies, including unequal error protection (UEP) and selective packet dropping, to optimize the effective channel $P_{Y|X}$.

To derive our mismatched-design result, we consider a specific family of randomized
$Q_{Y|X}$-based source-channel codes, where we call $Q_{Y|X}$ the design channel. The law of the common randomness shared by the encoder and
decoder is fixed at design time, and the channel-dependent components
of the construction---namely, the channel input distribution $Q_X^*$
and the decoder-side kernel---are determined by $Q_{Y|X}$; the
parameters of $P_{Y|X}$ do not enter the code construction. This makes our study different from the analysis of mismatch capacity in the mismatched decoding framework \cite{revisited_paper}, which entails optimizing over the channel input law while the decoding metric is arbitrary but fixed. Through refined asymptotic analysis based on the normal approximation, we analyze the performance of these codes when they are operated
over $P_{Y|X}$. We refer to the resulting achievability statements
under channel mismatch as \emph{mismatched-design} results.

Our JSCC construction departs from the conventional framework of shared i.i.d. source and channel codebooks between the encoder and decoder. Instead, following \cite{alternative_oneshot_JSCC}, we have a marked Poisson point process as shared randomness, and the encoder and decoder obtain the channel input and source reconstruction, respectively, using the Poisson functional representation \cite[Definition 1]{alternative_oneshot_JSCC} of suitable conditional probability measures. The source distribution $P_S$ is assumed to be fully known at design time, so this paper does not focus on source mismatch. 

Our one-shot results in the mismatched-design framework encompass arbitrary sources and channels on standard Borel spaces, albeit with certain technical assumptions which ensure, in particular, that the information density of the joint distribution $Q_X^* \times Q_{Y|X}$ is well-defined under the mismatched distribution $Q_X^* \times P_{Y|X}$.  On the other hand, our block coding results based on the normal approximation encompass memoryless, stationary sources and memoryless, but not necessarily stationary, channels, whose single-letter probability laws are arbitrary but satisfy mild regularity conditions. We prove refined asymptotic achievability bounds in the memoryless setting in terms of the usual source-related quantities, the rate-distortion and rate-dispersion functions \cite{kostina_SC}, as well as two channel-related single-letter quantities, which we call the \textit{mismatched-design rate} and the \emph{mismatched-design rate-dispersion}. A substantive contribution of this work, preceding the normal-approximation analysis, is a measure-theoretically rigorous formulation of mismatched JSCC on standard Borel alphabets, with conditions ensuring that arbitrary decoder kernels and Radon–Nikodym information densities are compatible with expectations under the true-channel law and that the resulting mismatched-design rate is independent of the chosen density version of the design channel law.

For block erasure channels (BLECs), the mismatched-design rate and mismatched-design rate-dispersion are simply the usual capacity and dispersion of the true block erasure channel $P_{Y|X}$ as long as $Q_{Y|X}$ has an erasure probability in $(0, 1)$. Hence, when both $P_{Y|X}$ and $Q_{Y|X}$ are block erasure channels, with $Q_{Y|X}$ having erasure probabilities in $(0, 1)$, we show that the family of $Q_{Y|X}$-based source-channel codes achieves the same first- and second-order asymptotic performance as the family of optimal $P_{Y|X}$-matched codes.
We use this observation to formalize a notion of second-order universal source-channel codes and then construct such codes over a regularized set of nonstationary block erasure channels and over the set of stationary Gaussian channels. At a high level, a JSCC scheme is second-order channel-universal over a set of channels $\mathcal{W}$ if its encoder and decoder do not depend on the true $P_{Y|X} \in \mathcal{W}$, yet for every $P_{Y|X} \in \mathcal{W}$, it achieves the same first- and second-order asymptotic performance as the $P_{Y|X}$-matched optimal JSCC scheme. This universal result should be understood as a property of a family of codes, indexed by the channel blocklength $n$
and source blocklength $k$. Our universal results do not imply the existence of a single source-channel code that is optimal for every channel in a given set $\mathcal{W}$. This is because a single source-channel code $U_{n, k}^{(d)}$ in our universal family is associated with an operating point $(n, k)$ and also depends on the target distortion level $d$. Then for a certain channel $P_{Y|X}$ in a channel set $\mathcal{W}$, the code $U_{n, k}^{(d)}$ might be optimal in the sense that the rate $R = k/n$ is the best achievable rate for that channel and the distortion level $d$, while for other channels in $\mathcal{W}$, a higher rate might be achievable, making the code $U_{n, k}^{(d)}$ suboptimal. Even though a single code $U_{n, k}^{(d)}$ need not be optimal for every channel, we show that for any given channel in $\mathcal{W}$, the code $U_{n, k}^{(d)}$ in our universal construction achieves any fixed nonzero excess-distortion probability under a sufficient condition that is asymptotically tight to second-order (Remark \ref{remarkeeusefulprop}). We make the notion of universality rigorous in Section \ref{universalpluginconstruction} and present our universal results in Sections \ref{tiredaf} and \ref{tiredafgaussian}.  

\subsection{Related Work and Contributions}

Mismatched-design frameworks have been studied in several works under a variety of closely related formulations and mostly in the context of channel coding. Second-order asymptotics of the achievable rate based on the normal approximation were studied in \cite{7605463} for the case when a channel code optimized for AWGN channels (i.e., one that uses spherical codebooks and nearest neighbor decoding) is deployed over arbitrary additive (non-Gaussian) noise memoryless channels; first-order results for this same framework were previously derived by Lapidoth \cite{lap_channel_coding}. The "mismatch" in \cite{7605463} and \cite{lap_channel_coding} is that between a Gaussian model and a non-Gaussian model for the additive noise in the channel. However, most works on mismatched-design frameworks have focused on mismatched decoding in the context of channel coding, where the decoding rule is the argmax of a fixed, possibly suboptimal, decoding metric. In this mismatched maximum-metric decoding formulation, first-order asymptotically achievable rates given by the generalized mutual information (GMI) and the lower bound on mismatched (LM) capacity have been derived; see the survey paper \cite{survey_mismatched.decoding} and references therein. In particular, \cite{revisited_paper} derives lower bounds to the mismatch capacity in several settings, which include general alphabet channels with memory and non-single-letter decoding metrics. Second-order asymptotics of achievable rates (second-order coding rates) with mismatched decoding are studied in \cite{scarlett_mismatched_decoding}, featuring the GMI and the LM rate together with the dispersion as the second-order term in the asymptotic expansion \cite[Theorem 5]{scarlett_mismatched_decoding}.

Our mismatched-design formulation is closer in spirit to the mismatched-decoding setting of \cite{revisited_paper}, \cite{survey_mismatched.decoding} and \cite{scarlett_mismatched_decoding}, but differs in several aspects. First, we derive achievability results for JSCC rather than for channel coding, which requires proof techniques different from the random-coding union bound used in \cite[Theorem 1]{scarlett_mismatched_decoding}. As we will see later\footnote{See the discussion surrounding \eqref{differentchannelcodingJSCC}.}, the analysis of a channel-mismatched JSCC decoder differs from that of a mismatched decoder in the channel coding setting. Second, we restrict our attention to a parameterized family of decoding metrics induced by the design channel $Q_{Y|X}$ rather than completely arbitrary decoding metrics. However, our general JSCC construction in Definition \ref{unified_def} allows the decoder-side Gibbs posterior $K_{X|Y}$ to be chosen as a metric-induced Gibbs posterior\footnote{We elaborate on this in our concluding remarks; see \eqref{mentionconcccc}.}. Third, unlike \cite{revisited_paper}, our mismatched-design results are in the finite-blocklength approximation setting similar to the second-order asymptotic result in \cite[Theorem 5]{scarlett_mismatched_decoding}. Unlike \cite[Theorem 5]{scarlett_mismatched_decoding}, our main mismatched-design result (Theorem \ref{thm_finite_block1}) applies to nonstationary channels, specifies an explicit third-order term and is in the JSCC framework instead of the channel coding framework. Our third-order term also improves upon that presented in \cite{kostina_JSCC} for the traditional matched JSCC. Our mismatched-design result recovers the GMI as the envelope over a parameterized family of $Q_{Y|X}$-based JSCC designs. We also highlight the significance of our $Q_{Y|X}$-based mismatched-design architecture in a post-design receiver-only optimization paradigm (see discussion after Theorem \ref{thm_finite_block1}). 

Literature on mismatched JSCC frameworks is relatively sparse, especially in the finite-blocklength regime. Optimal mismatched analysis, or even standard matched analysis, in JSCC in the finite-blocklength regime is harder than in channel coding because an optimal JSCC decoder in the finite-blocklength regime does not generally reduce to a maximum-likelihood (ML) channel decoder with equiprobable messages; a JSCC decoder must generally account for the residual encoded-source redundancy as we will explain later in the discussion section after Theorem \ref{thm_finite_block1}. Zhou, Tan and Motani \cite{mismatchJSCC} studied a JSCC architecture called NN-JSCC that is constrained to use random codebooks (for both source reproduction and channel codewords) that are either i.i.d. Gaussian or spherical. However, instead of nearest neighbor encoding and decoding rules that are optimal for Gaussian sources and channels, they use modified minimum distance and nearest neighbor schemes that serve to maximize the communication rate when transmitting an unknown arbitrary memoryless source over an unknown arbitrary additive noise memoryless channel, although the source power and noise power are known and fixed. 


The NN-JSCC is part of the mismatched-design class of JSCC. Indeed, the fact that the NN-JSCC uses random codebooks optimized for Gaussian channels represents exactly a channel-mismatch design because such codebooks are not necessarily optimal for general additive non-Gaussian noise channels. Hence, results for NN-JSCC characterize the performance under mismatch between a Gaussian model and a non-Gaussian model for the source and channel, while the source power and channel noise power are matched and known \cite[(2)]{mismatchJSCC}. 
The NN-JSCC is "robust" in the sense that for any arbitrary source with fixed source power (or second-moment) $\sigma^2$ and an additive channel with fixed SNR $P$, the first-order asymptotic rate (number of source symbols per channel use) achieved by the NN-JSCC is $C(P)/R_d(\sigma^2)$, where $C(P)$ and $R_d(\sigma^2)$ are the Gaussian capacity and rate-distortion function, respectively. A similar notion of robustness was proved in \cite{lap_channel_coding} for channel codes and in \cite{lap_source_coding} for lossy source codes. 

However, universality is a stronger notion than robustness. Our results on universal source-channel codes are therefore limited to certain classes of channels, namely nonstationary block erasure channels and stationary Gaussian channels, and assume a known source distribution. In particular, our universal JSCC results for Gaussian channels do not include additive non-Gaussian noise channels, but they allow for an arbitrary and unknown noise power unlike NN-JSCC. The latter is significant because, at finite blocklength, a channel-matched JSCC decoder that exploits the nonuniform prior on the codeword indices induced by source encoding generally depends on the channel noise power (see $(\ref{seethisforexampleknownp2})$ for example). In contrast, the conventional ML channel decoder used in separate source-channel coding (SSCC) over an AWGN channel reduces to the ordinary minimum-distance rule and does not require knowledge of the noise power. Hence, a channel-matched joint source-channel code can exhibit a stronger dependence on channel knowledge at the decoder than conventional channel coding. To investigate whether this dependence translates into a first- or second-order asymptotic performance penalty, we study the universal-design class of JSCC.  

The notion of universal codes was first introduced in the source coding literature \cite{davisson}. Universal source codes have two defining characteristics: they do not have knowledge
of the source distribution, and they achieve the same asymptotic
performance as optimal codes that do. For joint source-channel codes, we refer to the analogous first property as \emph{channel-blindness}. For the second property, we strengthen the usual notion from universal source coding: instead of requiring only first-order asymptotic optimality, we require the channel-blind JSCC construction to achieve both first- and second-order asymptotically optimal performance, matching an optimal channel-matched JSCC scheme. This is motivated by the fact that first-order optimality is already achieved by separate source-channel coding, so the study of first-order universal joint source-channel codes 
over a set of channels can simply be reduced to a universal channel coding problem. 

To prove our universal results, we first construct a channel-blind family of codes and then derive achievable normal-approximation bounds for these codes. For stationary Gaussian channels, second-order optimality of the said achievability bounds can be proven by invoking the converse result for stationary channels from \cite{kostina_JSCC}. However, to establish second-order optimality for nonstationary block erasure channels, we derive a new converse result (Theorem \ref{convoghlswhr0vm} and Corollary \ref{corocoro}). Beyond proving the second-order optimality of our universal construction for nonstationary block erasure channels, Theorem \ref{convoghlswhr0vm} is novel in that the normal-approximation bound therein holds uniformly over all $n \geq 1$ and all $p_1, \ldots, p_n \in [0, 1]$, where $n$ is the number of channel uses and $p_1, \ldots, p_n$ are the block erasure probabilities. Such a result has significance given the recent use of synthetic block erasure channels in training JSCC-based semantic communication systems as well as the block erasure channels being used as a model for modern packet networks \cite{homapaper2}, \cite{fayaz2026}. Often in these applications, a single block size $m$ is large so that the 
transmission of a source uses a small number of relatively large packets, which represents a small number $n$ of independent block erasure
channel uses. This makes the $n$-uniform normal-approximation bound presented in Theorem \ref{convoghlswhr0vm} more relevant than standard JSCC normal-approximation bounds that require both the channel blocklength $n$ and source blocklength $k$ to be sufficiently large.    

All source-channel code constructions and achievability results in this paper concern randomized codes, where the common randomness between the encoder and decoder is a marked Poisson point process. In particular, the excess-distortion probabilities of these codes are averaged over this shared randomness in addition to the source and channel randomness. For the $Q_{Y|X}$-based constructions, the distribution of the randomized code does not depend on the true channel $P_{Y|X}$. For each fixed true channel, the standard averaging argument implies the existence of some deterministic realization of a code whose excess-distortion probability is no larger than the ensemble average. However, the selected realization may depend on the true channel. Consequently, this pointwise argument does not establish a deterministic $Q_{Y|X}$-based code whose design is independent of $P_{Y|X}$. Similarly, 
for our universal results in Theorems \ref{BEC_particularization} and \ref{thm_finite_block1_awgn_univ}, choosing a good deterministic realization separately for each true channel could make the chosen realization of the code depend on that channel, violating channel-blindness. Establishing deterministic channel-blind counterparts would require a simultaneous derandomization argument over the relevant class of channels and is beyond the scope of this paper. In general, there is no
simple way to deduce the existence of a universal family of  deterministic codes from a universal family of random codes. In the universal source coding literature, a few subtle derandomization techniques are used instead; see, e.g., \cite[proof of Theorem 3]{minimaxRD} and \cite[proof of Theorem 2]{mryang}.

The remainder of the paper is organized as follows. Section \ref{dj80dnkfd} provides the notation, definitions and regularity assumptions that hold throughout the paper. Section \ref{unified_construction} provides the general one-shot JSCC construction as well as the specialization of the general one-shot construction to obtain a parameterized family of $Q_{Y|X}$-based source-channel codes. Section \ref{slqmda_section} provides  the definitions of mismatched-design rate and mismatched-design rate-dispersion. Section \ref{qyxbasedsourcechannelcodesresults} presents the mismatched-design results in the memoryless block coding framework for $Q_{Y|X}$-based source-channel codes when the true (unknown) channel is $P_{Y|X}$. Section \ref{3i8rt48} gives the definition of second-order universality and constructs second-order universal families of source-channel codes over specific sets of channels. Section \ref{conclusion} gives the concluding remarks and possible future directions. Appendices \ref{bxc_proof}--\ref{mytwosidedbndprop_proof} contain the technical proofs, while 
Appendix \ref{examples_of_sources} provides a few examples of source distributions that satisfy the regularity conditions of our main results.

\section{Preliminaries \label{prelims_section}}

\subsection{Notation, Regularity Assumptions and Basic Definitions \label{dj80dnkfd}}

Let $\mathcal{S}, \mathcal{X}, \mathcal{Y}$ and $\mathcal{Z}$ be the source, channel input, channel output and source reconstruction alphabets, respectively. Focusing on the one-shot paradigm, we will let these alphabets be arbitrary standard Borel spaces, with no imposition of a Cartesian product structure. We will write $\mathcal{B}(\mathcal{S})$ to denote the Borel $\sigma$-algebra on $\mathcal{S}$. When needed for certain results later in the paper, we will instantiate the alphabets as $n$-fold Cartesian products.     

Throughout the paper, we write $\log$ to denote the natural logarithm and $\exp(x)$ to mean $e$ to the power $x$. We write $Q$ to denote the complementary CDF of the standard normal distribution, and $Q^{-1}$ denotes its inverse. Let $(x)^+ \coloneqq \max \{ 0,x\}$. For any input distribution $Q \in \mathcal{P}(\mathcal{X})$ and a channel $W$ from $\mathcal{X}$ to $\mathcal{Y}$, we write $Q \times W$ for the joint distribution on $\mathcal{X} \times \mathcal{Y}$ and $QW$ for the induced marginal distribution on $\mathcal{Y}$.  
Given a kernel $K_{X|Y}$ from $\mathcal{Y}$ to $\mathcal{X}$ and a distribution $K_{\overline{X}} \in \mathcal{P}(\mathcal{X})$, we denote the information density as 
\begin{align}
    \imath_{K_{X|Y}, K_{\overline{X}}}(y, x) \coloneqq 
        \log \left(\frac{dK_{X|Y=y}}{d K_{\overline{X}} }(x)\right), \label{conven}
\end{align}
where we define the RHS only for $K_{\overline{X}}$-a.e. $x$ and for $y \in \mathcal{Y}$ such that $K_{X|Y = y} \ll K_{\overline{X}}$. We use the convention $\log 0 = -\infty$. 
Note that the definition $(\ref{conven})$ of information density does not require that there exists a distribution $K_{Y}$ such that $K_{\overline{X}}$ is the $\mathcal{X}$-marginal induced by the joint $K_Y \times K_{X|Y}$. Thus, $K_{X|Y}$ need not arise as the posterior induced by the prior $K_{\overline{X}}$ and any auxiliary channel from $\mathcal{X}$ to $\mathcal{Y}$.    

A lossy source-channel code is a pair of (possibly randomized) mappings $\operatorname{f} : \mathcal{S} \to \mathcal{X}$ and $\operatorname{g} : \mathcal{Y} \to \mathcal{Z}$. A distortion
measure $\operatorname{d} : \mathcal{S} \times \mathcal{Z} \to [0, \infty)$ is used to quantify the performance of the lossy code. The code does not rely on feedback.

\begin{definition}
A source-channel code $(\operatorname{f}, \operatorname{g})$ is called a $(d, \epsilon)$ source-channel code with respect to $\{P_S, \operatorname{d}, P_{Y|X} \}$ if \\$\mathbb{P}(\operatorname{d}(S, \operatorname{g}(Y)) > d) \leq \epsilon$, where $X = \operatorname{f}(S)$, $S \sim P_S$ and $Y|X \sim P_{Y|X}$. 
\label{wocostdef}
\end{definition}

For any channel $Q_{Y|X}$ from $\mathcal{X}$ to $\mathcal{Y}$, define  
\begin{align}
    C(Q_{Y|X}) \coloneqq \sup_{Q_X} I(Q_X, Q_{Y|X}), \label{cap}
\end{align}
where $I(Q_X, Q_{Y|X})$ denotes the mutual information. For any source distribution $P_S \in \mathcal{P}(\mathcal{S})$, distortion level $d > 0$ and distortion measure $\operatorname{d} : \mathcal{S} \times \mathcal{Z} \to [0, \infty)$, define 
\begin{align}
    R_d(P_S) &\coloneqq \inf_{\substack{P_{Z|S}:\\\mathbb{E}[\operatorname{d}(S, Z)] \leq d}} I(P_S, P_{Z|S}). \label{64}
\end{align}
We make the following source-distortion regularity  assumptions (cf. \cite[p. 2547]{kostina_JSCC} and \cite[p. 3310]{kostina_SC}) throughout the paper.

\textbf{Source-distortion regularity  assumptions:}  These are regularity conditions on the source distribution $P_S$, distortion level $d$ and distortion measure $\operatorname{d}$.
\begin{enumerate}[label=A\arabic*., ref=A\arabic*]
    \item \label{assumpr1_source} The infimum in $(\ref{64})$ is achieved by a unique $P_{Z|S}^*$.
    \item \label{assumpr2_source} The distortion level $d > d_{\min}$, where  
    $d_{\min} \coloneqq \inf \left \{ d : R_d(P_S) < \infty\right \}$.
    \item \label{assumpr3_source} There exists a finite set $E \subset \mathcal{Z}$ such that 
    \begin{align*}
        \mathbb{E}_{P_S}\left [ \min_{z \in E} \operatorname{d}(S, z) \right] < \infty. 
    \end{align*}
\end{enumerate}

For our mismatched-design definitions and mismatched-design result in Theorem \ref{thm_finite_block1}, we make the following additional assumptions: 

\textbf{Channel-pair regularity assumptions:} These are regularity conditions on the design channel $Q_{Y|X}$ and the true channel $P_{Y|X}$: 
\begin{enumerate}[start=4, label=A\arabic*., ref=A\arabic*]
    \item \label{assumpr1_channels} The supremum in $(\ref{cap})$ is achieved by a unique $Q_{X}^*$.
     \item \label{assumpr2_channels}  $C(P_{Y|X}) < \infty$.
    \item \label{assumpr3_channels} There exists a $\sigma$-finite measure $\mu$ on $\mathcal{Y}$ such that for $Q_X^*$-almost every $x$, $P_{Y|X = x} \ll \mu$ and $Q_{Y|X = x} \ll \mu$.
\end{enumerate}

\begin{remark}
We enforce Assumption \ref{assumpr3_channels} so that $P_{Y|X = x}$ and $Q_{Y|X = x}$ can be described by densities/PMFs relative to the same reference measure. This will yield an important simplification in the mismatched setting described later. Most channel pairs $(Q_{Y|X}, P_{Y|X})$ of interest, e.g., AWGN and discrete memoryless channels (DMC), satisfy Assumption \ref{assumpr3_channels} regarding the existence of a common  dominating measure $\mu$. For AWGN channels, this will be the Lebesgue measure; for DMCs, this will be the counting measure. In light of Assumption \ref{assumpr3_channels}, we define 
\begin{align*}
    q_\mu(y|x) &\coloneqq \frac{d Q_{Y|X=x}}{d \mu}(y),\\
    p_\mu(y|x) &\coloneqq \frac{d P_{Y|X=x}}{d \mu}(y)
\end{align*}
as arbitrary $\mu$-density versions of $Q_{Y|X}$ and $P_{Y|X}$, respectively. These definitions are unique up to $(Q_X^* \times \mu )$-null sets.    
\end{remark}

Let $Q_Y^* = Q_X^*Q_{Y|X}$. Then define 
\begin{align}
    V(Q_{Y|X}) \coloneqq \operatorname{Var}_{Q_{X}^* \times Q_{Y|X}}\left( \imath_{Q_{Y|X}, Q_Y^*}\left(X, Y \right) \right). \label{mismatched_setting}
\end{align}
Let 
\begin{align}
    P_Z^*(A) \coloneqq \int_{\mathcal{S}} P_{Z|S}^*(A|s) P_S(ds), \quad A \in \mathcal{B}(\mathcal{Z}). \label{pzstar}
\end{align}
Define the $\operatorname{d}$-tilted information \cite{kostina_SC} in a source realization $s$ with source distribution $P_S$ and distortion level $d > d_{\min}$ as  
\begin{align}
    \jmath_d(s, P_S, \operatorname{d}) &\coloneqq  \log \left( \frac{1}{\mathbb{E}\left [ \exp \left( \lambda^* \left(d -  \operatorname{d}(s, Z)\right) \right) \right ]}\right), \label{oneshotdtiltedinfo}
\end{align}
where the expectation above is w.r.t. $Z \sim P_{Z}^*$ and $\lambda^* = -\partial_d R_d(P_S)$. Note that $\mathbb{E}_{P_S}\left [ \jmath_d(S, P_S, \operatorname{d})\right] = R_d(P_S)$. We define 
\begin{align}
    V_d(P_S) \coloneqq \operatorname{Var}_{P_S}\left(  \jmath_d(S, P_S, \operatorname{d}) \right).
\end{align}  
Lastly, define the distortion ball $B_d(s) \coloneqq \left \{ z \in \mathcal{Z} : \operatorname{d}(s, z) \leq d \right \}.$ 
\begin{remark}
In the block coding setting with a memoryless, stationary channel $Q_{Y|X}^{\otimes n}$ and under certain assumptions, $C(Q_{Y|X})$ and $V(Q_{Y|X})$ are the capacity and dispersion, respectively, of the channel \cite{5452208}. Similarly, for a fixed-length compression of a memoryless, stationary source $S^k \sim P_S^{\otimes k}$ with separable distortion, $R_d(P_S)$ and $V_d(P_S)$ are the rate-distortion function and rate-dispersion, respectively \cite{kostina_SC}.          
\end{remark}

\textit{Marked Poisson point processes:} Our source-channel code construction uses common randomness (fixed at design time) between the encoder and decoder in the form of a marked Poisson point process, which we briefly describe. The Lebesgue measure restricted
to the set $S \subset \mathbb{R}$ is denoted as $\lambda_S$. Let $P_X \in \mathcal{P}(\mathcal{X})$ be a probability measure. Let $\{\bar{X}_i, T_i\}_{i=1}^\infty$ denote a marked Poisson point process with intensity measure $P_X \times \lambda_{\mathbb{R}_{\geq 0}}$, where the points are ordered so that $T_1<T_2<\cdots$. Here $T_i$ is the time of the $i$th arrival, and $\bar{X}_i \in \mathcal{X}$ is the mark attached to that arrival. The two defining properties of this process are as follows. First, the number of arrivals occurring in any time interval $[t_1, t_2]$ with marks falling in a subset $A \subset \mathcal{X}$ follows a Poisson distribution with mean $P_X(A)(t_2 - t_1)$. Second, the number of arrivals across any non-overlapping time intervals or disjoint mark sets are independent random variables. Equivalently, $T_1,T_2,\ldots$ are the arrival times of a rate-$1$ Poisson process, and each arrival independently receives a mark $\bar{X}_i$ drawn from $P_X$.

\subsection{One-Shot Joint Source-Channel Coding \label{unified_construction}}

Definition \ref{unified_def} below presents a general one-shot random  $(P_Z^*, Q_X, K_{X|Y})$ JSCC construction that will be specialized later in the paper via specific choices for $P_Z^*, Q_X$ and $K_{X|Y}$.   

\begin{definition}
   Fix any $Q_X \in \mathcal{P}(\mathcal{X})$ and $K_{X|Y}$ such that $K_{X|Y=y} \ll Q_X$ for $P_Y$-a.e. $y$, where $P_Y = Q_X P_{Y|X}$ is the marginal distribution induced by the joint $Q_X \times P_{Y|X}$. Then the source-channel code given by $(\ref{unif_enc})$
   and $(\ref{unif_dec})$ below is called a $(P_Z^*, Q_X, K_{X|Y})$ source-channel code. 
   \label{unified_def}
\end{definition}

Let $P_Z^* \in \mathcal{P}(\mathcal{Z})$ be the source reproduction distribution defined in $(\ref{pzstar})$. Let $\{ (\bar{X}_i, \bar{Z}_i), T_i  \}_{i \in \mathbb{N}}$ be the points of a marked Poisson process with intensity measure $Q_X \times P_Z^* \times \lambda_{\mathbb{R}_{\geq 0}}$ and independent of the source $S \sim P_S$. Here, each mark $(\bar{X}_i, \bar{Z}_i)$ is a pair of channel and reproduction codewords which are independently generated according to the product measure $Q_{X} \times P_{Z}^*$. We assume that both the encoder and decoder of the JSCC have access to this Poisson point process as shared randomness (roughly speaking, this is analogous to the conventional framework of shared source reproduction $\{\bar{Z}_i \}_{i \in \mathbb{N}}$ and channel codebooks $\{\bar{X}_i \}_{i \in \mathbb{N}}$ between the encoder and decoder). 
Given a source realization $S =s$, define $\rho(s) \coloneqq P_Z^*(B_d(s))$. 

\textbf{Encoder Operation:} Given $S = s$, the encoder transmits 
\begin{align}
    X = \bar{X}_{i^\star(s)}, \text{ where } i^\star(s) = \begin{cases}
        \argmin_{i:\bar{Z}_i \in B_d(s)} T_i & \text{ if } \rho(s) > 0,\\
        \argmin_i  T_i & \text{ otherwise.}
    \end{cases} \label{unif_enc} 
\end{align}

\textbf{Decoder Operation:} 
For $P_Y$-a.e. $y$, the decoder is specified as 
\begin{align}
    \widehat{Z} = \bar{Z}_{i^*(y)}, \text{ where } i^*(y) = \argmax_i \left [ \log\left( \frac{d K_{X|Y=y}}{d Q_X}(\bar{X}_i) \right) - \log(T_i) \right ] \label{unif_dec} 
\end{align}  
and arbitrarily for $y$ in the $P_Y$-null set. In case of ties, the smallest index is chosen.  
\begin{remark}
$K_{X|Y}$ is a decoder-side kernel and need not be induced by the true channel $P_{Y|X}$ or by any auxiliary channel with input prior $Q_X$. This freedom allows us to use different choices of Gibbs posteriors for $K_{X|Y}$. Proposition \ref{bxc} below extends the one-shot JSCC bound \cite[Theorem 4]{alternative_oneshot_JSCC} to the general JSCC construction in Definition \ref{unified_def}, which permits an arbitrary Gibbs posterior $K_{X|Y}$ for the decoder-side kernel. The technical condition on $K_{X|Y}$ under which such an extension is possible is identified in Definition \ref{unified_def}.    
\end{remark}

\begin{proposition}[{cf. \cite[Theorem 4]{alternative_oneshot_JSCC}}]
Assume that $Q_X \times P_{Y|X} \ll Q_X \times P_Y$, where $P_Y = Q_X P_{Y|X}$. Then the ensemble-average excess-distortion probability $\overline{P}_e(d)$ with respect to $\{P_S, \operatorname{d}, P_{Y|X} \}$ of the $(P_Z^*, Q_X, K_{X|Y})$ source-channel code satisfies
    \begin{align}
        \overline{P}_e(d) &\leq \mathbb{E}\left [ \left(1 + P_Z^*(B_d(S)) \exp\left( \imath_{K_{X|Y}, Q_X}(Y, X) \right) \right)^{-1}  \right], \label{ljskfj777}
    \end{align}
    where the expectation above is w.r.t. $P_S, Q_X$ and $P_{Y|X}$. 
\label{bxc}
\end{proposition}
\textit{Proof:} The proof of Proposition \ref{bxc} is given in Appendix \ref{bxc_proof}. The proof first establishes that the encoder and decoder functions in $(\ref{unif_enc})$ and $(\ref{unif_dec})$ are equivalent to Poisson functional representations of certain conditional probability measures. The proof then applies the conditional Poisson matching lemma \cite[Lemma 2]{alternative_oneshot_JSCC}.   

The proof of Proposition \ref{bxc} is similar to that of \cite[Theorem 4]{alternative_oneshot_JSCC}. But we give it for completeness and to explain the application of the technical condition given in Definition 
\ref{unified_def} and the use of an arbitrary decoder kernel. Note that $\imath_{K_{X|Y}, Q_X}(y, x)$ is defined for $(Q_X \times P_Y)$-a.e. $(x, y)$ because of
the definition of the information density in \eqref{conven} and the technical condition on $K_{X|Y}$ in Definition \ref{unified_def}. On the other hand, the expectation in \eqref{ljskfj777} involves $(X,Y) \sim Q_X \times P_{Y|X}$. Hence, the absolute continuity assumption in Proposition \ref{bxc} ensures that the expectation in 
\eqref{ljskfj777} is well-defined. 

We now formulate the $Q_{Y|X}$-based source-channel code by directly specializing Definition \ref{unified_def}.  
\begin{definition}
    Fix any $t > 0$ and a design channel $Q_{Y|X}$. Let $Q_X^* = \argmax_Q I(Q, Q_{Y|X})$, and let $P_Y^*$ denote the $\mathcal{Y}$-marginal induced by the joint $Q_X^* \times P_{Y|X}$. Assume that  
    \begin{align}
     0 < \int_\mathcal{X} q_\mu(y|u)^t Q_X^*(du) < \infty   
\end{align}
for $P_Y^*$-a.e. $y$. Define $Q^{\sim t}_{X|Y}(\cdot|y) \in \mathcal{P}(\mathcal{X})$ as
\begin{align}
    Q^{\sim t}_{X|Y}(A|y) &\coloneqq \frac{ \int_A q_\mu(y|x)^t Q_X^*(dx)  }{\int_\mathcal{X} q_{\mu}(y|u)^t Q_X^*(du)  }  \label{QXYkernel}
\end{align}
for $P_Y^*$-a.e. $y$ and arbitrarily otherwise, where $A \in \mathcal{B}(\mathcal{X})$. We then define a $(t, Q_{Y|X})$-based source-channel code as the $(P_Z^*, Q_X^*, Q^{\sim t}_{X|Y})$ source-channel code from Definition \ref{unified_def}.
\label{qyxbasedcodedef}
\end{definition}

The following corollary is obtained by specializing the result in Proposition \ref{bxc} for a $(t, Q_{Y|X})$-based code.   
\begin{corollary}
The ensemble-average excess-distortion probability $\overline{P}_e(d)$ with respect to $\{P_S, \operatorname{d}, P_{Y|X} \}$ of the $(t, Q_{Y|X})$-based source-channel code satisfies
    \begin{align}
        \overline{P}_e(d) &\leq \mathbb{E}\left [ \left(1 + P_Z^*(B_d(S)) \exp\left( \imath_{Q_{X|Y}^{\sim t}, Q_X^*}(Y, X) \right) \right)^{-1}  \right], \label{g2vn}
    \end{align}
    where the expectation above is w.r.t. $P_S, Q_X^*$ and $P_{Y|X}$. 
\label{bc}
\end{corollary} 

\subsection{Single-Letter Quantities for Mismatched-Design Analysis \label{slqmda_section}}

When we later carry out refined asymptotic analysis in the memoryless block coding setting based on the normal approximation, then under certain conditions, the \emph{mismatched-design rate} (defined below) is the single-letter quantity that features in the first-order term in the asymptotic expansion. The second-order term in the asymptotic expansion features the single-letter quantity which we call the \emph{mismatched-design rate-dispersion}.

\begin{definition}
Consider any two channels $Q_{Y|X}$ and $P_{Y|X}$ satisfying Assumptions \ref{assumpr1_channels}\,-\,\ref{assumpr3_channels}. Let $Q_{X}^* = \argmax_Q I(Q, Q_{Y|X})$ and $P_Y^* = Q_X^* P_{Y|X}$. Let $t > 0$ such that 
\begin{align}
    0 < \int_\mathcal{X} q_{\mu}(y|u)^t Q_X^*(du) < \infty \text{ for } P_Y^*\text{-a.e. } y. \label{thisimpliesit}  
\end{align}
We then define the \emph{mismatched-design rate} and the \emph{mismatched-design rate-dispersion} as
\begin{align}
    C_{t}(P_{Y|X} \| Q_{Y|X}) &\coloneqq \left(\mathbb{E}_{Q_{X}^* \times P_{Y|X}}\left[ \imath_{Q_{X|Y}^{\sim t}, Q_X^*}\left(Y, X \right) \right] \right)^+,
     \label{gccn}\\
     V_{t}(P_{Y|X} \| Q_{Y|X}) &\coloneqq \operatorname{Var}_{Q_{X}^* \times P_{Y|X}}\left( \imath_{Q_{X|Y}^{\sim t}, Q_X^*}\left(Y, X \right) \right)
    \label{mismatched_disp_def},
\end{align}
respectively. 
\label{def_mismatched_rate}
\end{definition}

A few remarks about $(\ref{gccn})$ are in order: 
\begin{itemize}
\item \textit{Role of Condition \eqref{thisimpliesit}:} This condition ensures that $(a)$ $Q_{X|Y = y}^{\sim t}$ is a valid probability measure for $P_Y^*$-a.e. $y$, and $(b)$ $Q_{X|Y = y}^{\sim t} \ll Q_X^*$ for $P_Y^*$-a.e. $y$.  The latter property is the absolute-continuity condition needed to invoke the information-density definition in \eqref{conven}. Note also that the arbitrary definition of
$Q_{X|Y=y}^{\sim t}$ on the 
$P_Y^*$-null set does not affect the expectation in
\eqref{gccn}.
\item \textit{Independence from the density version:}
Assumption \ref{assumpr3_channels} implies
$P_Y^*\ll\mu$. If $q_\mu$ and $\bar q_\mu$ are two $\mu$-density versions of the same design
channel $Q_{Y|X}$, then
\[
q_\mu(y|x)=\bar q_\mu(y|x)
\quad \text{ for } \quad 
(Q_X^*\times\mu)\text{-a.e. }(x,y).
\]
Consequently, the normalizing integrals and the kernels
constructed from $q_\mu$ and $\bar q_\mu$ through
\eqref{QXYkernel} agree for $\mu$-a.e.\ $y$, and hence for
$P_Y^*$-a.e.\ $y$. Thus, the information density and its
expectation in \eqref{gccn} do not depend on the selected
$\mu$-density version of the design channel. Moreover, since $t>0$, Condition \eqref{thisimpliesit} implies
$P_Y^*\ll Q_Y^*$, where $Q_Y^*=Q_X^*Q_{Y|X}$.
\item \textit{What can fail without Assumption
\ref{assumpr3_channels}:}
Suppose that $\mu$ dominates the design channel
$Q_{Y|X=x}$ but does not dominate the true channel
$P_{Y|X=x}$. There may then exist a $\mu$-null set
$A\subseteq\mathcal Y$ such that $P_Y^*(A)>0$, while
$Q_Y^*(A)=0$. The values of the density version
$q_\mu(y|x)$ on $A$ can be modified arbitrarily without
changing $Q_{Y|X}$. Such modifications can
nevertheless change the normalizing integral in
\eqref{QXYkernel}, the resulting kernel
$Q_{X|Y=y}^{\sim t}$ for $y\in A$, and consequently the
expectation in \eqref{gccn}. Thus, without Assumption
\ref{assumpr3_channels}, Condition \eqref{thisimpliesit} alone
does not guarantee that
$C_t(P_{Y|X}\|Q_{Y|X})$ is a well-defined functional of the
channel pair.
\end{itemize}
Having established the preceding density version-invariance under
Assumption \ref{assumpr3_channels}, we henceforth suppress the
notation $q_\mu(y|x)$ and, whenever a pointwise likelihood is
intended, write $Q_{Y|X}(y|x)$ for the fixed $\mu$-density
$q_\mu(y|x)$. A few more discussion points about the mismatched-design rate definition in \eqref{gccn} are as follows:  
\begin{itemize}  
    \item Under Assumption \ref{assumpr3_channels}, a necessary condition\footnote{See \cite[(2.46)]{survey_mismatched.decoding} for a similar argument} for $C_{t}(P_{Y|X} \| Q_{Y|X}) > 0$ is
\begin{align}
    P_{Y|X}(\cdot|x) \ll Q_{Y|X}(\cdot|x) \text{ for } Q_X^*\text{-almost every } x. \label{nec_for_>zero}   
\end{align}  
\item We have $C_{t}(P_{Y|X} \| Q_{Y|X}) < \infty$. This follows from Assumption \ref{assumpr2_channels} and the fact that 
\begin{align}
        C_{t}(P_{Y|X} \| Q_{Y|X}) &= \left(I(Q_X^*, P_{Y|X}) - D(P_{X|Y} \| Q_{X|Y}^{\sim t} | P_Y^*)\right)^+, \label{532bb}
    \end{align}
where $P_{X|Y}$ is the conditional distribution associated with the joint $P_{X, Y} = Q_X^* \times P_{Y|X}$. 
\end{itemize}

\begin{remark}
For simplicity and clarity, we will occasionally use notation that is specialized for finite alphabets, e.g., summations instead of integrals and ratio of PMFs instead of Radon–Nikodym derivatives. For example, we may write $(\ref{gccn})$ as 
\begin{align}
    C_{t}(P_{Y|X} \| Q_{Y|X}) &= \left(\sum_{x \in \mathcal{X}, y \in \mathcal{Y}} Q_X^*(x) P_{Y|X}(y|x) \log \left(\frac{Q_{Y|X}(y|x)^t}{Q_Y^{\sim t}(y)} \right)\right)^+, \label{hpv}
\end{align}
where 
\begin{align*}
    Q_Y^{\sim t}(y) = \sum_{\bar{x} \in \mathcal{X}} Q_X^*(\bar{x}) Q_{Y|X}(y|\bar{x})^t.  
\end{align*}
Unless otherwise noted, all our results will hold for abstract probability spaces where the underlying measurable space is a standard Borel space. 
\label{simplicity_remark}
\end{remark}

Since $D(\cdot \| \cdot) \geq 0$, it follows from $(\ref{532bb})$ that 
\begin{align*}
    C_{t}(P_{Y|X} \| Q_{Y|X}) \leq I(Q_X^*, P_{Y|X}) \leq C(P_{Y|X}), 
\end{align*}
which aligns well with the intuition
that a channel code optimized for $Q_{Y|X}$ but deployed on $P_{Y|X}$ cannot exceed the true capacity (of $P_{Y|X}$). Writing the mismatched-design rate $(\ref{532bb})$ as  
\begin{align*}
    C_{t}(P_{Y|X} \| Q_{Y|X}) = \left(C(P_{Y|X}) - \left( C(P_{Y|X}) - I(Q_X^*, P_{Y|X}) \right) - D(P_{X|Y} \| Q_{X|Y}^{\sim t} | P_Y^*)\right)^+, 
\end{align*}
one can interpret the first penalty, $C(P_{Y|X}) - I(Q_X^*, P_{Y|X})$, as coming from the mismatched codebook distribution and the second penalty, $D(P_{X|Y} \| Q_{X|Y}^{\sim t} | P_Y^*)$, as coming from mismatched decoding. The quantity $C_t(P_{Y|X} \| Q_{Y|X})$ belongs to the $t$-parameterized family of achievable rates in channel coding with mismatched decoding, where the codebook distribution is fixed as $Q_X^*$. Specifically, define the generalized mutual information \cite[(2.10)]{survey_mismatched.decoding} as  
\begin{align}
    I_{\operatorname{GMI}}(Q_X^*, q) &\coloneqq \left( \sup_{t \geq 0} \mathbb{E}_{Q_X^* \times P_{Y|X}} \left [ \log \left( \frac{q(X, Y)^t }{\sum_{\bar{x}} Q_X^*(\bar{x}) q(\bar{x},Y)^t} \right) \right ]\right)^+, \label{b4vn}
\end{align}
where $q = q(x, y)$ is a nonnegative function called the decoding metric. Note that $(\cdot)^+$ in $(\ref{b4vn})$
is redundant since the expression inside the parenthesis is always nonnegative \cite[(2.9)]{survey_mismatched.decoding}. Hence,   
\begin{align*}
    C_t(P_{Y|X} \| Q_{Y|X}) &= \left(\mathbb{E}_{Q_X^* \times P_{Y|X}} \left [ \log \left( \frac{ Q_{Y|X}(Y|X)^t }{\sum_{\bar{x}} Q_X^*(\bar{x}) Q_{Y|X}(Y|\bar{x})^t} \right) \right ]\right)^+ 
\end{align*}
can be obtained from $(\ref{b4vn})$ by choosing the decoding metric $q(x, y) = Q_{Y|X}(y|x)$ and a fixed value of $t$. Hence, we have $C_t(P_{Y|X} \| Q_{Y|X}) \leq I_{\operatorname{GMI}}(Q_X^*,Q_{Y|X})$. When $t = 1$ and $P_{Y|X} = Q_{Y|X}$,
the mismatched-design rate and  mismatched-design rate-dispersion are equal to the capacity $C(P_{Y|X})$ and dispersion $V(P_{Y|X})$, respectively.

\section{Mismatched-design results in the memoryless block coding framework \label{qyxbasedsourcechannelcodesresults}}

In the remainder of the paper, we let the source $S$ be a sequence of length $k$ and the channel input $X$ be a sequence of length $n$. We denote a sequence of length $n$ as $X^n = (X_1, \ldots, X_n)$.
\begin{definition}
An $(n, k)$ source-channel code is a pair of (possibly randomized) mappings $\operatorname{f} : \mathcal{S}^k \to \mathcal{X}^n$ and $\operatorname{g}: \mathcal{Y}^n \to \mathcal{Z}^k$. 
\label{defjsccnk}
\end{definition}
Let $\operatorname{d}_k : \mathcal{S}^k \times \mathcal{Z}^k \to [0, \infty)$ denote the distortion measure defined for sequences. Then we have the following definition as the counterpart of Definition \ref{wocostdef} in the block coding setting: 
\begin{definition}
An $(n, k)$ source-channel code $(\operatorname{f}, \operatorname{g})$ is called a $(d, \epsilon, n, k)$ source-channel code with respect to $\{P_{S^k}, \operatorname{d}_k, P_{Y^n|X^n} \}$ if $\mathbb{P}(\operatorname{d}_k(S^k, \operatorname{g}(Y^n)) > d) \leq \epsilon$, where $X^n = \operatorname{f}(S^k)$, $S^k \sim P_{S^k}$ and $Y^n|X^n \sim P_{Y^n|X^n}$. 
\label{wocostdefnk}
\end{definition}

We assume that the source is stationary and memoryless and the channel is memoryless, but not necessarily stationary. We let the distortion measure $\operatorname{d}_k$ over sequences be additively separable. Effectively, we make the following substitutions in the one-shot setting of the previous sections:    
\begin{align}
\begin{split}
    \mathcal{S} &\to \mathcal{S}^k,\\
    \mathcal{Z} &\to \mathcal{Z}^k,\\
    \mathcal{X} &\to \mathcal{X}^n,\\
    \mathcal{Y} &\to \mathcal{Y}^n,\\
    P_S &\to P_{S^k} = P_S \times \cdots \times P_S,\\
    P_{Y|X} &\to P_{Y^n|X^n} = P^{(1)}_{Y|X} \times \cdots \times P^{(n)}_{Y|X},\\
    Q_{Y|X} &\to Q_{Y^n|X^n} = Q^{(1)}_{Y|X} \times \cdots \times Q^{(n)}_{Y|X},\\
    \operatorname{d}(S, Z) &\to \operatorname{d}_k(S^k, Z^k) = \frac{1}{k} \sum_{i=1}^k \operatorname{d}(S_i, Z_i),\\
    B_d(s) &\to B_d(s^k) = \left \{z^k \in \mathcal{Z}^k : \operatorname{d}_k(s^k, z^k) \leq d \right\}.
\end{split}
\label{rep1}
\end{align}

We now present mismatched-design results in the block coding setting. By the memoryless property of both the source and the channels and the additive separability of the distortion measure, the substitutions in $(\ref{rep1})$ imply the following replacements in the $Q_{Y|X}$-based source-channel code construction in Definition \ref{qyxbasedcodedef}: 
\begin{align}
\begin{split}
    P_Z^* &\to P_{Z^k}^* = (P_{Z}^*)^{\otimes k},\\
    Q_X^* &\to Q^*_{X^n} = \bigotimes_{i=1}^n Q_{X_i}^*, \\
     Q_{X|Y}^{\sim t} &\to Q_{X^n|Y^n}^{\sim t} = \bigotimes_{i=1}^n Q_{X_i|Y_i}^{(i),\sim t},
\end{split}
\label{rep2}
\end{align}
where 
\begin{align}
\begin{split}
    P_{Z|S}^* &= \argmin_{\substack{P_{Z|S}:\\\mathbb{E}[\operatorname{d}(S, Z)] \leq d}} I(P_S, P_{Z|S}),\\
    P_{Z}^* &= P_S P_{Z|S}^*,\\
    Q_{X_i}^* &= \argmax_{Q} I(Q, Q_{Y|X}^{(i)}),\\
    Q_{X|Y}^{(i),\sim t}(x_i|y_i) &= \frac{Q_{X_i}^*(x_i) Q_{Y|X}^{(i)}(y_i|x_i)^t}{Q_{Y_i}^{\sim t}(y_i)},\\
    Q_{Y_i}^{\sim t}(y_i) &= \int_\mathcal{X} Q^{(i)}_{Y|X}(y_i|u)^t Q_{X_i}^*(du).
\end{split} 
\label{rep3}
\end{align}

Note that we continue to enforce 
the Assumptions \ref{assumpr1_source}\,-\,\ref{assumpr3_source} for $P_S, d$ and $\operatorname{d}$, and Assumptions \ref{assumpr1_channels}\,-\,\ref{assumpr3_channels} for $Q_{Y|X}^{(i)}$ and $ P_{Y|X}^{(i)}$ for each $i$; in particular, $Q_{X_i}^*$ and $P_{Z|S}^*$ in $(\ref{rep3})$ are assumed to be unique. With the above extensions to the block coding setting, the $(t, Q_{Y|X})$-based source-channel code defined in Definition \ref{qyxbasedcodedef} is now called a $(t, Q_{Y^n|X^n}, n, k)$-based source-channel code.
\begin{definition}
    We define a $(t, Q_{Y^n|X^n}, n, k)$-based source-channel code as the $(P_{Z^k}^*, Q_{X^n}^*, Q^{\sim t}_{X^n|Y^n})$ source-channel code from Definition \ref{unified_def}. \label{Qbasedwithoutcost}
\end{definition}

We then have the following straightforward corollary to Corollary \ref{bc}.

\begin{corollary}
The ensemble-average excess-distortion probability $\overline{P}_e(d, n, k)$ with respect to $\{P_{S^k}, \operatorname{d}_k, P_{Y^n|X^n} \}$ of the \\
$(t, Q_{Y^n|X^n}, n, k)$-based source-channel code satisfies  
\begin{align}
         \overline{P}_e(d, n, k) \leq \mathbb{E}\left [ \left(1 + P_{Z^k}^*(B_d(S^k)) \exp\left( \imath_{Q_{X^n|Y^n}^{\sim t}, Q_{X^n}^*}(Y^n, X^n) \right) \right)^{-1}  \right], \label{g2n}
    \end{align}
    where the expectation above is w.r.t. $P_{S^k}, Q_{X^n}^*$ and $P_{Y^n|X^n}$, and
    \begin{align*}
         \imath_{Q_{X^n|Y^n}^{\sim t}, Q_{X^n}^*}(Y^n, X^n) = \sum_{i=1}^n \imath_{Q_{X|Y}^{(i), \sim t}, Q_{X_i}^*}(Y_i, X_i). 
    \end{align*}
    \label{cor378ryhre}
\end{corollary}

Next, to obtain refined asymptotic results via a normal approximation, we enforce two conditions as described next. Define 
$$d_{\max} = \inf_{z \in \mathcal{Z}} \mathbb{E}_{P_S}\left [\operatorname{d}(S, z) \right],$$ 
which is the lowest value of $d$ such that $R_d(P_S) = 0$.

\textbf{Conditions:}
\begin{enumerate}[label=C\arabic*., ref=C\arabic*]   
\item \label{cond1dee} The distortion level $d \in (d_{\min}, d_{\max})$ so that $0 < R_d(P_S) < \infty$.
\item \label{weird9thmomcond} $\mathbb{E}\left [(\operatorname{d}(S, Z))^9 \right] < \infty$, where the expectation is w.r.t. $S \sim P_S$ and $Z \sim P_Z^*$ with $S \ind Z$.
\end{enumerate}
Under these conditions, the following lemma from \cite{kostina_SC} provides a useful bound for $P_{Z^k}^*(B_d(S^k))$, enabling us to further upper bound the excess-distortion probability in $(\ref{g2n})$. Note that Condition \ref{weird9thmomcond} is only used for the purpose of applying the following lemma.

\begin{lemma}[{\cite[Lemma 2]{kostina_SC}}] Let the source distribution $P_S$, distortion level $d$ and distortion measure $\operatorname{d}$ satisfy Assumptions 
\ref{assumpr1_source}\,-\,\ref{assumpr3_source} and Conditions \ref{cond1dee}\,-\,\ref{weird9thmomcond}. Let the source $S^k$ be i.i.d. $P_S$. Let $P_{Z^k}^*$ be as defined in $(\ref{rep2})$. Then there exist positive constants $k_0, c_1$ and $K$ such that for all $k \geq k_0$, 
\begin{align*}
        \mathbb{P}\left( \log \frac{1}{P_{Z^k}^*(B_d(S^k))} \leq \sum_{i=1}^k \jmath_d(S_i, P_S, \operatorname{d}) + \left( \overline{c} - \frac{1}{2} \right) \log k + c_1  \right) \geq 1 - \frac{K}{\sqrt{k}},
\end{align*}
where 
\begin{align}
\overline{c} &= 1 + \frac{\text{Var}_{P_S}\left( \Lambda'(S, \lambda^*) \right)}{\mathbb{E}_{P_S}\left[ | \Lambda''(S, \lambda^*) | \right]  }, \label{c28f}\\
        \Lambda(s, \lambda) &= \log \left(\frac{1}{ \mathbb{E}_{P_Z^*}\left [ \exp \left( \lambda \left( d -  \operatorname{d}(s, Z)\right) \right) \right ]}\right),
    \end{align}
$\lambda^* = -\partial_d R_d(P_S)$, and $(\cdot)'$ denotes differentiation with respect to $\lambda$.
\label{lemma5}
\end{lemma}

Examples of source distributions $P_S$, value ranges for $d$, and distortion measures $\operatorname{d}$ satisfying Assumptions 
\ref{assumpr1_source}\,-\,\ref{assumpr3_source} and Conditions \ref{cond1dee}\,-\,\ref{weird9thmomcond} are given in Appendix \ref{examples_of_sources}. Before stating Theorem \ref{thm_finite_block1}, we need one more definition. 

\begin{definition}
Let $\operatorname{S} \subset \mathbb{N}^2$ and $\operatorname{S}_n \coloneqq \{k:(n, k) \in \operatorname{S} \}$. We say that $\operatorname{S}$ is an admissible source-channel blocklength set if 
    \begin{enumerate}
    \item $\operatorname{S}_n \neq \varnothing$ for all sufficiently large $n$, and 
    \item there exist constants $0 < \beta_- \leq  \beta_+ < \infty$ such that for all $(n, k) \in \operatorname{S}$, we have  $\beta_- \leq \frac{k}{n} \leq \beta_+$.
\end{enumerate}
A family of $(n, k)$ source-channel codes
\begin{align}
     \left \{ \mathcal{C}_{n, k} \right \}_{(n,k) \in \operatorname{S}}  \label{C_A}  
\end{align}
is called $\operatorname{S}$-admissible. We will sometimes write $(\ref{C_A})$ as $\{\mathcal{C}_{n, k} \}_{\operatorname{S}}$.  
\label{admissible_def}
\end{definition}

Theorem \ref{thm_finite_block1} is a result about an $\operatorname{S}$-admissible family of $(n, k)$ source-channel codes for an arbitrary admissible set $\operatorname{S}$, where each code in the family is a $(t, Q_{Y^n|X^n}, n, k)$-based source-channel code from Definition \ref{Qbasedwithoutcost}. Given such a family, then under certain conditions and for sufficiently large $n$ and $k$, Theorem \ref{thm_finite_block1} gives a sufficient condition for a code in the family to have excess-distortion probability at most $\epsilon$ over the true channel.

\begin{theorem}[Achievability for an admissible family of source-channel codes]
Fix any $\epsilon \in (0, 1)$, $t > 0$, and an admissible source-channel blocklength set $\operatorname{S}$. Let the source distribution $P_S$, distortion level $d$ and distortion measure $\operatorname{d}$ satisfy Assumptions 
\ref{assumpr1_source}\,-\,\ref{assumpr3_source} and Conditions \ref{cond1dee}\,-\,\ref{weird9thmomcond}. Let $\left(P_{Y|X}^{(i)}, Q_{Y|X}^{(i)} \right)_{i \geq 1}$ be a sequence of channel pairs such that $P_{Y|X}^{(i)}$ and $Q_{Y|X}^{(i)}$ satisfy Assumptions \ref{assumpr1_channels}\,-\,\ref{assumpr3_channels} and $C_t(P_{Y|X}^{(i)} \| Q_{Y|X}^{(i)})$ is well-defined (Definition \ref{def_mismatched_rate}) and positive for each $i$. Let $\left \{\mathcal{C}_{n, k}\right\}_{\operatorname{S}}$ denote the $\operatorname{S}$-admissible family of source-channel codes, where $\mathcal{C}_{n, k}$ is the $(t, Q_{Y^n|X^n}, n, k)$-based source-channel code from Definition \ref{Qbasedwithoutcost}. Assume that there exist constants $0 < \kappa_1, \kappa_2 < \infty$ such that for all $(n, k) \in \operatorname{S}$,
    \begin{align}
        &\frac{1}{n + k} \left [ k\, \mathbb{E}_{P_S}\left[ \big | \jmath_d(S, P_S, \operatorname{d}) - R_d(P_S) \big |^3 \right] +  \sum_{i=1}^n \mathbb{E}_{Q_{X_i}^* \times P_{Y|X}^{(i)}}\left [  \Big | \imath_{Q_{X|Y}^{(i), \sim t}, Q_{X_i}^*}\left(Y_i, X_i \right) - C_t(P_{Y|X}^{(i)} \| Q_{Y|X}^{(i)}) \Big|^3 \right] \right] \leq \kappa_1, \label{assumpbe1}\\
        &\frac{1}{n + k}\left [ k V_d(P_S) +  \sum_{i=1}^n V_t(P_{Y|X}^{(i)} \| Q_{Y|X}^{(i)})  \right]  \geq \kappa_2. \label{assumpbe2}
    \end{align} 
Then there exist constants $c$, $K_0$ and $N_0$ such that for every $(n, k) \in \operatorname{S}$ with $n \geq N_0$ and $k \geq K_0$, the $(t, Q_{Y^n|X^n}, n, k)$-based source-channel code is $(d, \epsilon, n, k)$  with respect to $\{P_{S^k}, \operatorname{d}_k, P_{Y^n|X^n} \}$ if
    \begin{align}
        \sum_{i=1}^n C_t(P_{Y|X}^{(i)} \| Q_{Y|X}^{(i)}) - k R_d(P_S) &\geq  \sqrt{k V_d(P_S) + \sum_{i=1}^n V_t(P_{Y|X}^{(i)} \| Q_{Y|X}^{(i)})} \,\,Q^{-1}(\epsilon) + \left(\overline{c} - \frac{1}{2} \right) \log (k) + c,    \label{thm_cond}
    \end{align}
    where $\overline{c}$ is\footnote{We calculate closed-form expressions for $\overline{c}$ for specific source distributions and distortion measures in Appendix \ref{examples_of_sources}.} defined in $(\ref{c28f})$. 
    \label{thm_finite_block1}
\end{theorem}
\textit{Proof:} The proof of Theorem \ref{thm_finite_block1} is given in Appendix \ref{thm_finite_block1_proof}. 

\subsection{Discussion of Theorem \ref{thm_finite_block1}}

First, note that if $P_{Y|X}^{(i)} = Q_{Y|X}^{(i)}$ for all $i$, then Theorem \ref{thm_finite_block1} is simply an achievability result for traditional matched JSCC with known channel laws. In this case, Theorem \ref{thm_finite_block1} with $t = 1$ extends the achievability half of \cite[Theorem 10]{kostina_JSCC} to nonstationary channels. If, in addition, $P_{Y|X}^{(i)} = Q_{Y|X}^{(i)} = P_{Y|X}$ for all $i$, then Theorem \ref{thm_finite_block1} recovers (for most channels) the achievability half of \cite[Theorem 10]{kostina_JSCC} with an improved third-order term since \cite[Theorem 10]{kostina_JSCC} contains an extra $\frac{1}{2} \log (k) + \log \log k$ term in the RHS of $(\ref{thm_cond})$. The $\log \log k$ improvement is already implicit in \cite[Appendix D]{alternative_oneshot_JSCC}. Indeed, \cite[Theorem 4]{alternative_oneshot_JSCC} is an improved one-shot achievability bound for JSCC than \cite[Theorem 8]{kostina_JSCC}, and our Proposition \ref{bxc} is a generalization of \cite[Theorem 4]{alternative_oneshot_JSCC} to the mismatched setting. However, the more substantial $\frac{1}{2} \log(k)$ improvement in Theorem \ref{thm_finite_block1} is unique to our analysis. The key step in our proof yielding this improvement is the rewriting of the RHS of $(\ref{g2n})$ using the identity 
\begin{align}
    \mathbb{E}\left[ \frac{1}{1 + \mathbb{P}(S) \exp\left( \mathcal{I}_C\right)} \right ] = \operatorname{Pr}\left( \mathcal{I}_C + L \leq \log \left(\frac{1}{\mathbb{P}(S)} \right) \right), \label{log_identity}
\end{align}
where we abbreviated the source distortion ball probability as $\mathbb{P}(S)$, the information density as $\mathcal{I}_C$ and introduced an independent standard logistic random variable $L$. This identity allows the source random variable and the channel information density to be analyzed directly after adding only a bounded-scale logistic perturbation $L$.

\begin{remark}
Under the admissibility condition and the uniform third-moment condition $(\ref{assumpbe1})$, the total dispersion term
\[
    B_{n,k}
    \triangleq
    kV_d(P_S)+
    \sum_{i=1}^n
    V_t(P_{Y|X}^{(i)}\|Q_{Y|X}^{(i)})
\]
is of order \(n\). The lower bound \(B_{n,k}=\Omega(n)\) follows from $(\ref{assumpbe2})$, while the upper bound \(B_{n,k}=O(n)\) is shown in $(\ref{6n})$.
\label{total_dispersion_term_is_linear}
\end{remark}

\begin{remark}[Admissibility and the average mismatched-design rate] Definition \ref{admissible_def} formalizes
the JSCC asymptotic regime in which the source and channel blocklengths grow together. Admissibility of a family of source-channel codes is a condition only on the allowed source-channel blocklength pairs $(n, k)$, and it is imposed
independently of the particular channel sequence $\left (P_{Y|X}^{(i)}, Q_{Y|X}^{(i)} \right)_{i \geq 1}$ in the theorem. Define the average mismatched-design rate 
\begin{align}
    \overline C_n
    \triangleq
    \frac{1}{n}\sum_{i=1}^n
    C_t(P_{Y|X}^{(i)}\|Q_{Y|X}^{(i)}).
\end{align}
Although Theorem \ref{thm_finite_block1} does not impose any condition on the asymptotic behavior of $\overline{C}_n$, for an admissible family of source-channel codes, the theorem is most informative when $\overline{C}_n = \Theta(1)$. If $\overline{C}_n = o(1)$, then the condition $(\ref{thm_cond})$ is impossible to satisfy for all sufficiently large $n$, since the LHS of $(\ref{thm_cond})$ is negative of order $n$ due to the $k R_d(P_S)$ term, whereas the RHS is $O(\sqrt{n})$ with sign depending on $\epsilon$. On the other hand, if $\overline{C}_n = \omega(1)$, then the sufficient reliability condition $(\ref{thm_cond})$ remains valid but it is no longer sharp in the asymptotic sense. 
\label{secremarkafter6}
\end{remark}

Further remarks on Theorem \ref{thm_finite_block1} are as follows:
\begin{enumerate}
\item For an admissible source-channel blocklength set $\operatorname{S}$ from Definition \ref{admissible_def}, the condition $(\ref{assumpbe2})$ is automatically satisfied for all $(n, k) \in \operatorname{S}$ if either $V_d(P_S) > 0$ or a $\delta_1 \in (0, 1)$ fraction of channel pairs $\{ (P_{Y|X}^{(i)}, Q_{Y|X}^{(i)})\}_{i=1}^n$ satisfy $V_t(P_{Y|X}^{(i)} \| Q_{Y|X}^{(i)}) > \delta_2$ for some fixed positive $\delta_1$ and $\delta_2$ independent of $(n, k)$. 
\item \label{iskomw0455} Given that Condition \ref{weird9thmomcond} holds, we automatically have that $\mathbb{E}_{P_S}\left[ \big | \jmath_d(S, P_S, \operatorname{d}) - R_d(P_S) \big |^3 \right] < \infty$ \cite[(339)]{kostina_JSCC}, so the additional assumption $(\ref{assumpbe1})$ only requires that the third absolute moments of the information density under the mismatched joint law do not grow faster than linearly in aggregate. 
\end{enumerate}

To place the result in Theorem \ref{thm_finite_block1} in a broader context, consider a $t$-parameterized family of achievable rates in channel coding with mismatched channels, where the supremum over all $t > 0$ leads to the generalized mutual information \cite[2.10]{survey_mismatched.decoding}:
\begin{align}
     \sup_{t \geq 0} \mathbb{E}_{Q_X^* \times P_{Y|X}} \left [ \log \left( \frac{Q(Y|X)^t }{\sum_{\bar{x}} Q_X^*(\bar{x}) Q(Y|\bar{x})^t} \right) \right ]. \label{gmitachieve}
\end{align}
The first-order channel-dependent term in our achievability result $(\ref{thm_cond})$, namely the mismatched-design rate $C_t(P_{Y|X} \| Q_{Y|X})$, matches the $t$-achievable rate in $(\ref{gmitachieve})$. Since $t > 0$ is an arbitrary parameter in Theorem \ref{thm_finite_block1}, the GMI in $(\ref{gmitachieve})$ is achieved as the envelope over the $t$-parameterized family of $(t, Q_{Y^n|X^n}, n, k)$-based designs as $n, k \to \infty$. Note that the choice of $t$ only affects the Gibbs posterior used as the decoder-side kernel.         

We also remark that there is a notable difference between the achievability of $(\ref{gmitachieve})$ by the $(t, Q_{Y|X})$-based JSCC that we show in Theorem \ref{thm_finite_block1}  and by a channel code with a fixed $Q_{Y|X}$-based maximum likelihood (ML) decoder that can be shown using the random coding union bound, e.g., in \cite[2.6.4]{survey_mismatched.decoding}. In the latter case, $(\ref{gmitachieve})$ is achieved by a single $Q_{Y|X}$-based channel coding scheme in which the parameter $t$ only exists in the analysis and is not an actual parameter of the code construction. We call this a \emph{single-scheme achievability} of the GMI, and it follows from the fact that the $Q_{Y|X}$-based ML decoding rule $\argmax_{j} Q_{Y|X}(y|x_j)$ is mathematically equivalent to $\argmax_{j} Q_{Y|X}(y|x_j)^t$ for any $t > 0$. In contrast, a $Q_{Y|X}$-based joint source-channel code does not straightforwardly admit the aforementioned mathematical equivalence. This is because unlike channel coding with equiprobable messages, a JSCC decoder must generally account for a nonuniform source-induced prior over the candidate indices—or, equivalently, exploit the residual encoded-source redundancy as discussed in \cite[p. 2552]{kostina_JSCC}. In our Poisson construction,
this source-side weighting appears through the Poisson-priority term $- \log(T_i)$, as can be seen by writing down the decoding rule of our $Q_{Y|X}$-based code which is   
\begin{align}
    i^*(y) &= \argmax_i \left \{ t \log Q_{Y|X}(y|\bar{X}_i)  -\log(T_i)\right \}. \label{differentchannelcodingJSCC} 
\end{align}
The GMI parameter $t$ only changes the channel term while leaving the Poisson-priority term fixed. Consequently, different values of $t$ generally induce genuinely different JSCC decoding rules, rather than merely different analyses of a fixed decoding rule. Hence, $\eqref{gmitachieve}$ should not be interpreted as the first-order channel-dependent achievable rate of a single $Q_{Y|X}$-based JSCC design. Rather, it is a rate envelope over the $t$-parameterized family of $(t,Q_{Y|X})$-based JSCC designs. We call this \emph{scheme-envelope achievability}.  Note that such a distinction also has precedent in channel coding with mismatched decoding. For example, the achievability of the LM rate shown in \cite[Theorem 5]{scarlett_mismatched_decoding} via a cost-constrained random-coding ensemble is a scheme-envelope achievability because the parameters $t$ and the auxiliary cost functions $a(\cdot)$ are code-design parameters; specifically, they influence the random codebook ensemble. In contrast, for the case of finite alphabets, it is possible to prove a single-scheme achievability of the LM rate via constant-composition random coding (e.g., \cite[2.6.5]{survey_mismatched.decoding}) using the metric-equivalence result in \cite[Proposition 2.2]{survey_mismatched.decoding}: on a constant-composition codebook, the decoding metrics $q(x,y)$ and $q(x,y)^t e^{a(x)}$ induce the same maximum-metric decoding rule for every $t>0$.

\textit{Post-design receiver-only optimization paradigm:} For each $t > 0$ satisfying the given assumptions, Theorem \ref{thm_finite_block1} gives a sufficient reliability condition $(\ref{thm_cond})$ for the
$(t, Q_{Y^n|X^n}, n, k)$-based source-channel code over the true channel $P_{Y^n|X^n}$. This gives a useful one-dimensional set of achievability bounds, parameterized by $t$, while the design channel $Q_{Y^n|X^n}$ remains fixed. In general, the design channel $Q_{Y^n|X^n}$ may be a complicated object, and redesigning the whole JSCC construction for a newly identified true channel $P_{Y^n|X^n}$ may be substantially more demanding than keeping the $Q_{Y^n|X^n}$-based source-channel architecture fixed and tuning only the scalar parameter $t$. Thus, if some post-design knowledge or estimate of $P_{Y^n|X^n}$ is available, for example, during runtime, the bound $(\ref{thm_cond})$ can be optimized over the one-dimensional family 
$$\left \{ \left ( C_t(P_{Y|X}^{(i)}\|Q_{Y|X}^{(i)}), V_t(P_{Y|X}^{(i)}\|Q_{Y|X}^{(i)}) \right)_{1 \leq i \leq n} : t>0\right \}.$$
This is in contrast to the scheme-envelope achievability of the LM rate in \cite[Theorem 5]{scarlett_mismatched_decoding}, where a post-design optimization involves both the parameters $t$ and the auxiliary cost functions $a(\cdot)$ that influence the random codebook distribution. In contrast, a post-design optimization of our scheme-envelope achievability of the GMI is simpler because only a scalar optimization is needed and the parameter $t$ only affects the decoder design, while both the encoder and the codebook ensemble (i.e., the shared randomness) remain unchanged. This is useful in scenarios where the true channel knowledge is limited to the receiver only. Such situations are common because the receiver can estimate the instantaneous channel directly from pilots or known signal components, whereas the transmitter must rely on feedback, channel reciprocity, or statistical channel knowledge. Lastly, even though the GMI can be below the LM rate, equality is achieved for many channels such as block erasure and Gaussian channels.    

In the matched setting, i.e., $Q_{Y^n|X^n} = P_{Y^n|X^n}$, setting $t = 1$ in our results recovers the optimal dispersion as well as the optimal third-order term up to a multiplicative constant factor \cite[Theorem 10]{kostina_JSCC}. In other words, the mismatched-design rate and the mismatched-design rate-dispersion reduce to the channel capacity and dispersion, respectively, when $Q_{Y^n|X^n} = P_{Y^n|X^n}$ and $t = 1$.

\section{Results on universal source-channel codes \label{3i8rt48}}

In the first part of this section, we give the definition of universal source-channel codes, where the universality is only w.r.t. a given set of channels. In the second and third parts, we present universal source-channel code constructions for block erasure channels and Gaussian channels. 

\subsection{Definition of Second-Order Universality \label{universalpluginconstruction}}

In this subsection, we rigorously define the notion of second-order universality of joint source-channel codes over the set of memoryless, nonstationary channels. Let $\mathcal{W}$ be any set of single-use channels $P_{Y|X} : \mathcal{X} \to \mathcal{P}(\mathcal{Y})$. Let 
\begin{align*}
    \mathscr{C}(\mathcal{W}) \coloneqq \left \{ \left(P_{Y|X}^{(i)} \right)_{i \geq 1} :  P_{Y|X}^{(i)} \in \mathcal{W}   \right \}
\end{align*}
denote the set of all memoryless, nonstationary channels with channel laws in $\mathcal{W}$. We will sometimes write $\mathbf{P} \coloneqq \left(P_{Y|X}^{(i)} \right)_{i \geq 1}$ as shorthand. Given any $\mathbf{P} \in \mathscr{C}(\mathcal{W})$, denote its $n$-block prefix channel by $P_{Y^n|X^n}^{\mathbf{P}}$ where  
\begin{align*}
    P_{Y^n|X^n}^{\mathbf{P}} \coloneqq \bigotimes_{i=1}^n P_{Y|X}^{(i)}. 
\end{align*}
Then let  
\begin{align}
    \mathcal{W}_n &\coloneqq \left \{ \bigotimes_{i=1}^n P_{Y|X}^{(i)}:P_{Y|X}^{(i)} \in \mathcal{W} \right \} \label{hellownnonstat}
\end{align}
denote the set of all $n$-block prefix channels, or equivalently, the set of $n$-block memoryless, nonstationary channels with channel laws in $\mathcal{W}$. We also write $P_n \coloneqq P_{Y^n|X^n}^{\mathbf{P}}$ as shorthand.  

To simplify the notation, assume that the source distribution $P_{S^k} = P_S^{\otimes k}$, distortion level $d$ and distortion measure $\operatorname{d}_k$ have been fixed as in $(\ref{rep1})$ such that $P_S$, $d$ and $\operatorname{d}$ satisfy Assumptions 
\ref{assumpr1_source}\,-\,\ref{assumpr3_source} and Conditions \ref{cond1dee}\,-\,\ref{weird9thmomcond}. Let $U_{n, k}^{(d)}$ denote an arbitrary $(n, k)$ source-channel code, where the dependence of $U_{n, k}^{(d)}$ on $P_S$ and $\operatorname{d}$ is implicit.
\begin{definition}
A source-channel code $U_{n, k}^{(d)}$ is channel-blind over $\mathcal{W}_n$ if it does not depend on any particular channel $P_n \in \mathcal{W}_n$, although it may depend on $\mathcal{W}_n$ itself. Furthermore, a family of source-channel codes 
$$\mathcal{U}^{(d)} \coloneqq \{ U_{n, k}^{(d)} : (n, k) \in \mathbb{N}^2\}$$ 
is \emph{channel-blind} over $(\mathcal{W}_n)_{n \geq 1}$ if 
every $U_{n, k}^{(d)} \in \mathcal{U}^{(d)}$ is channel-blind over $\mathcal{W}_n$. 
\end{definition}

Given a channel-blind family $\mathcal{U}^{(d)}$ of source-channel codes over $(\mathcal{W}_n)_{n \geq 1}$, define, for every $P_n \in \mathcal{W}_n$ and $\epsilon \in (0, 1)$,
\begin{align}
    k^*_{\mathcal{U}^{(d)}}(d, \epsilon, n|P_n) &\coloneqq \sup \, \Big \{ k \in \mathbb{N}: 
        U_{n, k}^{(d)} \text{ is a } 
        (d, \epsilon, n, k) \text{ source-channel code w.r.t. }  \{P_{S^{k}}, \operatorname{d}_{k}, P_n \} \Big \},
\label{channelblind_achievability_region}
\end{align}
with the convention that $\sup \varnothing = 0$. The optimal benchmark is defined as 
\begin{align}
    k^*\left(d, \epsilon, n|P_n\right) &\coloneqq \sup \left \{ k \in \mathbb{N}: \exists (d, \epsilon, n, k) \text{ source-channel code w.r.t. } \{P_{S^k}, \operatorname{d}_k, P_n \}  \right \}. \label{unrestricted_benchmark}
\end{align}

\begin{definition}
 Let $\mathscr{C}_{\text{reg}}(\mathcal{W}) \subset \mathscr{C}(\mathcal{W})$. Then a channel-blind family $\mathcal{U}^{(d)} = \{ U_{n, k}^{(d)}\}$ of source-channel codes over $(\mathcal{W}_n)_{n \geq 1}$ is called second-order universal 
    over $\mathscr{C}_{\text{reg}}(\mathcal{W})$ if for every $\mathbf{P} \in \mathscr{C}_{\text{reg}}(\mathcal{W})$ and every $\epsilon \in (0, 1)$, 
    \begin{align}
        \limsup_{n \to \infty} \left(\frac{k^*\left(d, \epsilon, n|P_{n}\right) - k^*_{\mathcal{U}^{(d)}}(d, \epsilon, n|P_{ n})}{\sqrt{n}} \right) = 0. \label{hereheregap}
    \end{align} 
    \label{2ndorderuniversaldef}
\end{definition}

In practice, universality is typically proven for channel sequences satisfying certain regularity assumptions, which motivates the introduction of a subset $\mathscr{C}_{\text{reg}}(\mathcal{W}) \subset \mathscr{C}(\mathcal{W})$. When proving universality only over the set of memoryless, stationary channels, the definition of $\mathcal{W}_n$ in $(\ref{hellownnonstat})$ simplifies to $\mathcal{W}_n = \left \{ P_{Y|X}^{\otimes n} : P_{Y|X} \in \mathcal{W}  \right \}$, and Definition \ref{2ndorderuniversaldef} simplifies to the following: 
\begin{definition}[Stationary channels]
A channel-blind family $\mathcal{U}^{(d)} = \{ U_{n, k}^{(d)}\}$ of source-channel codes over $(\mathcal{W}_n)_{n \geq 1}$ is called second-order universal 
    over $\mathcal{W}$ if for every $P_{Y|X} \in \mathcal{W}$ and every $\epsilon \in (0, 1)$, 
    \begin{align}
        \limsup_{n \to \infty} \left(\frac{k^*\left(d, \epsilon, n|P_{Y|X}^{\otimes n}\right) - k^*_{\mathcal{U}^{(d)}}\left(d, \epsilon, n|P_{Y|X}^{\otimes n} \right)}{\sqrt{n}} \right) = 0. \label{herestatheregap}
    \end{align} 
\end{definition}

\begin{remark}
We note that universality is a property of a channel-blind family of source-channel codes, and it is an asymptotic notion. A single code $U_{n, k}^{(d)}$ in the family is not what one calls universal. Rather, a universality result in the above framework asserts that the same channel-blind
family contains codes whose performance is asymptotically near-optimal for every
channel in the class. This does not, however, remove the need to select an appropriate operating point $(n, k)$ for the channel at hand. In our second-order universal construction (next section), the channel-blind family of codes is generated by a common rule
whose operating point $(n,k)$ can be adjusted. That is, we provide a single code-generation rule over a set of channels that is reused across all $(n,k)$ to obtain the channel-blind family of codes. This is in contrast to a
purely existence result in which the codes $U_{n,k}^{(d)}$ in the family could be unrelated across different blocklength pairs. In this sense, our universality result can be viewed as a form of construction-level separation: the rule used to generate the source-channel code is fixed uniformly over the set of channels, while channel adaptation is reduced to selecting the source blocklength $k$ for a given channel blocklength $n$. 
\label{remarkeeuseful22prop}
\end{remark}

\begin{remark} \label{remarkeeusefulprop}
Although a single code $U_{n, k}^{(d)}$ is not universal or optimal for every channel in $\mathcal{W}$, there is one useful per-code property of our channel-blind construction that is worth mentioning. First, recall that every code $U_{n, k}^{(d)}$ is independent of the target excess-distortion probability $\epsilon$. Then for every $\epsilon$ and for every channel sequence $\mathbf{P} \in \mathscr{C}_{\text{reg}}(\mathcal{W})$, we later show that the individual channel-blind code $U^{(d)}_{n,k}$ in our construction is $(d,\epsilon,n,k)$ w.r.t. $P_n$ under a sufficient condition that is asymptotically tight to second-order, i.e., it matches the corresponding  necessary condition in our converse results through the first two asymptotic terms. As we show later, this is true for both our block erasure and Gaussian channel constructions.
\end{remark}

Lastly, the definition of universality in Definition \ref{2ndorderuniversaldef} requires pointwise convergence to the optimal limit for each fixed sequence $\mathbf{P}$ as opposed to uniform convergence over the set $\mathscr{C}_{\text{reg}}(\mathcal{W})$:
\begin{align}
      \limsup_{n \to \infty} \, \sup_{\mathbf{P} \in \mathscr{C}_{\text{reg}}(\mathcal{W})}\left(\frac{k^*\left(d, \epsilon, n|P_{n}\right) - k^*_{\mathcal{U}^{(d)}}(d, \epsilon, n|P_{ n})}{\sqrt{n}} \right) = 0.
\end{align}
This parallels the distinction in source coding: weak universality requires convergence for each fixed source law, whereas strong universality, or minimax universality, requires uniform convergence over the source class.
Such a distinction is not merely technical; it leads to different convergence rates \cite{minimaxRD}. In this paper, we only enforce a pointwise, weak-universality requirement over channel sequences $\mathbf P\in\mathscr{C}_{\mathrm{reg}}(\mathcal W)$.

\subsection{Universal Source-Channel Code Construction for Block Erasure Channels \label{tiredaf}}

Throughout this subsection, let $\mathcal{W}$ be the set of block erasure channels with input alphabet $\mathcal{X} = \{0,1 \}^m$, output alphabet $\mathcal{Y} = \mathcal{X} \cup \{ e\}$, and erasure probability in $[0, 1)$. We denote $P_{Y|X} \in \mathcal{W}$ as $P_{Y|X} = \operatorname{BLEC}(2^m, p)$. Specifically,
\begin{align}
    \mathcal{W} &= \left \{ \operatorname{BLEC}(2^m, p) : p \in [0, 1) \right \},\\    \mathcal{W}_n &= \left \{ \bigotimes_{i=1}^n P^{(i)}_{Y|X} : (p_1, \ldots, p_n) \in [0, 1)^n\right \}, \label{v38}
\end{align}
where we write $P_{Y|X}^{(i)} = \operatorname{BLEC}(2^m, p_i)$. Also note that 
\begin{align*}
    C(P_{Y|X}) &= (1-p) \log(2^m),\\
    V(P_{Y|X}) &= p(1-p) (\log (2^m))^2.
\end{align*}
\begin{proposition}
Fix any $t > 0$ and a sequence $(q_i)_{i \geq 1}$, where each  $q_i \in (0, 1)$. Let $Q_{Y^n|X^n} = Q^{(1)}_{Y|X} \times \cdots \times Q^{(n)}_{Y|X}$, where $Q^{(i)}_{Y|X} = \operatorname{BLEC}(2^m, q_i)$. Let $U_{n, k}^{(d)}$ be the $(t, Q_{Y^n|X^n}, n, k)$-based source-channel code from Definition \ref{Qbasedwithoutcost}. Then the family $\{U_{n, k}^{(d)} :(n, k) \in \mathbb{N}^2 \}$ is channel-blind over $(\mathcal{W}_{n})_{n \geq 1}$, where $\mathcal{W}_n$ is given in $(\ref{v38})$.    
    \label{becuniversaltheorem}
\end{proposition}
\begin{IEEEproof}
We first show that the code construction is independent of $t$ and $Q_{Y^n|X^n}$, which means  that the $(t, Q_{Y^n|X^n}, n, k)$-based source-channel codes have the same encoder and decoder mappings for all values of $t > 0$ and $(q_1, \ldots, q_n) \in (0, 1)^n$. Recall from Definition \ref{Qbasedwithoutcost} that the $(t, Q_{Y^n|X^n}, n, k)$-based source-channel code is simply the $(P_{Z^k}^*, Q_{X^n}^*, Q^{\sim t}_{X^n|Y^n})$ source-channel code from Definition \ref{unified_def}, where $P_{Z^k}^*$ is defined in $(\ref{rep2})$, $Q_{X^n}^* = \operatorname{Unif}( \mathcal{X}^n)$ and, from $(\ref{rep2})$, $Q^{\sim t}_{X^n|Y^n}$ is given by 
\begin{align}
\begin{split}
Q^{\sim t}_{X^n|Y^n} &= \bigotimes_{i=1}^n Q^{(i), \sim t}_{X|Y}, \text{ where } \\ 
    Q^{(i), \sim t}_{X|Y}(x|y) &= \begin{cases}
        2^{-m} & y = e\\
        1 & x=y \neq e\\
        0 & x \neq y \neq e. 
    \end{cases}
\end{split}
\label{kernelforbec}
\end{align}
Hence, the $(t, Q_{Y^n|X^n}, n, k)$-based code has no dependence on $t$ and $(q_1, \ldots, q_n)$. In particular, it is based on the Poisson point process $\{ (\bar{X}_i^n, \bar{Z}_i^k), T_i  \}_{i \in \mathbb{N}}$ with intensity measure $Q_{X^n}^* \times P_{Z^k}^* \times \lambda_{\mathbb{R}_{\geq 0}}$, and the encoder and decoder mappings can be obtained by particularizing $(\ref{unif_enc})$ and $(\ref{unif_dec})$ for the $(P_{Z^k}^*, Q_{X^n}^*, Q^{\sim t}_{X^n|Y^n})$ source-channel code as follows: 
\begin{align}
   &s^k \mapsto \bar{X}_{i^\star}^n, \text{ where } i^\star = \begin{cases}
        \argmin_{i:\bar{Z}_i^k \in B_d(s^k)} T_i & \text{ if } \rho(s^k) > 0,\\
        \argmin_i  T_i & \text{ otherwise,}
    \end{cases} \label{bec_enc}\\
    &y^n \mapsto \bar{Z}_{i^*}^k, \text{ where } i^* = \argmin_{i: \bar{X}_{i, j} = y_j \forall j \text{ with } y_j \neq e} T_i. \label{bec_dec}
\end{align}
The proof of $(\ref{bec_dec})$ is given in Appendix \ref{bec_dec_proof} for convenience.
\end{IEEEproof}

Next, we lower bound $k^*_{\mathcal{U}^{(d)}}(d, \epsilon, n|P_n)$ for the channel-blind family $\mathcal{U}^{(d)} = \{U_{n, k}^{(d)} : (n, k) \in \mathbb{N}^2 \}$ from Proposition \ref{becuniversaltheorem}, where $P_n = P_{Y^n|X^n}^{\mathbf{P}}$ and $\mathbf{P} \in \mathscr{C}_{\text{reg}}(\mathcal{W})$ is a block erasure channel sequence satisfying certain regularity assumptions as described in Theorem \ref{BEC_particularization} below.
\begin{theorem}
Fix any $\epsilon \in (0, 1)$. Let the source distribution $P_S$, distortion level $d$ and distortion measure $\operatorname{d}$ satisfy Assumptions 
\ref{assumpr1_source}\,-\,\ref{assumpr3_source} and Conditions \ref{cond1dee}\,-\,\ref{weird9thmomcond}. Let $ \mathbf{P} = \left(P_{Y|X}^{(i)} \right)_{i \geq 1}$ be a sequence of block erasure channels, where $P_{Y|X}^{(i)} = \operatorname{BLEC}(2^m, p_i)$ and $p_i < 1$, such that there exist constants $\delta > 0$ and $n_0$ depending on $\mathbf{P}$ such that  
\begin{itemize}
    \item $\sum_{i=1}^n C(P_{Y|X}^{(i)}) \geq n\delta$ for all $n \geq n_0$, and 
    \item $\sum_{i=1}^n V(P_{Y|X}^{(i)} ) \geq n \delta$ for all $n \geq n_0$. 
\end{itemize}
If $V_d(P_S) > 0$, the second requirement above can be removed. Let $\mathscr{C}_{\text{reg}}(\mathcal{W})$ denote the set of all such channel sequences. Then for every $\mathbf{P} \in \mathscr{C}_{\text{reg}}(\mathcal{W})$, there exist constants $c$ and $n_0'$ such that for $n \geq n_0'$, 
\begin{align}
    k^*_{\mathcal{U}^{(d)}}(d, \epsilon, n|P_n) &\geq \frac{1}{R_d(P_S)} \sum_{i=1}^n C(P_{Y|X}^{(i)}) - \frac{Q^{-1}(\epsilon)}{R_d(P_S)} \sqrt{\frac{V_d(P_S)}{R_d(P_S)} \sum_{i=1}^n C(P_{Y|X}^{(i)}) + \sum_{i=1}^n V(P_{Y|X}^{(i)})} - \mbox{} \notag \\
    & \quad \quad \quad \quad \quad \quad \quad \quad \quad \quad \frac{(\overline{c}-1/2)}{R_d(P_S)} \log(n) - \frac{c}{R_d(P_S)}, \label{conjecturedoptimal}
\end{align}
where $\mathcal{U}^{(d)} = \{U_{n, k}^{(d)}\}$ is the channel-blind family of source-channel codes over $(\mathcal{W}_{n})_{n \geq 1}$ from Proposition \ref{becuniversaltheorem}, $P_n \coloneqq P_{Y^n|X^n}^{\mathbf{P}} = P_{Y|X}^{(1)} \times \cdots \times P_{Y|X}^{(n)}$, the constant $\overline{c}$ is as given in $(\ref{c28f})$, and
\begin{align*}
    C(P_{Y|X}^{(i)}) &= (1-p_i) \log(2^m),\\
    V(P_{Y|X}^{(i)}) &= p_i(1-p_i) (\log (2^m))^2.
\end{align*}
\label{BEC_particularization}
\end{theorem}
\textit{Proof:} The proof of Theorem \ref{BEC_particularization} is given in Appendix \ref{BEC_particularization_proof}.

Let $\operatorname{S}$ be any admissible source-channel blocklength set from Definition \ref{admissible_def}. An important intermediate result $(\ref{n350c})$ in the proof of Theorem \ref{BEC_particularization} establishes that for every $\epsilon \in (0, 1)$ and every channel sequence $\mathbf{P} \in \mathscr{C}_{\text{reg}}(\mathcal{W})$, the channel-blind code $U_{n, k}^{(d)}$ is $(d, \epsilon, n, k)$ whenever $(n, k) \in \operatorname{S}$, $n$ and $k$ are sufficiently large, and 
\begin{align}
  \sum_{i=1}^n C(P_{Y|X}^{(i)}) - k R_d(P_S) \geq  \sqrt{k V_d(P_S) + \sum_{i=1}^n V(P_{Y|X}^{(i)} )} \,\,Q^{-1}(\epsilon) + \left(\overline{c} - \frac{1}{2}\right) \log(k) + c,  \label{xwrsda}
\end{align}
where $c$ is a constant term.

In order to establish second-order universality over $\mathscr{C}_{\text{reg}}(\mathcal{W})$, where $\mathscr{C}_{\text{reg}}(\mathcal{W})$ is as specified in Theorem \ref{BEC_particularization}'s statement, we need a converse result establishing a refined asymptotic upper bound on 
$$k^*\left (d, \epsilon, n| P_{Y^n|X^n}^{\mathbf{P}} \right ),$$
that matches the RHS of $(\ref{conjecturedoptimal})$ up to the first two terms,
for every $\mathbf{P} \in \mathscr{C}_{\text{reg}}(\mathcal{W})$. For block erasure channels, such a converse result follows as a corollary of Theorem \ref{convoghlswhr0vm} below that we prove using \cite[Theorem 2]{kostina_JSCC} as the starting point. Theorem \ref{convoghlswhr0vm} gives a necessary condition for any code to be $(d, \epsilon, n, k)$ which in turn shows that the sufficient condition \eqref{xwrsda} for the channel-blind code is asymptotically tight to second-order; this was mentioned in Remark \ref{remarkeeusefulprop}. Another notable thing about Theorem \ref{convoghlswhr0vm} is that it is non-asymptotic in the channel blocklength $n$ and holds uniformly over all erasure probabilities $p_1, \ldots, p_n \in [0, 1]$. This is unlike standard JSCC normal approximation results which require both $n$ and $k$ to be sufficiently large. The technical reason for this specific strength of Theorem \ref{convoghlswhr0vm} is that the centered third absolute moment of the BLEC information density is upper bounded by a constant times its variance $V_n$. This means that the Berry--Esseen remainder term is of the form $C/\sqrt{V_n}$; through a careful subsequent analysis that splits the high and low variance cases, we obtain the final result that holds uniformly over all $n \geq 1$ and $p_1, \ldots, p_n \in [0, 1]$.     

\begin{theorem}
    Fix any $\epsilon \in (0, 1)$. Let the source distribution $P_S$, distortion level $d$ and distortion measure $\operatorname{d}$ satisfy Assumptions 
\ref{assumpr1_source}\,-\,\ref{assumpr3_source} and Conditions \ref{cond1dee}\,-\,\ref{weird9thmomcond}. Then there exist positive constants $k_0$ and $K < \infty$, depending only on $P_S$, $d$, $\operatorname{d}$, $\epsilon$ and $m$, such that for every $n \geq 1$, every $k \geq k_0$, and every collection of block erasure probabilities $p_1, \ldots, p_n \in [0, 1]$, if a $(d, \epsilon, n, k)$ source-channel code w.r.t. $\{P_{S^k}, \operatorname{d}_k, P_{Y^n|X^n} \}$ exists, where $P_{Y^n|X^n} = P_{Y|X}^{(1)} \times \cdots \times P_{Y|X}^{(n)}$ and $P_{Y|X}^{(i)} = \operatorname{BLEC}(2^m, p_i)$, then 
\begin{align}
\log(2^m)\sum_{i=1}^n(1-p_i)-kR_d(P_S)
&\geq
\sqrt{ k V_d(P_S)
+
\bigl(\log(2^m)\bigr)^2
\sum_{i=1}^n p_i(1-p_i)
}
\,Q^{-1}(\epsilon)
-\frac{1}{2} \log(k) -K.
\label{eq:normal-approximation-converse}
\end{align}
\label{convoghlswhr0vm}
\end{theorem}
\textit{Proof:} The proof of Theorem \ref{convoghlswhr0vm} is given in Appendix \ref{convoghlswhr0vm_proof}.

\begin{corollary}
Fix any $\epsilon\in(0,1)$. Let the source distribution $P_S$,
distortion level $d$ and distortion measure $\operatorname{d}$ satisfy
Assumptions $(\ref{assumpr1_source})$--$(\ref{assumpr3_source})$
and Conditions $(\ref{cond1dee})$--$(\ref{weird9thmomcond})$. Let $ \mathbf{P} = \left(P_{Y|X}^{(i)} \right)_{i \geq 1}$ be a sequence of block erasure channels, where $P_{Y|X}^{(i)} = \operatorname{BLEC}(2^m, p_i)$ and $p_i < 1$, such that there exist constants $\delta > 0$ and $n_0$ depending on $\mathbf{P}$ such that  
\begin{itemize}
    \item $\sum_{i=1}^n C(P_{Y|X}^{(i)}) \geq n\delta$ for all $n \geq n_0$, and 
    \item $\sum_{i=1}^n V(P_{Y|X}^{(i)} ) \geq n \delta$ for all $n \geq n_0$. 
\end{itemize}
If $V_d(P_S) > 0$, the second requirement above can be removed. Then there exist positive constants $n_1$ and $\widetilde{K} <\infty$ depending on $\mathbf{P}$ such
that for every $n\geq n_1$,
\begin{align}
    k^*(d,\epsilon,n\mid P_n) &\leq
    \frac{1}{R_d(P_S)}
    \sum_{i=1}^n C\left(P_{Y|X}^{(i)}\right)
    -
    \frac{Q^{-1}(\epsilon)}{R_d(P_S)}
    \sqrt{
        \frac{V_d(P_S)}{R_d(P_S)}
        \sum_{i=1}^n C\left(P_{Y|X}^{(i)}\right)
        +
        \sum_{i=1}^n V\left(P_{Y|X}^{(i)}\right)
    } \notag \\
    & \quad \quad \quad \quad \quad \quad 
    +
    \frac{1}{2R_d(P_S)}\log n
    +
    \widetilde{K},
    \label{eq:kstar-converse-nonstationary-bec}
\end{align}
where $P_n \coloneqq P_{Y^n|X^n}^{\mathbf{P}} = P_{Y|X}^{(1)} \times \cdots \times P_{Y|X}^{(n)}$. 
\label{corocoro}
\end{corollary}
\textit{Proof:} The proof of Corollary \ref{corocoro} is given in Appendix \ref{corocoro_proof}.

In view of Theorem \ref{BEC_particularization} and Corollary \ref{corocoro}, we have that the channel-blind family $\mathcal{U}^{(d)} = \{U_{n, k}^{(d)} \}$ is second-order universal over $\mathscr{C}_{\text{reg}}(\mathcal{W})$. Note also that we have proved a slightly stronger result than second-order universality because we have shown that the gap $k^* - k^*_{\mathcal{U}}$ in $(\ref{hereheregap})$ is $O(\log n)$ instead of merely $o(\sqrt{n})$.

For block erasure channels, the mechanism behind universality is particularly transparent. Indeed, the decoder in \eqref{bec_dec} restricts attention to Poisson points whose channel codewords are consistent with all unerased output blocks and then selects the consistent point with the smallest arrival time. The erasure probabilities affect only the random number and locations of the revealed blocks; they affect neither the consistency constraint nor the ordering among consistent candidates. Thus, the channel-blindness of this construction over BLECs is natural. Our results formalize this intuition for nonstationary BLECs and also show that the absence of a universality penalty persists beyond the first-order limit to the second-order term---in fact, up to an $O(\log n)$ gap. The Gaussian setting considered next is less immediate.

\subsection{Universal Source-Channel Code Construction for Stationary Gaussian Channels \label{tiredafgaussian}}

In this subsection, we let 
\begin{align}
    \mathcal{W} &= \left \{ P_{Y|X} : P_{Y|X}(\cdot|x) = \mathcal{N}(x, p^2) \text{ for all } x \in \mathbb{R}, p \in (0, \infty) \right \},\\  \mathcal{W}_n &= \left \{ P_{Y|X}^{\otimes n} : P_{Y|X} \in \mathcal{W}\right \}. \label{vn2438}
\end{align}
As in the previous subsection, we first construct a channel-blind family $\mathcal{U}^{(d, \alpha)}$ of source-channel codes over $(\mathcal{W}_n)_{n \geq 1}$
and then show that $\mathcal{U}^{(d, \alpha)}$ achieves the optimal first- and second-order asymptotic performance. Here, $\alpha > 0$ specifies the cost constraint on the channel input as follows: $\| X^n \|^2 \leq n \alpha$ almost surely, where $\| \cdot\|$ is the Euclidean norm. The same definitions from Section \ref{universalpluginconstruction} carry over with the additional parameter $\alpha$ included alongside $d$. The value of $\alpha$ is known at design time while the channel noise variance is not known.

Recall that for the case of block erasure channels (BLECs), a family of codes based on an arbitrary block erasure channel sequence with erasure probabilities in the interval $(0, 1)$ was automatically channel-blind (Proposition \ref{becuniversaltheorem}), which meant that the same JSCC construction was optimal across (almost) the whole set of BLEC sequences (Theorem \ref{BEC_particularization} and Corollary \ref{corocoro}). A similar observation applies to Gaussian channels but in the context of  channel coding only. A noise variance mismatch incurs no channel coding performance penalty because (i) spherical codebooks with codeword power $\alpha$ are optimal for every noise variance $p^2$ and (ii) an optimal channel decoder for Gaussian channels with equiprobable codewords is simply a nearest-neighbor decoder rule. However, in the context of JSCC, (ii) no longer holds since the decoder is now a MAP rule which corrects for a non-uniform prior induced by source encoding. In particular, the decoder of a JSCC matched to a specific Gaussian channel has dependence on the noise variance of the channel (see Remark \ref{b25remakr} given shortly).  

Hence, in the following channel-blind source-channel code construction for Gaussian channels, the code design will not be matched to any particular noise variance. Instead, the noise variance $p^2$ will be estimated by a plug-in estimator based on the channel output $y^n$. Define 
\begin{align*}
    \widehat{p}(y^n) &\coloneqq \max \left \{ \frac{1}{n}, \frac{||y^n||^2}{n} - \alpha\right \}
\end{align*}
Note that $\widehat{p}(y^n)$ is essentially the plug-in estimator of the channel noise variance; taking the maximum with $1/n$ prevents a nonpositive value of the estimate. Next, for every $y^n \in \mathbb{R}^n$ and $A \in \mathcal{B}(\mathbb{R}^n)$, define $\widehat{P}_{X^n|Y^n}(\cdot|y^n) \in \mathcal{P}(\mathbb{R}^n)$ by
\begin{align*}
    \widehat{P}_{X^n|Y^n}(A|y^n) &\coloneqq \frac{\int_A  \exp\left(- \frac{||y^n - x^n||^2}{2 \widehat{p}(y^n)} \right) dQ_{X^n}^{\operatorname{u}}(x^n)  }{\int_{\mathbb{R}^n} \exp\left(- \frac{||y^n - u^n||^2}{2 \widehat{p}(y^n)} \right) dQ_{X^n}^{\operatorname{u}}(u^n)   }. 
\end{align*}

\begin{definition}
We let $U_{n, k}^{(d, \alpha)}$ be the $(P_{Z^k}^*, Q_{X^n}^{\operatorname{u}}, \widehat{P}_{X^n|Y^n} )$ source-channel code from Definition \ref{unified_def}, where $P_{Z^k}^*$ is defined in $(\ref{rep2})$-$(\ref{rep3})$ and $$Q_{X^n}^{\operatorname{u}} = \operatorname{Unif}\left( \left \{x^n \in \mathbb{R}^n : \frac{1}{n}\sum_{i=1}^n x_i^2 = \alpha \right \} \right).$$ 
Then let 
\begin{align}
    \mathcal{U}^{(d, \alpha)} \coloneqq \{ U_{n, k}^{(d, \alpha)} : (n, k) \in \mathbb{N}^2\} \label{familyofcodes}
\end{align}
be the family of all such source-channel codes. \label{defofcodesforgaussianuniv}
\end{definition}

Evidently, $U_{n, k}^{(d, \alpha)}$ is channel-blind over $\mathcal{W}_n$ and hence, the family of codes  $(\ref{familyofcodes})$ comprising these $U_{n, k}^{(d, \alpha)}$ is channel-blind over $(\mathcal{W}_n)_{n \geq 1}$. Specializing Definition \ref{unified_def} for the $(P_{Z^k}^*, Q_{X^n}^{\operatorname{u}}, \widehat{P}_{X^n|Y^n} )$ source-channel code, we obtain that the shared randomness between the encoder and decoder is $\{ (\bar{X}_i^n, \bar{Z}_i^k), T_i  \}_{i \in \mathbb{N}}$, which are the points of a marked Poisson process with intensity measure $Q_{X^n}^{\operatorname{u}} \times P_{Z^k}^* \times \lambda_{\mathbb{R}_{\geq 0}}$ and independent of $S^k \sim P_{S^k}$. And the encoder and decoder mappings are 
\begin{align}
    &s^k \mapsto \bar{X}^n_{i^\star}, \text{ where } i^\star = \begin{cases}
        \argmin_{i:\bar{Z}_i^k \in B_d(s^k)} T_i & \text{ if } \rho(s^k) > 0\\
        \argmin_i  T_i & \text{ otherwise,}
    \end{cases} \label{univ_gaussian_enc}\\
    &y^n \mapsto \bar{Z}_{i^*}^k, \text{ where } i^* = \argmin_i \left \{ ||y^n - \bar{X}_i^n||^2 + 2 \widehat{p}(y^n) \log T_i\right \}, \label{univ_gaussian_dec}
\end{align}
where $\rho(s^k) = P_{Z^k}^*(B_d(s^k))$. 
\begin{remark}
    Note that with a known $p^2$, a matched decoder for the Gaussian channel will have
    \begin{align}
        \argmin_i \left \{ ||y^n - \bar{X}_i^n||^2 + 2 p^2 \log T_i\right \} \label{seethisforexampleknownp2}
    \end{align}
    instead of the $\widehat{p}(y^n)$-based rule in $(\ref{univ_gaussian_dec})$. Since $\widehat{p}(Y^n) \to p^2$ in probability, first-order optimality of the source-channel code $(\ref{univ_gaussian_enc})$-$(\ref{univ_gaussian_dec})$ is to be expected. Theorem \ref{thm_finite_block1_awgn_univ} shows, however, that the plug-in scheme also attains the optimal second-order asymptotic performance.
    \label{b25remakr} Before proving Theorem \ref{thm_finite_block1_awgn_univ}, we state a number of intermediate results. 
\end{remark}

The following is a corollary of Proposition \ref{bxc} for the source-channel code $U_{n, k}^{(d, \alpha)}$ as defined in Definition \ref{defofcodesforgaussianuniv}.
\begin{corollary}
Fix a cost level $0<\alpha < \infty$. Let $P_{Y^n|X^n} \in \mathcal{W}_n$ where $\mathcal{W}_n$ is given in $(\ref{vn2438})$. Then the ensemble-average excess distortion probability $\overline{P}_e(d, \alpha, n, k)$  with respect to $\{P_{S^k}, \operatorname{d}_k, P_{Y^n|X^n}\}$ of the code $U_{n, k}^{(d, \alpha)}$ satisfies  
    \begin{align}
        \overline{P}_e(d, \alpha, n, k) \leq \mathbb{E}\left [ \left(1 + P_{Z^k}^*(B_d(S^k)) \exp\left( \imath_{\widehat{P}_{X^n|Y^n}, Q_{X^n}^{\operatorname{u}}}(Y^n, X^n) \right) \right)^{-1}  \right], \label{gva2na2}
    \end{align}
    where the expectation above is w.r.t. $P_{S^k}, Q_{X^n}^{\operatorname{u}}$ and $P_{Y^n|X^n}$.  
\end{corollary}

A central limit theorem-based analysis of the information density 
\begin{align}
  \imath_{\widehat{P}_{X^n|Y^n}, Q_{X^n}^{\operatorname{u}}}(Y^n, X^n) = \log \left( \frac{d \widehat{P}_{X^n|Y^n  = y^n} }{dQ_{X^n}^{\operatorname{u}}}(X^n) \right) \Bigg |_{y^n = Y^n} \label{b2556wgtgh}   
\end{align}
requires new techniques that are given in Lemmas \ref{lemmafirstlowerboundthen} and \ref{mytwosidedbndprop} below. 
\begin{lemma}
    Let $(X^n, Y^n)$ have the joint law $Q_{X^n}^{\operatorname{u}} \times P_{Y^n|X^n}$. Then there exist constants $C_1, C_2 > 0$ and $n_0 > 0$ such that for $n \geq n_0$, 
    \begin{align*}
        \mathbb{P}\left( \imath_{\widehat{P}_{X^n|Y^n}, Q_{X^n}^{\operatorname{u}}}(Y^n, X^n) \geq \frac{n}{2} \log \left(1 + \frac{\alpha}{\widehat{p}(Y^n)} \right) + \frac{||Y^n||^2}{2(\alpha + \widehat{p}(Y^n) )}  - \frac{||Y^n - X^n||^2}{2 \widehat{p}(Y^n)} - C_1\right) \geq 1 - \frac{C_2}{n}.
    \end{align*} \label{lemmafirstlowerboundthen}
\end{lemma}
\textit{Proof:} The proof of Lemma \ref{lemmafirstlowerboundthen} is given in Appendix \ref{lemmafirstlowerboundthen_proof}.

\begin{lemma}
Let $(X^n, Y^n)$ have the joint law $Q_{X^n}^{\operatorname{u}} \times P_{Y^n|X^n}$, so that we can write $Y^n = X^n + Z^n$,
    \begin{align*}
    X^n &= \sqrt{n \alpha} \frac{G^n}{||G^n||} \, \text{ and } \,  Z^n = p N^n, 
\end{align*}
where $G^n \ind N^n$, and $G^n$ and $N^n$ have i.i.d. $\mathcal{N}(0, 1)$ components. Define 
\begin{align}
    \Lambda(X^n, Y^n) &\coloneqq \frac{n}{2} \log \left(1 + \frac{\alpha}{\widehat{p}(Y^n)} \right) + \frac{||Y^n||^2}{2(\alpha + \widehat{p}(Y^n) )}  - \frac{||Y^n - X^n||^2}{2 \widehat{p}(Y^n)}, \label{lambdax6ny6n}\\
        S_n &\coloneqq \sum_{i=1}^n \left [ \frac{1}{2} \log \left(1 + \frac{\alpha}{p^2} \right) +  \frac{p \sqrt{\alpha}}{\alpha + p^2} G_i N_i - \frac{\alpha}{2(\alpha + p^2)} (N_i^2 - 1) \right ]. \label{Sndefiid}
\end{align}
Then there exist positive constants $C_1, C_2$ and $n_0$ such that for $n \geq n_0$,
\begin{align*}
        \mathbb{P}\left(  \left| \Lambda(X^n, Y^n) - S_n \right| \leq C_1 \log(n) \right) \geq 1 - \frac{C_2}{n}.
\end{align*} 
    \label{mytwosidedbndprop}
\end{lemma}
\textit{Proof:} The proof of Lemma \ref{mytwosidedbndprop} is given in Appendix \ref{mytwosidedbndprop_proof}.

The proof of Lemma \ref{mytwosidedbndprop} relies on the multivariate delta method. The key step is to express $\Lambda(X^n, Y^n)/n$, on a suitably restricted high-probability domain, as a smooth function of empirical averages of i.i.d. variables. Then we Taylor-expand around the mean so the linear term is the i.i.d. sum $S_n$. Finally, the proof uses concentration to show that the nonlinear remainder term over the expansion neighborhood is $O(n^{-1} \log n)$, while the probability of leaving this neighborhood is only $O(n^{-1})$.

The high-level objective of Lemmas \ref{lemmafirstlowerboundthen} and  \ref{mytwosidedbndprop} is simple: lower bound the complicated information density $(\ref{b2556wgtgh})$ by $\Lambda(X^n, Y^n)$ as defined in 
$(\ref{lambdax6ny6n})$, which in turn is approximated by a sum of i.i.d. random variables $S_n$ as defined in $(\ref{Sndefiid})$. Taken together with the sum of i.i.d. source random variables obtained from the application of Lemma \ref{lemma5}, we can prove the following theorem.   

\begin{theorem}
Fix any $\epsilon \in (0, 1)$ and $0 < \alpha < \infty$. Let the source distribution $P_S$, distortion level $d$ and distortion measure $\operatorname{d}$ satisfy Assumptions 
\ref{assumpr1_source}\,-\,\ref{assumpr3_source} and Conditions \ref{cond1dee}\,-\,\ref{weird9thmomcond}. Then for every $p \in (0, \infty)$, there exist constants $c$ and $n_0$ such that for $n \geq n_0$,
\begin{align}
    k^*_{\mathcal{U}}(d, \alpha, \epsilon, n|P_{Y|X}^{\otimes n}) \geq \frac{n C_\alpha(P_{Y|X})}{R_d(P_S)} - \frac{Q^{-1}(\epsilon)}{R_d(P_S)} \sqrt{ V_d(P_S) \frac{n C_\alpha(P_{Y|X}) }{ R_d(P_S)}   + n V_\alpha(P_{Y|X})} - c \log(n), \label{desire_without_proof}
\end{align}
where $k^*_{\mathcal{U}}$ is defined similarly to \eqref{channelblind_achievability_region}, $\mathcal{U} = \{U_{n, k}^{(d, \alpha)}\}$ is the channel-blind family of source-channel codes over $(\mathcal{W}_{n})_{n \geq 1}$ from Definition \ref{defofcodesforgaussianuniv}, $P_{Y|X = x} = \mathcal{N}(x, p^2)$, and 
\begin{align}
\begin{split}
    C_{\alpha}(P_{Y|X}) &= \frac{1}{2} \log \left(1 + \Gamma \right),\\
    V_{\alpha}(P_{Y|X}) &= \frac{\Gamma(\Gamma + 2)}{2(\Gamma + 1)^2},\\
        \Gamma &= \frac{\alpha}{p^2}.
\end{split}
     \label{anfhdf}
\end{align}
    \label{thm_finite_block1_awgn_univ}
\end{theorem}

\begin{IEEEproof}[Proof Outline] The full proof of Theorem \ref{thm_finite_block1_awgn_univ} is omitted because the main novelty in the argument lies in Lemmas \ref{lemmafirstlowerboundthen} and \ref{mytwosidedbndprop} above, while the rest of the argument is similar to the derivation given in the proofs of Theorems \ref{thm_finite_block1} and \ref{BEC_particularization}. Specifically, the proof starts by directly evaluating and upper bounding the RHS of \eqref{gva2na2} and applies Lemma \ref{lemma5}, Lemma \ref{lemmafirstlowerboundthen} and Lemma \ref{mytwosidedbndprop} in that order. Then it uses the identity \eqref{log_identity} and uses similar derivation in the proof of Theorem \ref{thm_finite_block1} in order to obtain the result that $U_{n, k}^{(d, \alpha)}$ is $(d, \epsilon, n, k)$ with respect to $\{P_{S^k}, \operatorname{d}_k, P_{Y^n|X^n}\}$ whenever $n$ and $k$ are sufficiently large, $(n, k)$ is in some admissible set $\operatorname{S}$, and 
\begin{align}
        n C_{\alpha}(P_{Y|X}) - k R_d(P_S) &\geq  \sqrt{k V_d(P_S) + n V_{\alpha}(P_{Y|X})} \,\,Q^{-1}(\epsilon) + \tilde{c} \log(n),    \label{thm_cond_awgn_univ}
\end{align}
where $\tilde{c}$ is some constant independent of $n$ and $k$. Then a similar derivation as was done after \eqref{n350c} in the proof of Theorem \ref{BEC_particularization} obtains the inequality \eqref{desire_without_proof}.
\end{IEEEproof}

Finally, to establish second-order universality of the channel-blind family of codes in Theorem \ref{thm_finite_block1_awgn_univ}, we invoke the JSCC converse result for memoryless and stationary Gaussian channels from \cite{kostina_JSCC}, which is stated below for convenience.    

\begin{proposition}[{\cite[Theorem 10]{kostina_JSCC}}]
Fix any $\epsilon \in (0, 1)$ and $0 < \alpha < \infty$. Let the source distribution $P_S$, distortion level $d$ and distortion measure $\operatorname{d}$ satisfy Assumptions 
\ref{assumpr1_source}\,-\,\ref{assumpr3_source} and Conditions \ref{cond1dee}\,-\,\ref{weird9thmomcond}. Then for every $p \in (0, \infty)$, there exist constants $c$ and $n_0$ such that for $n \geq n_0$, 
\begin{align}
    k^*\left (d, \alpha,\epsilon,  n| P_{Y|X}^{\otimes n} \right ) \leq \frac{n C_\alpha(P_{Y|X})}{R_d(P_S)} - \frac{Q^{-1}(\epsilon)}{R_d(P_S)} \sqrt{ V_d(P_S) \frac{n C_\alpha(P_{Y|X}) }{ R_d(P_S)} + n V_\alpha(P_{Y|X})} + \frac{1}{2 R_d(P_S)} \log(n) + c,  \label{kostinaqwyandverdu}
\end{align}
where $P_{Y|X = x} = \mathcal{N}(x, p^2)$, and $C_\alpha(P_{Y|X})$ and $V_\alpha(P_{Y|X})$ are as specified in $(\ref{anfhdf})$. 
\label{propyuofromkostina}
\end{proposition}
\textit{Proof:} Proposition \ref{propyuofromkostina} is a  corollary of \cite[Theorem 10]{kostina_JSCC}. Specifically, $(\ref{kostinaqwyandverdu})$ can be derived by applying \cite[(119)]{kostina_JSCC} to \cite[(127)]{kostina_JSCC}.

Note also that for both block erasure channels and Gaussian channels, we have proved a slightly stronger result than second-order universality because we have shown that the gap $k^* - k^*_{\mathcal{U}}$ in $(\ref{herestatheregap})$ is $O(\log n)$ instead of merely $o(\sqrt{n})$.

\section{Conclusion and Future Directions \label{conclusion}}

Our mismatched-design result (Theorem \ref{thm_finite_block1}) is an achievability result for a specific mismatched-design randomized JSCC scheme. This result entails the scheme-envelope achievability of the generalized mutual information in the JSCC setting and for nonstationary channels over a fixed general alphabet. Theorem \ref{thm_finite_block1} also subsumes achievability results for nonstationary channels in the traditional matched JSCC. For future directions, one can study mismatched decoding involving an arbitrary decoding metric $q(x, y)$ as well as derive the achievability of the LM rate in the JSCC framework. A possible step toward that extension is to choose the Gibbs posterior $K_{X|Y}$ in Definition \ref{unified_def} of the form 
\begin{align}
    K_{X|Y}(A|y) &= \frac{ \int_A q(x, y)^t e^{a(x)} Q_X(dx)  }{\int_\mathcal{X} q(u, y)^t e^{a(u)} Q_X(du)  }, \label{mentionconcccc} 
\end{align}
and the decoding rule $(\ref{unif_dec})$ specialized for this choice is 
\begin{align*}
    i^*(y) &= \argmax_i \left \{ t  \log q(\bar{X}_i, y) + a(\bar{X}_i)  -\log(T_i)\right \}. 
\end{align*} 
Furthermore, upgrading from scheme-envelope achievability to a single-scheme achievability of the GMI/LM rate in the JSCC setting is an interesting problem for future study. 

We also defined a new notion of universal JSCC and showed that a channel-blind family of JSCC can be second-order universal over a regularized set of memoryless and nonstationary block erasure channels and a set of memoryless and stationary Gaussian channels. The broader motivation for formalizing a notion of universal JSCC is to better understand when JSCC designs can be separated, at least partially, from exact channel knowledge. As explained in Remarks \ref{remarkeeuseful22prop} and \ref{remarkeeusefulprop}, our universal results imply a construction-level separation, while retaining second-order asymptotic optimality at the level of a family of codes. Each channel-blind code in the family is associated with an operating point $(d, n, k)$ and is certified as $(d,\epsilon,n,k)$ for every $\epsilon$ and every channel in the channel class by an explicit sufficient condition that is asymptotically tight to second-order.





\appendices

\section{Proof of Proposition \ref{bxc} \label{bxc_proof}} 

For $P_S$-a.e. $s \in \mathcal{S}$, let $P_{\check{Z};s} \in \mathcal{P}(\mathcal{Z})$ be defined as 
    \begin{align*}
        P_{\check{Z};s}(A) \coloneqq \begin{cases}
            P_Z^*\left( A \cap B_d(s)  \right)/\rho(s) & \text{ if } \rho(s) > 0,\\
            P_Z^*(A) & \text{ if } \rho(s) = 0, 
        \end{cases}
    \end{align*}
where $\rho(s) \coloneqq P_Z^*(B_d(s))$. Whenever $\rho(s) > 0$, a sample $\check{Z}$ drawn from the distribution $P_{\check{Z};s}$ satisfies $\operatorname{d}(s, \check{Z}) \leq d$ almost surely. The core idea is that the encoder and decoder both have access to a Poisson point process as shared randomness. Given a source realization $S = s$, the encoder queries the shared Poisson point process using $Q_X \times P_{\check{Z};s} \in \mathcal{P}(\mathcal{X} \times \mathcal{Z})$, obtaining a point $(X, \check{Z})$ and then transmits $X$. Given $Y = y$, the decoder queries the same Poisson point process using the probability measure $K_{X|Y=y} \times P_Z^* \in \mathcal{P}(\mathcal{X} \times \mathcal{Z})$, and denotes the selected point by $(\widehat{X}, \widehat{Z})$. The decoder then outputs $\widehat{Z}$. The excess-distortion probability is then $\overline{P}_e(d) = \mathbb{P}(\operatorname{d}(S, \widehat{Z}) > d)$, which can be upper bounded by $\mathbb{P}(\rho(S) = 0) + \mathbb{P}(\rho(S) > 0, \check{Z} \neq \widehat{Z})$. We now make this operation rigorous.

Let $\{ (\bar{X}_i, \bar{Z}_i), T_i  \}_{i \in \mathbb{N}}$ be the points of a Poisson process with intensity measure $Q_X \times P_Z^* \times \lambda_{\mathbb{R}_{\geq 0}}$ and independent of $S \sim P_S$. Using this Poisson point process, the Poisson functional representation \cite[Definition 1]{alternative_oneshot_JSCC} of the target probability measure $Q_X \times P_{\check{Z};s}$ is 
\begin{align*}
    \left(X, \check{Z}\right) &=  \left( \bar{X}_{i^\star(s)}, \bar{Z}_{i^\star(s)} \right),  
\end{align*}
where 
\begin{align*}
    i^\star(s) = \argmin_{i} T_i \left( \frac{d(Q_X \times P_{\check{Z};s} )}{d(Q_X \times P_Z^*)} (\bar{X}_i, \bar{Z}_i) \right )^{-1}. 
\end{align*}
It can be checked that the encoding function $s \mapsto \bar{X}_{i^\star(s)}$ is the same as  $(\ref{unif_enc})$. Likewise, using the same Poisson point process, the Poisson functional representation of the target probability measure $K_{X|Y=y} \times P_Z^*$ is
\begin{align*}
    \left( \widehat{X}, \widehat{Z} \right) = \left( \bar{X}_{i^\star(y)}, \bar{Z}_{i^\star(y)} \right),  
\end{align*}
where 
\begin{align*}
    i^\star(y) = \argmin_{i} T_i \left( \frac{d(K_{X|Y=y} \times P_Z^* )}{d(Q_X \times P_Z^*)} (\bar{X}_i, \bar{Z}_i) \right )^{-1}. 
\end{align*}
It can be checked that the decoding function $y \mapsto \bar{Z}_{i^\star(y)}$ is the same as $(\ref{unif_dec})$. Note that while $(X, \check{Z}) \sim Q_X \times P_{\check{Z};s}$, it is not necessarily true that $( \widehat{X}, \widehat{Z}) \sim K_{X|Y=y} \times P_Z^*$, since $Y$ depends on the Poisson process. 

We have $(S, X, Y, \check{Z}) \sim P_S \times Q_X \times P_{Y|X} \times P_{\check{Z};S}$, where $P_{Y|X}$ is the true channel. Then we have 
\begin{align*}
        &\mathbb{P}(\operatorname{d}(S, \widehat{Z}) > d) \\
        &\leq \mathbb{P}(\rho(S) = 0) + \mathbb{P}\left( \rho(S) > 0 \cap \widehat{Z} \neq \check{Z} \right)\\
        &= \mathbb{P}(\rho(S) = 0) + \mathbb{E}\left [ \mathds{1}\left( \rho(S) > 0   \right) \mathds{1}\left(\widehat{Z} \neq \check{Z} \right) \right]\\
        &= \mathbb{P}(\rho(S) = 0) + \mathbb{E}\left [ \mathbb{E}\left [  \mathds{1}\left( \rho(S) > 0   \right) \mathds{1}\left(\widehat{Z} \neq \check{Z} \right)| S, X, Y, \check{Z} \right] \right]\\
        &= \mathbb{P}(\rho(S) = 0) + \mathbb{E}\left [ \mathds{1}\left( \rho(S) > 0   \right) \mathbb{E}\left [   \mathds{1}\left(\widehat{Z} \neq \check{Z} \right)| S, X, Y, \check{Z} \right] \right]\\
        &= \mathbb{P}(\rho(S) = 0) + \mathbb{E}\left [ \mathds{1}\left( \rho(S) > 0   \right) \mathbb{P}\left (   \widehat{Z} \neq \check{Z} | S, X, Y, \check{Z} \right) \right]\\
        &\leq \mathbb{P}(\rho(S) = 0) + \mathbb{E}\left [ \mathds{1}\left( \rho(S) > 0   \right) \mathbb{P}\left ( (\widehat{X} ,\widehat{Z} ) \neq (X,   \check{Z}) | S, X, \check{Z}, Y  \right) \right]\\
        &\stackrel{(a)}{\leq} \mathbb{P}(\rho(S) = 0) + \mathbb{E}\left [ \mathds{1}\left( \rho(S) > 0   \right) \left [1 - \left(1 + \frac{d (Q_X \times P_{\check{Z}; S} )}{d (K_{X|Y}(\cdot|Y) \times P_Z^* )}(X, \check{Z}) \right)^{-1} \right] \right]\\
        &= \mathbb{P}(\rho(S) = 0) + \mathbb{E}\left [ \mathds{1}\left( \rho(S) > 0   \right) \left [1 - \left(1 + \frac{1}{\rho(S)} \exp\left(-\log \frac{d K_{X|Y}(\cdot|Y)}{dQ_X}(X) \right) \right)^{-1} \right] \right]\\
        &= \mathbb{P}(\rho(S) = 0) + \mathbb{E}\left [ \mathds{1}\left( \rho(S) > 0   \right) \left [ \left(1 + \rho(S) \exp\left( \imath_{K_{X|Y}, Q_X}(Y, X) \right) \right)^{-1}\right] \right]\\
        &= \mathbb{E}\left [ \left(1 + \rho(S) \exp\left( \imath_{K_{X|Y}, Q_X}(Y, X) \right) \right)^{-1}  \right]. 
\end{align*}
Inequality $(a)$ follows by applying the Conditional Poisson Matching Lemma \cite[Lemma 2]{alternative_oneshot_JSCC}, which requires the assumption $K_{X|Y}(\cdot|Y) \ll Q_X$ almost surely from Definition \ref{unified_def}. Specifically, 
\cite[Lemma 2]{alternative_oneshot_JSCC} requires both $Q_X \times P_{\check{Z};s} \ll Q_X \times P_{Z}^*$ and $K_{X|Y = y} \times P_Z^* \ll Q_X \times P_{Z}^*$, both of which are true in our case. Note also that in inequality $(a)$, we used the shorthand  
\begin{align*}
    \frac{d (Q_X \times P_{\check{Z}; S} )}{d (K_{X|Y}(\cdot|Y) \times P_Z^* )}(X, \check{Z}) = \frac{\frac{d (Q_X \times P_{\check{Z}; S} )}{d (Q_X \times P_Z^* )}(X, \check{Z})}{\frac{d (K_{X|Y = y} \times P_{Z}^* )}{d (Q_X \times P_Z^* )}(X, \check{Z})}. 
\end{align*}

\section{Proof of Theorem \ref{thm_finite_block1} \label{thm_finite_block1_proof}}

We start by directly evaluating the RHS of $(\ref{g2n})$. Using Lemma \ref{lemma5}, the RHS of $(\ref{g2n})$ can be upper bounded for $k \geq k_0$ by
\begin{align}
    \mathbb{E}\left [ \left(1 + \exp\left( \sum_{i=1}^n \imath_{Q_{X|Y}^{(i), \sim t}, Q_{X_i}^*}(Y_i, X_i) - \sum_{i=1}^k \jmath_d(S_i, P_S, \operatorname{d}) - \left( \overline{c} - \frac{1}{2} \right) \log k - c_1 \right)  \right)^{-1}  \right] + \frac{K}{\sqrt{k}}. \label{g34rg}
\end{align}
To simplify further analysis, we write $(\ref{g34rg})$ as 
\begin{align}
\mathbb{E}\left [ \frac{1}{1 + \exp(U_{k, n} -u_k )} \right] + \frac{K}{\sqrt{k}},
 \label{67b}
\end{align}
where
\begin{align}
    U_{k, n} &=  \sum_{i=1}^n \imath_{Q_{X|Y}^{(i), \sim t}, Q_{X_i}^*}(Y_i, X_i) - \sum_{i=1}^k \jmath_d(S_i, P_S, \operatorname{d}),\\
    u_k &=  \left( \overline{c} - \frac{1}{2} \right) \log k + c_1.
\end{align}
Now let $L$ be a zero-mean and independent random variable with CDF given by
\begin{align}
    F_L(\ell) = \frac{1}{1 + e^{-\ell}}. \label{32095} 
\end{align}
It can be checked that both the variance $V_L$ and the third absolute moment of $L$ are finite. Then we can write $(\ref{67b})$ as 
\begin{align}
    \operatorname{Pr}\left( U_{k, n} + L \leq u_k \right) + \frac{K}{\sqrt{k}}. \label{-60....}
\end{align}
\begin{remark}
    The proof step in $(\ref{-60....})$ lets us apply the Berry--Esseen Theorem directly to the sum $U_{k, n} + L$ as opposed to applying Berry--Esseen Theorem on a typical event as was done in \cite[(409)]{kostina_JSCC}. This is the reason we obtain an improved third-order term compared to the achievability half of \cite[Theorem 10]{kostina_JSCC}.   
\end{remark}

Recall that $(S^k, X^n, Y^n, L) \sim P_{S^k} Q^*_{X^n} P_{Y^n|X^n} F_L$. Note that
\begin{align*}
    \mathbb{E}\left [ U_{k, n} \right] &= \sum_{i=1}^n C_t(P_{Y|X}^{(i)} \| Q_{Y|X}^{(i)}) -k R_d(P_S),\\ 
   \operatorname{Var}\left(U_{k, n} \right) &= \sum_{i=1}^n V_t(P_{Y|X}^{(i)} \| Q_{Y|X}^{(i)}) + k V_d(P_S).
\end{align*} 
Now applying the Berry--Esseen Theorem to $(\ref{-60....})$ gives us, for all $(n, k) \in \operatorname{S}$, 
\begin{align}
    &\operatorname{Pr}\left( U_{k, n} + L \leq u_k \right) + \frac{K}{\sqrt{k}} \notag \\
    &\leq Q\left(\frac{\mathbb{E}[U_{k, n}] - u_k }{\sqrt{\operatorname{Var}(U_{k, n}) + V_L}} \right) + \frac{K_1}{\sqrt{n + k}} + \frac{K_2}{\sqrt{k}}, \label{06jdhwgki}
\end{align}
where $K_1$ and $K_2$ are finite constant terms whose existence is implied by the Assumptions $(\ref{assumpbe1})$-$(\ref{assumpbe2})$. Now define 
$$\epsilon_{n, k} = \epsilon - \frac{K_1}{\sqrt{n + k}} - \frac{K_2}{\sqrt{k}}.$$
Then $(\ref{06jdhwgki})$ is less than or equal to $\epsilon$ if 
\begin{align}
    \mathbb{E}[U_{k, n}] \geq \sqrt{\operatorname{Var}(U_{k, n}) + V_L} Q^{-1}(\epsilon_{n, k}) + u_k. \label{b35}
\end{align}
It remains to show that $(\ref{b35})$ is implied by $(\ref{thm_cond})$ for some constant $c$ as written in $(\ref{thm_cond})$. We write $\epsilon_{n, k} = \epsilon - \Delta_{n, k}$, 
where 
\begin{align}
    \Delta_{n, k} &= \frac{K_1}{\sqrt{n + k}} + \frac{K_2}{\sqrt{k}}. \label{deltankdefadeel}
\end{align} 
For any fixed $\epsilon \in (0,1)$, define
\begin{align}
    C_\epsilon
    &= \sup_{u \in [\epsilon/2,\epsilon]}
    \frac{1}{\phi\left(Q^{-1}(u)\right)}, \label{Cepsilondef}
\end{align}
where $\phi$ is the standard normal PDF. Since the function $u \mapsto 1/\phi\left(Q^{-1}(u)\right)$ is continuous on the compact interval $[\epsilon/2,\epsilon]$, the constant $C_\epsilon$ is finite. Moreover, for sufficiently large $n$ and $k$, we have $\Delta_{n,k} \leq \epsilon/2$, and hence $\epsilon_{k,n}
    = \epsilon-\Delta_{n,k}
    \in [\epsilon/2,\epsilon].$
Using the mean value theorem, we can write
\begin{align}
    Q^{-1}(\epsilon_{k,n})
    &= Q^{-1}(\epsilon)
    + \frac{\Delta_{n,k}}
    {\phi\left(Q^{-1}(\tilde{\epsilon})\right)} \notag \\
    &\leq Q^{-1}(\epsilon)
    + C_\epsilon \Delta_{n,k}, \notag
\end{align}
for some $\tilde{\epsilon} \in (\epsilon-\Delta_{n,k},\epsilon)$. Therefore, for all sufficiently large $n$ and $k$,
\begin{align}
    Q^{-1}(\epsilon_{k,n})
    \leq Q^{-1}(\epsilon)+C_\epsilon\Delta_{n,k}. \notag
\end{align}

Note also that 
\begin{align}
        &\frac{\operatorname{Var}\left( U_{k, n} \right)}{n + k} \notag  \\
        &\leq \frac{1}{n + k}\left [ k \left(\mathbb{E}\left[ \big | \jmath_d(S, P_S, \operatorname{d}) - R_d(P_S) \big |^3 \right]\right)^\frac{2}{3} +  \sum_{i=1}^n \left(\mathbb{E}\left [  \Big | \imath_{Q_{X|Y}^{(i), \sim t }, Q_{X_i}^*}(Y_i, X_i) - C_t(P_{Y|X}^{(i)} \| Q_{Y|X}^{(i)}) \Big|^3 \right]\right)^{\frac{2}{3}}  \right]  \notag \\
        &\leq \left(\frac{1}{n + k}\left [ k \mathbb{E}\left[ \big | \jmath_d(S, P_S, \operatorname{d}) - R_d(P_S) \big |^3 \right] +  \sum_{i=1}^n \mathbb{E}\left [  \Big | \imath_{Q_{X|Y}^{(i), \sim t}, Q_{X_i}^*}(Y_i, X_i) - C_t(P_{Y|X}^{(i)} \| Q_{Y|X}^{(i)}) \Big|^3 \right]  \right]\right)^{\frac{2}{3}} \notag \\
        &\leq \kappa_1^{\frac{2}{3}} < \infty, \label{6n}
    \end{align}
where the first inequality above uses   
monotonicity of $L^p$ norms, the second inequality above uses Jensen's inequality and the third inequality above uses the assumption $(\ref{assumpbe1})$.

Then for sufficiently large $n$ and $k$, we can upper bound 
\begin{align}
    &\sqrt{\operatorname{Var}(U_{k, n}) + V_L} Q^{-1}(\epsilon_{n, k}) \notag \\
    &\leq \sqrt{\operatorname{Var}(U_{k, n}) + V_L } Q^{-1}(\epsilon) + \sqrt{\operatorname{Var}(U_{k, n}) +V_L }C_\epsilon \Delta_{n, k} \notag \\
    &\stackrel{(a)}{\leq} \sqrt{\operatorname{Var}(U_{k, n}) +V_L} Q^{-1}(\epsilon) +  \sqrt{2\kappa_1^{\frac{2}{3}} (n + k) } \, C_\epsilon \left( \frac{K_1}{\sqrt{n + k}} + \frac{K_2}{\sqrt{k}}\right) \notag \\
    &\stackrel{(b)}{\leq} \sqrt{\operatorname{Var}(U_{k, n})} Q^{-1}(\epsilon) + c_2 +  \sqrt{2\kappa_1^{\frac{2}{3}} (n + k) } \, C_\epsilon \left( \frac{K_1}{\sqrt{n + k}} + \frac{K_2}{\sqrt{k}}\right) \notag \\
    &\leq \sqrt{\operatorname{Var}(U_{k, n}) } Q^{-1}(\epsilon) + \tilde{c}_2  + \tilde{c}_3 \sqrt{1 + \frac{n}{k}} \notag \\
    &\stackrel{(c)}{\leq } \sqrt{\operatorname{Var}(U_{k, n}) } Q^{-1}(\epsilon) + c_2  \label{v3c}
\end{align}
for some constants $c_2, \tilde{c}_2$ and $\tilde{c}_3$ independent of $(n, k) \in \operatorname{S}$. Inequality $(a)$ above follows from $(\ref{deltankdefadeel})$, $(\ref{6n})$ and the fact that $V_L < \infty$. Inequality $(b)$ uses Assumption $(\ref{assumpbe2})$. Inequality $(c)$ follows from the admissibility property that requires $n/k$ to be uniformly upper bounded by a fixed constant for all $(n, k) \in \operatorname{S}$.

Therefore, in view of $(\ref{b35})$ and $(\ref{v3c})$, if $(n, k) \in \operatorname{S}$ with $n$ and $k$ sufficiently large, and  
\begin{align}
    \sum_{i=1}^n C_t(P_{Y|X}^{(i)} \| Q_{Y|X}^{(i)}) - k R_d(P_S) &\geq    \sqrt{\operatorname{Var}(U_{k, n}) } Q^{-1}(\epsilon) +  \left(\overline{c} - \frac{1}{2} \right)  \log (k) + c_1 + c_2,   \label{bv4.35}
\end{align}
then the excess-distortion probability $\overline{P}_e(d, n, k) \leq \epsilon$. The constant $c$ in $(\ref{thm_cond})$ can be taken to be $c = c_1 + c_2$.

\section{\label{bec_dec_proof}}

We consider the setting of block erasure channels from Section \ref{tiredaf}. Let $Q_{X^n}^* = \operatorname{Unif}( \mathcal{X}^n)$ and let $Q^{\sim t}_{X^n|Y^n}$ be as given in $(\ref{kernelforbec})$. For the $(P_{Z^k}^*, Q_{X^n}^*, Q^{\sim t}_{X^n|Y^n})$ source-channel code, the argmax rule in $(\ref{unif_dec})$ particularizes to 
\begin{align}
    &\argmax_i \left [ \log\left( \frac{Q^{\sim t}_{X^n|Y^n}(\bar{X}_i^n|y^n)}{Q_{X^n}^*(\bar{X}_i^n)} \right) - \log(T_i) \right ]\\
    &= \argmax_i \left [  -\log(T_i)  + \sum_{j=1}^n L_j(\bar{X}_{i, j}, y_j) \right], \label{gfndjfn}
\end{align}
where 
\begin{align*}
    L_j(x, y) \coloneqq \log\left(\frac{Q_{Y|X}^{(j)}(y|x)^t}{\sum_{u \in \mathcal{X}} 2^{-m} Q_{Y|X}^{(j)}(y|u)^t }\right) &= \begin{cases}
         0 & y = e\\
         m \log(2) & y \in \mathcal{X}, y = x\\
         -\infty & y \in \mathcal{X}, y \neq x.\\
     \end{cases}
\end{align*}
Hence, $(\ref{gfndjfn})$ simplifies to  
\begin{align*}
    &\argmax_{i: \bar{X}_{i, j} = y_j \forall j \text{ with } y_j \neq e} \left [ -\log(T_i)  + \sum_{j=1}^n  L_j(\bar{X}_{i, j}, y_j) \right ]\\
    &= \argmax_{i: \bar{X}_{i, j} = y_j \forall j \text{ with } y_j \neq e} \left [ -\log(T_i)  + N(y^n) m \log(2) \right ]\\
    &= \argmin_{i: \bar{X}_{i, j} = y_j \forall j \text{ with } y_j \neq e} T_i, 
\end{align*}
where $N(y^n) \coloneqq \#\{j:y_j \neq e \}$.

\section{Proof of Theorem \ref{BEC_particularization} \label{BEC_particularization_proof}}

Consider a sequence of channel pairs $\left(P_{Y|X}^{(i)}, Q_{Y|X}^{(i)} \right)_{i \geq 1}$ where $\left(P_{Y|X}^{(i)} \right)_{i \geq 1}$ is as specified in the statement of Theorem \ref{BEC_particularization} and 
\begin{align}
    Q^{(i)}_{Y|X}(y|x) &= \begin{cases}
        1/2 & y = e\\
        1/2 & y = x
    \end{cases}
\end{align}
for all $i$. Furthermore, consider the channel-blind family $\mathcal{U}^{(d)} = \{ U_{n, k}^{(d)}\}$ of source-channel codes from Proposition \ref{becuniversaltheorem}, where 
$U_{n, k}^{(d)}$ is the $(t, Q_{Y^n|X^n}, n, k)$-based source-channel code where we fix an arbitrary value of $t > 0$. We intend to apply Theorem \ref{thm_finite_block1} for this sequence of channel pairs and for the channel-blind and admissible family of source-channel codes given by 
\begin{align}
    \mathcal{U}_{\operatorname{S}}^{(d)} \coloneqq \left \{ U_{n, k}^{(d)} : (n, k) \in \operatorname{S} \right \},
\end{align}
where $\operatorname{S}$ is an admissible source-channel blocklength set to be specified later. 

It can be checked that this sequence of channel pairs satisfies the hypotheses of Theorem \ref{thm_finite_block1}, specifically Assumptions \ref{assumpr1_channels}\,-\, \ref{assumpr3_channels} and $C_t(P_{Y|X}^{(i)} \| Q_{Y|X}^{(i)}) > 0$. The latter holds because $C_t(P_{Y|X}^{(i)} \| Q_{Y|X}^{(i)}) = C(P_{Y|X}^{(i)}) = (1-p_i) \log(2^m) > 0$ since $p_i < 1$ by assumption.
We also need to check that $(\ref{assumpbe1})$ and $(\ref{assumpbe2})$ are satisfied for all $(n, k) \in \operatorname{S}$ in order to apply Theorem \ref{thm_finite_block1}. First note that for all $t > 0$, 
$C_t(P_{Y|X}^{(i)} \| Q_{Y|X}^{(i)}) = C(P_{Y|X}^{(i)})$ and $V_t(P_{Y|X}^{(i)} \| Q_{Y|X}^{(i)}) = V(P_{Y|X}^{(i)})$. Now define $L : \mathcal{X} \times \mathcal{Y} \to [-\infty, +\infty )$ as 
\begin{align*}
    L(x, y) &= \begin{cases}
        \log (2^m) & y \neq e \text{ and } x = y,\\
        -\infty & y \neq e \text{ and } x \neq y,\\
        0 & y = e.
    \end{cases}
\end{align*}
Then the LHS of $(\ref{assumpbe1})$ from Theorem \ref{thm_finite_block1} simplifies to  
    \begin{align}
        \frac{1}{n + k} \left [ k \mathbb{E}_{P_S}\left[ \big | \jmath_d(S, P_S, \operatorname{d}) - R_d(P_S) \big |^3 \right] +  \sum_{i=1}^n \mathbb{E}_{Q_{X} \times P_{Y|X}^{(i)}}\left [  \Big | L(X_i, Y_i) - C(P_{Y|X}^{(i)}) \Big|^3 \right] \right], \label{tiredforbec,,}
\end{align}
where $Q_X = \operatorname{Unif}( \mathcal{X})$. Then there exists a 
constant $\kappa_1$ such that
$(\ref{tiredforbec,,})$ is upper bounded by $\kappa_1$ for all $(n, k)$ because of the argument given in the second remark after Remark \ref{secremarkafter6} and because 
\begin{align*}
    \sup_{p \in [0, 1]} \mathbb{E}_{Q_{X} \times P_{Y|X}}\left [  \Big | L(X, Y) - C(P_{Y|X}) \Big|^3 \right] < \infty
\end{align*}
for $P_{Y|X} = \operatorname{BLEC}(2^m, p)$. Note that $L(X, Y)$ is finite with probability one whenever $(X, Y) \sim Q_X \times P_{Y|X}$. Similarly, the LHS of condition $(\ref{assumpbe2})$ simplifies to 
\begin{align}
    \frac{1}{n + k}\left [ k V_d(P_S) +  \sum_{i=1}^n V(P_{Y|X}^{(i)})  \right], 
\end{align}
which is uniformly lower bounded by some constant $\kappa_2 > 0$  for all $(n, k) \in \operatorname{S}$ with $n \geq n_0$ due to the assumptions made in the statement of Theorem \ref{BEC_particularization}. Hence, we now obtain from Theorem \ref{thm_finite_block1} that
there exist constants $c$, $K_0$ and $N_0$ such that for every $(n, k) \in \operatorname{S}$ with $n \geq N_0$ and $k \geq K_0$, the source-channel code $U_{n, k}^{(d)} \in \mathcal{U}_{\operatorname{S}}^{(d)}$ is $(d, \epsilon, n, k)$  with respect to $\{P_{S^k}, \operatorname{d}_k, P_{Y^n|X^n} \}$ if 
\begin{align}
    \sum_{i=1}^n C(P_{Y|X}^{(i)}) - k R_d(P_S) &\geq  \sqrt{k V_d(P_S) + \sum_{i=1}^n V(P_{Y|X}^{(i)} )} \,\,Q^{-1}(\epsilon) + \left(\overline{c} - \frac{1}{2}\right) \log(k) + c. \label{n350c}
\end{align}

We now specify an admissible source-channel blocklength set $\operatorname{S}$ for the above result. Define 
\begin{align}
    \widehat{k}_n &=  \left \lfloor \frac{1}{R_d(P_S)} \sum_{i=1}^n C(P_{Y|X}^{(i)}) - \frac{Q^{-1}(\epsilon)}{R_d(P_S)} \sqrt{\frac{V_d(P_S)}{R_d(P_S)} \sum_{i=1}^n C(P_{Y|X}^{(i)}) + \sum_{i=1}^n V(P_{Y|X}^{(i)})} - \frac{\overline{c}-1/2}{R_d(P_S)} \log(n) - \frac{\tilde{c}}{R_d(P_S)} \right \rfloor, \label{choiceofkn}
\end{align}
where 
\begin{align*}
    \tilde{c} \coloneqq (\overline{c}-1/2)\log\left( \frac{2 \log(2^m)}{R_d(P_S)} \right) + c. 
\end{align*}
Then since $\sum_{i=1}^n C(P_{Y|X}^{(i)}) = \Theta(n)$ and
$$\frac{V_d(P_S)}{R_d(P_S)} \sum_{i=1}^n C(P_{Y|X}^{(i)}) + \sum_{i=1}^n V(P_{Y|X}^{(i)}) = \Theta(n)$$
by assumption, we have $\frac{\widehat{k}_n}{n} = \Theta(1)$. Hence, there exists an $N > 0$ such that $\operatorname{S} = \left \{ (n, \widehat{k}_n  ) : n \geq N \right \}$ is a valid (in the sense of Definition \ref{admissible_def}) admissible source-channel blocklength set for the result $(\ref{n350c})$. Next, we show that for sufficiently large $n$, inequality 
$(\ref{n350c})$ is satisfied for $k = \widehat{k}_n$. Indeed, the LHS of $(\ref{n350c})$ evaluated at $k = \widehat{k}_n$ is lower bounded as  
\begin{align}
    &\sum_{i=1}^n C(P_{Y|X}^{(i)}) - \widehat{k}_n R_d(P_S) \notag \\
    &\geq Q^{-1}(\epsilon) \sqrt{\frac{V_d(P_S)}{R_d(P_S)} \sum_{i=1}^n C(P_{Y|X}^{(i)}) + \sum_{i=1}^n V(P_{Y|X}^{(i)})} + (\overline{c}-1/2) \log(n) + \tilde{c}.
\end{align}
The RHS of $(\ref{n350c})$ evaluated at $k = \widehat{k}_n$ is 
\begin{align}
    Q^{-1}(\epsilon)\sqrt{\widehat{k}_n V_d(P_S) + \sum_{i=1}^n V(P_{Y|X}^{(i)} )}  + (\overline{c}-1/2) \log(\widehat{k}_n) + c. \label{c2568}
\end{align}
Hence, it suffices to show that for sufficiently large $n$, 
\begin{align}
    &Q^{-1}(\epsilon) \sqrt{\frac{V_d(P_S)}{R_d(P_S)} \sum_{i=1}^n C(P_{Y|X}^{(i)}) + \sum_{i=1}^n V(P_{Y|X}^{(i)})}  + (\overline{c}-1/2)\log(n) +  \tilde{c} \notag  \\
    &\quad \quad \quad \quad \geq Q^{-1}(\epsilon)\sqrt{\widehat{k}_n V_d(P_S) + \sum_{i=1}^n V(P_{Y|X}^{(i)} )} + (\overline{c}-1/2) \log(\widehat{k}_n) + c. \label{finalsufficientkn}
\end{align}
\textit{Case 1:} If $\epsilon \leq 1/2$, then $(\ref{choiceofkn})$ implies that there exists a constant $N_0'$ such that for $n \geq N_0'$, 
\begin{align}
    \widehat{k}_n \leq  \frac{1}{R_d(P_S)} \sum_{i=1}^n C(P_{Y|X}^{(i)}).   \label{sjhafgdhaf}
\end{align}
Then the RHS of $(\ref{finalsufficientkn})$ can be upper bounded by 
\begin{align*}
    &Q^{-1}(\epsilon) \sqrt{\frac{V_d(P_S)}{R_d(P_S)} \sum_{i=1}^n C(P_{Y|X}^{(i)}) + \sum_{i=1}^n V(P_{Y|X}^{(i)})} + (\overline{c}-1/2) \log\left(\frac{1}{R_d(P_S)} \sum_{i=1}^n C(P_{Y|X}^{(i)}) \right) + c\\
    &\leq Q^{-1}(\epsilon) \sqrt{\frac{V_d(P_S)}{R_d(P_S)} \sum_{i=1}^n C(P_{Y|X}^{(i)}) + \sum_{i=1}^n V(P_{Y|X}^{(i)})} + (\overline{c}-1/2) \log\left(\frac{n \log(2^m)}{R_d(P_S)}  \right) + c.
\end{align*}
Therefore, $(\ref{finalsufficientkn})$ holds for $n \geq N_0'$.

\textit{Case 2:} For $\epsilon > 1/2$, $(\ref{choiceofkn})$ implies that there exists $N_1'$ such that for all $n \geq N_1'$, 
\begin{align}
    \widehat{k}_n \geq  \frac{1}{R_d(P_S)} \sum_{i=1}^n C(P_{Y|X}^{(i)})   \label{sjha2fgdhaf}
\end{align}
and 
\begin{align}
    \widehat{k}_n \leq  \frac{2n \log(2^m)}{R_d(P_S)}.    \label{sjnidadhaf}
\end{align}
Using $(\ref{sjha2fgdhaf})$ and $(\ref{sjnidadhaf})$, the RHS of $(\ref{finalsufficientkn})$ can be upper bounded by
\begin{align*}
    Q^{-1}(\epsilon) \sqrt{\frac{V_d(P_S)}{R_d(P_S)} \sum_{i=1}^n C(P_{Y|X}^{(i)}) + \sum_{i=1}^n V(P_{Y|X}^{(i)})} + (\overline{c} -1/2)\log\left(\frac{2n \log(2^m)}{R_d(P_S)}  \right) + c.
\end{align*}
Hence, $(\ref{finalsufficientkn})$ holds for $n \geq N_1'$.  

Combining both cases, we have that for every $\epsilon \in (0, 1)$,
\begin{align*}
    k^*_{\mathcal{U}^{(d)}}(d, \epsilon, n|P_n) &\geq \frac{1}{R_d(P_S)} \sum_{i=1}^n C(P_{Y|X}^{(i)}) - \frac{Q^{-1}(\epsilon)}{R_d(P_S)} \sqrt{\frac{V_d(P_S)}{R_d(P_S)} \sum_{i=1}^n C(P_{Y|X}^{(i)}) + \sum_{i=1}^n V(P_{Y|X}^{(i)})} - \mbox{} \notag \\
    & \quad \quad \quad \quad \quad \quad \quad \quad \quad \quad \frac{\overline{c}-1/2}{R_d(P_S)} \log(n) - \frac{\tilde{c}}{R_d(P_S)} 
\end{align*}
for sufficiently large $n$.

\section{Proof of Theorem \ref{convoghlswhr0vm} \label{convoghlswhr0vm_proof}}

  Let $P_{\overline{Y}^n} = P_{\overline{Y}_1} \times \cdots \times P_{\overline{Y}_n}$ be given as 
\begin{align}
    P_{\overline{Y}_i}(y_i) &= \begin{cases}
        p_i & y_i = e,\\
        \frac{1 -p_i}{2^m} & y_i \in \mathcal{X}. 
    \end{cases}
    \label{2p2yi}
\end{align}
Then when $Y^n \sim P_{Y^n|X^n = x^n}$, we have 
\begin{align}
    \imath_{P_{Y^n|X^n}, P_{\overline{Y}^n} }(x^n, Y^n) &= \log(2^m) \sum_{i=1}^n  \mathds{1}\{ Y_i \neq e \}, \label{dn1}
\end{align}
where $\{Y_i \neq e \} \sim \text{Bernoulli}(1-p_i)$ independently across $i$. Since the distribution of $(\ref{dn1})$ is independent of $x^n$, we obtain from \cite[Theorem 2]{kostina_JSCC} that if a $(d, \epsilon, n, k)$ source-channel code exists, then 
\begin{align*}
    \epsilon \geq \sup_{\gamma > 0} \left \{ \mathbb{P}\left( \sum_{i=1}^k \jmath_d(S_i, P_S, \operatorname{d}) - \log(2^m) \sum_{i=1}^n  \mathds{1}\{ Y_i \neq e \}  \geq \gamma \right) - e^{-\gamma} \right \}
\end{align*}
where $(S^k, Y^n) \sim P_S^{\otimes k} P_{Y^n|X^n = x^n} $ for an arbitrary $x^n$. Since there is no dependence on $x^n$, the preceding necessary condition is   
\begin{align*}
 \mathbb{P}\left( \sum_{i=1}^k \jmath_d(S_i, P_S, \operatorname{d}) - \log(2^m) \sum_{i=1}^n  B_i  \geq \gamma \right) \leq \epsilon + e^{-\gamma} \quad  \text{ for every } \gamma > 0,
\end{align*}
where $B_1, \ldots, B_n$ are independent and each $B_i \sim \text{Bern}(1 - p_i)$. We first center the source and channel random variables so that the preceding necessary condition is equivalently written as
\begin{align}
    &\mathbb{P}\Bigg(
        \sum_{j=1}^k
        \left(
            \jmath_d(S_j,P_S,\operatorname{d})
            -
            R_d(P_S)
        \right)
        -
        \log(2^m)
        \sum_{i=1}^n
        \left(
            B_i-(1-p_i)
        \right)
        \nonumber\\
    &\hspace{35mm}\geq
        \log(2^m)\sum_{i=1}^n(1-p_i)
        -
        kR_d(P_S)
        +
        \gamma
    \Bigg)
    \leq
    \epsilon+e^{-\gamma}
    \qquad \text{for every } \gamma>0.
    \label{nonstationary_bec_centered_converse}
\end{align}
Define the combined variance
\begin{align}
    \sigma_{n,k}^2
    &:=
    kV_d(P_S)
    +
    \bigl(\log(2^m)\bigr)^2
    \sum_{i=1}^n p_i(1-p_i).
    \label{nonstationary_bec_combined_variance}
\end{align}
This is the variance of the zero-mean random variable on the left-hand side of the inequality
inside the probability in \eqref{nonstationary_bec_centered_converse}. Recall that Condition $(\ref{weird9thmomcond})$ implies that \cite[(339)]{kostina_JSCC}
\begin{align*}
    \mathbb{E}\left[
        \left|
            \jmath_d(S,P_S,\operatorname{d})
            -
            R_d(P_S)
        \right|^3
    \right]
    <\infty.
\end{align*}
Consequently, there exists
a finite constant $K_{\mathrm{S}}$ such that
\begin{align}
    \mathbb{E}\left[
        \left|
            \jmath_d(S,P_S,\operatorname{d})
            -
            R_d(P_S)
        \right|^3
    \right]
    \leq
    K_{\mathrm{S}}V_d(P_S).
    \label{source_third_moment_variance_bound}
\end{align}
Indeed, if $V_d(P_S)>0$, then
$K_{\mathrm{S}}$ can be taken to be the ratio of the two
quantities in \eqref{source_third_moment_variance_bound}.
If $V_d(P_S)=0$, then
$\jmath_d(S,P_S,\operatorname{d})=R_d(P_S)$ almost surely,
and hence the left-hand side of
\eqref{source_third_moment_variance_bound} is zero.

For the Bernoulli summands in $(\ref{nonstationary_bec_centered_converse})$, for each $i$, we have 
\begin{align}
    &\mathbb{E}\left[
        \left|
            \log(2^m)
            \left(
                B_i-(1-p_i)
            \right)
        \right|^3
    \right]
    \nonumber\\
    &\qquad =
    \bigl(\log(2^m)\bigr)^3
    p_i(1-p_i)
    \left(
        p_i^2+(1-p_i)^2
    \right)
    \nonumber\\
    &\qquad \leq
    \bigl(\log(2^m)\bigr)^3p_i(1-p_i).
    \label{bec_third_absolute_moment}
\end{align}
It follows from
\eqref{source_third_moment_variance_bound} and
\eqref{bec_third_absolute_moment} that there exists a finite
constant $K_3$, independent of $n$, $k$ and
$(p_1,\ldots,p_n)$, such that
\begin{align}
    &k\,
    \mathbb{E}\left[
        \left|
            \jmath_d(S,P_S,\operatorname{d})
            -
            R_d(P_S)
        \right|^3
    \right]
    \nonumber\\
    &\quad+
    \bigl(\log(2^m)\bigr)^3
    \sum_{i=1}^n
    \mathbb{E}\left[
        \left|
            B_i-(1-p_i)
        \right|^3
    \right]
    \leq
    K_3\sigma_{n,k}^2.
    \label{aggregate_third_moment_bound}
\end{align}
Therefore, the Berry--Esseen theorem for independent,
not necessarily identically distributed, random variables
implies that there exists a finite constant
$K_{\mathrm{BE}}$, independent of $n$, $k$ and
$(p_1,\ldots,p_n)$, such that, whenever
$\sigma_{n,k}>0$,
\begin{align}
    &\mathbb{P}\Bigg(
        \sum_{j=1}^k
        \left(
            \jmath_d(S_j,P_S,\operatorname{d})
            -
            R_d(P_S)
        \right)
        -
        \log(2^m)
        \sum_{i=1}^n
        \left(
            B_i-(1-p_i)
        \right)
        >
        t
    \Bigg)
    \nonumber\\
    &\qquad\geq
    Q\left(
        \frac{t}{\sigma_{n,k}}
    \right)
    -
    \frac{K_{\mathrm{BE}}}{\sigma_{n,k}}
    \qquad \text{for every } t\in\mathbb{R}.
    \label{nonstationary_bec_berry_esseen}
\end{align}
We now choose
\begin{align}
    \gamma
    =
    \frac{1}{2}\log k.
    \label{nonstationary_bec_gamma_choice}
\end{align}
For $k>1$, this choice is admissible, and
\eqref{nonstationary_bec_centered_converse} implies the following: 
\begin{align}
    &\mathbb{P}\Bigg(
        \sum_{j=1}^k
        \left(
            \jmath_d(S_j,P_S,\operatorname{d})
            -
            R_d(P_S)
        \right)
        -
        \log(2^m)
        \sum_{i=1}^n
        \left(
            B_i-(1-p_i)
        \right)
        \nonumber\\
    &\hspace{31mm}\geq
        \log(2^m)\sum_{i=1}^n(1-p_i)
        -
        kR_d(P_S)
        +
        \frac{1}{2}\log k
    \Bigg)
    \leq
    \epsilon+\frac{1}{\sqrt{k}}.
    \label{nonstationary_bec_converse_gamma_choice}
\end{align}
Let $\sigma_0>0$ and $k_0 \geq 2$ be constants sufficiently large such that
\begin{align}
    \frac{K_{\mathrm{BE}}}{\sigma_0}
    \leq
    \frac{1-\epsilon}{4},
    \label{sigma_zero_choice}
\end{align}
and 
\begin{align}
    \frac{1}{\sqrt{k}}
    \leq
    \frac{1-\epsilon}{4}
    \qquad \text{for every } k\geq k_0.
    \label{k_zero_choice}
\end{align}

The goal is to show that $(\ref{nonstationary_bec_converse_gamma_choice})$ implies  
$(\ref{eq:normal-approximation-converse})$ so that $(\ref{eq:normal-approximation-converse})$ is a necessary condition for a code being $(d, \epsilon, n, k)$. We show this for the cases of bounded and unbounded
$\sigma_{n,k}$ separately.

\emph{Case 1: $\sigma_{n,k}\leq\sigma_0$.}
First suppose that $\sigma_{n,k}=0$. The centered random
variable in \eqref{nonstationary_bec_converse_gamma_choice}
is then equal to zero almost surely. Since
$\epsilon+1/\sqrt{k}<1$ for $k\geq k_0$,
\eqref{nonstationary_bec_converse_gamma_choice} implies
\begin{align}
    \log(2^m)\sum_{i=1}^n(1-p_i)
    -
    kR_d(P_S)
    +
    \frac{1}{2}\log k
    >
    0.
    \label{zero_variance_converse}
\end{align}
Now suppose that $0<\sigma_{n,k}\leq\sigma_0$, and define
\begin{align}
    a_\epsilon
    :=
    \sqrt{\frac{1+\epsilon}{1-\epsilon}}.
    \label{cantelli_constant}
\end{align}
By Cantelli's inequality,
\begin{align}
    &\mathbb{P}\Bigg(
        \sum_{j=1}^k
        \left(
            \jmath_d(S_j,P_S,\operatorname{d})
            -
            R_d(P_S)
        \right)
        -
        \log(2^m)
        \sum_{i=1}^n
        \left(
            B_i-(1-p_i)
        \right)
        \geq
        -a_\epsilon\sigma_{n,k}
    \Bigg)
    \nonumber\\
    &\qquad\geq
    \frac{a_\epsilon^2}{1+a_\epsilon^2} \notag \\
    & \qquad =
    \frac{1+\epsilon}{2} \geq  \epsilon + \frac{1}{\sqrt{k}},
\label{nonstationary_bec_cantelli}
\end{align}
where the last inequality holds for $k \geq k_0$ since $\epsilon \in (0, 1)$ and $k_0$ was chosen to satisfy \eqref{k_zero_choice}. Then given that \eqref{nonstationary_bec_converse_gamma_choice} holds and given 
\eqref{nonstationary_bec_cantelli}, it follows that 
\begin{align}
    \log(2^m)\sum_{i=1}^n(1-p_i)
    -
    kR_d(P_S)
    +
    \frac{1}{2}\log k
    >
    -a_\epsilon\sigma_{n,k}.
    \label{bounded_variance_coarse_converse}
\end{align}
Thus, for every $\sigma_{n,k}\leq\sigma_0$, a  necessary condition for a $(d, \epsilon, n, k)$ code is 
\begin{align}
    \log(2^m)\sum_{i=1}^n(1-p_i)
    -
    kR_d(P_S)
    &\geq -a_\epsilon \sigma_{n, k} - \frac{1}{2} \log(k) + \sigma_{n, k} Q^{-1}(\epsilon) - \sigma_0 |Q^{-1}(\epsilon)| \notag \\
    &\geq 
    \sigma_{n,k}Q^{-1}(\epsilon)
    -
    \frac{1}{2}\log k
    -
    K_{\mathrm{small}},
    \label{bounded_variance_final_converse}
\end{align}
where
\begin{align}
    K_{\mathrm{small}}
    :=
    \left(
        a_\epsilon+
        \left|Q^{-1}(\epsilon)\right|
    \right)\sigma_0.
    \label{bounded_variance_constant}
\end{align}
Hence, we have the desired result $(\ref{eq:normal-approximation-converse})$ in this case. 

\emph{Case 2: $\sigma_{n,k}>\sigma_0$.}
Applying \eqref{nonstationary_bec_berry_esseen} with
\begin{align*}
    t
    =
    \log(2^m)\sum_{i=1}^n(1-p_i)
    -
    kR_d(P_S)
    +
    \frac{1}{2}\log k
\end{align*}
and then using
\eqref{nonstationary_bec_converse_gamma_choice}, we obtain
\begin{align}
    Q\Bigg(
        \frac{
            \log(2^m)\sum_{i=1}^n(1-p_i)
            -
            kR_d(P_S)
            +
            \frac{1}{2}\log k
        }{
            \sigma_{n,k}
        }
    \Bigg)
    &\leq
    \epsilon
    +
    \frac{1}{\sqrt{k}}
    +
    \frac{K_{\mathrm{BE}}}{\sigma_{n,k}} \notag \\
   \implies \quad          \log(2^m)\sum_{i=1}^n(1-p_i)
            -
            kR_d(P_S)
         &\geq \sigma_{n, k} Q^{-1} \left( \epsilon  +
    \frac{1}{\sqrt{k}}
    +
    \frac{K_{\mathrm{BE}}}{\sigma_{n,k}} \right) - \frac{1}{2}\log k. \label{normal_tail_upper_bound}
\end{align}
By \eqref{sigma_zero_choice} and \eqref{k_zero_choice}.
\begin{align}
    0 < \frac{1}{\sqrt{k}}
    +
    \frac{K_{\mathrm{BE}}}{\sigma_{n,k}}
    \leq
    \frac{1-\epsilon}{2}.
    \label{normal_approximation_error_range}
\end{align}
Let $\phi$ denote the standard Gaussian density and define
\begin{align}
    L_\epsilon
    :=
    \max_{\epsilon\leq u\leq(1+\epsilon)/2}
    \frac{1}{
        \phi\left(Q^{-1}(u)\right)
    }.
    \label{inverse_Q_lipschitz_constant}
\end{align}
This is a finite constant. By the mean value theorem, for
every $\eta\in[0,(1-\epsilon)/2]$,
\begin{align}
    Q^{-1}(\epsilon+\eta)
    \geq
    Q^{-1}(\epsilon)
    -
    L_\epsilon\eta.
    \label{inverse_Q_mean_value_bound}
\end{align}
Since $Q$ is decreasing, it follows from
\eqref{normal_tail_upper_bound}--\eqref{inverse_Q_mean_value_bound}
that
\begin{align}
    &\log(2^m)\sum_{i=1}^n(1-p_i)
    -
    kR_d(P_S)
    \nonumber\\
    &\qquad\geq
    \sigma_{n,k}Q^{-1}(\epsilon)
    -
    \frac{1}{2}\log k
    -  
    \left(
        \frac{\sigma_{n,k}}{\sqrt{k}}
        +
        K_{\mathrm{BE}}
    \right) L_\epsilon.
    \label{berry_esseen_preliminary_converse}
\end{align}
It remains to show that $(\ref{berry_esseen_preliminary_converse})$ implies $(\ref{eq:normal-approximation-converse})$. We show this separately for two subcases of \textit{Case 2:} $\sigma_{n, k} > \sigma_0$. Let $D_\epsilon>0$ be a constant sufficiently large such that,
for every $x\geq D_\epsilon$,
\begin{align}
    x^2
    -
    \left(
        V_d(P_S)
        +
        \log(2^m)R_d(P_S)
    \right)
    \geq
    \log(2^m)
    \left|Q^{-1}(\epsilon)\right|x.
    \label{large_variance_threshold_choice}
\end{align}
Such a constant exists because the left-hand side grows
quadratically in $x$, whereas the right-hand side grows
linearly.

If $\sigma_{n,k}
    \leq
    D_\epsilon\sqrt{k},$ then \eqref{berry_esseen_preliminary_converse} gives
\begin{align}
    &\log(2^m)\sum_{i=1}^n(1-p_i)
    -
    kR_d(P_S)
    \nonumber\\
    &\qquad\geq
    \sigma_{n,k}Q^{-1}(\epsilon)
    -
    \frac{1}{2}\log k
    -
    L_\epsilon
    \left(
        D_\epsilon+K_{\mathrm{BE}}
    \right),
    \label{moderate_variance_final_converse}
\end{align}
giving us the desired result. We now consider the case $\sigma_{n,k}
    >
    D_\epsilon\sqrt{k}$ in the rest of the proof. Using $p_i(1-p_i)\leq 1-p_i$, we have
\begin{align}
    \sigma_{n,k}^2
    &=
    kV_d(P_S)
    +
    \bigl(\log(2^m)\bigr)^2
    \sum_{i=1}^n p_i(1-p_i)
    \nonumber\\
    &\leq
    kV_d(P_S)
    +
    \bigl(\log(2^m)\bigr)^2
    \sum_{i=1}^n(1-p_i)
    \nonumber\\
    &=
    kV_d(P_S)
    +
    \log(2^m)
    \left(
        \log(2^m)\sum_{i=1}^n(1-p_i)
    \right).
    \label{dispersion_capacity_relation}
\end{align}
Rearranging \eqref{dispersion_capacity_relation} gives
\begin{align}
    &\log(2^m)\sum_{i=1}^n(1-p_i)
    -
    kR_d(P_S)
    \nonumber\\
    &\qquad\geq
    \frac{
        \sigma_{n,k}^2
        -
        k\left(
            V_d(P_S)
            +
            \log(2^m)R_d(P_S)
        \right)
    }{
        \log(2^m)
    }.
    \label{large_variance_capacity_surplus}
\end{align}
By the choice of
$D_\epsilon$ in
\eqref{large_variance_threshold_choice}, the right-hand side
of \eqref{large_variance_capacity_surplus} is lower bounded as
\begin{align}
    &\frac{k}{\log(2^m)}
    \Bigg[
        \left(
            \frac{\sigma_{n,k}}{\sqrt{k}}
        \right)^2
        -
        \left(
            V_d(P_S)
            +
            \log(2^m)R_d(P_S)
        \right)
    \Bigg]
    \nonumber\\
    &\qquad\geq
    k\left|Q^{-1}(\epsilon)\right|
    \frac{\sigma_{n,k}}{\sqrt{k}}
    \nonumber\\
    &\qquad=
    \sqrt{k}
    \left|Q^{-1}(\epsilon)\right|
    \sigma_{n,k}
    \nonumber\\
    &\qquad\geq
    Q^{-1}(\epsilon)\sigma_{n,k}.
    \label{large_variance_direct_converse}
\end{align}
Consequently, the desired bound also holds in this case.

Combining
\eqref{bounded_variance_final_converse},
\eqref{moderate_variance_final_converse} and
\eqref{large_variance_direct_converse}, there exists a finite
constant
\begin{align}
    K
    :=
    \max\left\{
        K_{\mathrm{small}},
        L_\epsilon
        \left(
            D_\epsilon+K_{\mathrm{BE}}
        \right)
    \right\}
    \label{nonstationary_bec_final_constant}
\end{align}
such that, for every $n\geq 1$ and every $k\geq k_0$,
the existence of a $(d,\epsilon,n,k)$ source-channel code
implies
\begin{align}
    &\log(2^m)\sum_{i=1}^n(1-p_i)
    -
    kR_d(P_S)
    \nonumber\\
    &\qquad\geq
    \sqrt{
        kV_d(P_S)
        +
        \bigl(\log(2^m)\bigr)^2
        \sum_{i=1}^n p_i(1-p_i)
    }
    \,Q^{-1}(\epsilon)
    -
    \frac{1}{2}\log k
    -
    K.
    \label{nonstationary_bec_final_converse}
\end{align}

\section{Proof of Corollary \ref{corocoro} \label{corocoro_proof}}

Recall that we have $0<R_d(P_S)<\infty$ from Condition $(\ref{cond1dee})$. Define
\begin{align}
    \mathsf{B}_n
    &:=
    \frac{V_d(P_S)}{R_d(P_S)}
    \sum_{i=1}^n C\left(P_{Y|X}^{(i)}\right)
    +
    \sum_{i=1}^n V\left(P_{Y|X}^{(i)}\right).
    \label{eq:combined-variance-nonstationary-bec}
\end{align}
For a block erasure channel
$P_{Y|X}^{(i)}=\operatorname{BLEC}(2^m,p_i)$, we have
\begin{align}
    C\left(P_{Y|X}^{(i)}\right)
    &=
    (1-p_i)\log(2^m),
    \label{eq:bec-capacity-corollary}\\
    V\left(P_{Y|X}^{(i)}\right)
    &=
    p_i(1-p_i)\bigl(\log(2^m)\bigr)^2.
    \label{eq:bec-dispersion-corollary}
\end{align}
In particular,
\begin{align}
    0
    &\leq
    C\left(P_{Y|X}^{(i)}\right)
    \leq
    \log(2^m),
    \label{eq:bec-capacity-uniform-upper-bound}\\
    0
    &\leq
    V\left(P_{Y|X}^{(i)}\right)
    \leq
    \frac{1}{4}\bigl(\log(2^m)\bigr)^2.
    \label{eq:bec-dispersion-uniform-upper-bound}
\end{align}
Hence, by the assumptions on $\mathbf{P}$,
we have that there exist
constants $0<b_- \leq b_+<\infty$ such that, for all sufficiently
large $n$,
\begin{align}
    b_-n
    \leq
    \mathsf{B}_n
    \leq
    b_+n.
    \label{eq:combined-variance-linear-bounds}
\end{align}

It suffices to prove that if a $(d, \epsilon, n, k)$ source-channel code exists and $n \geq n_1$, then  
\begin{align}
    k
    &\leq
    \frac{1}{R_d(P_S)}
    \sum_{i=1}^n C\left(P_{Y|X}^{(i)}\right)
    -
    \frac{Q^{-1}(\epsilon)}{R_d(P_S)}
    \sqrt{\mathsf{B}_n}
    +
    \frac{1}{2R_d(P_S)}\log n + \widetilde{K}, \label{desired_uppbndjgkg4}
\end{align}
where $n_1$ and $\widetilde{K}$ are some positive constants independent of $n$ and $k$. The following are the trivial cases that we can rule out at the outset: 
\begin{itemize}
    \item If $k < k_0$, where $k_0$ is the constant from Theorem \ref{convoghlswhr0vm}, then   \eqref{desired_uppbndjgkg4} holds for sufficiently large $n \geq n_1$ for some $n_1$ because the RHS of \eqref{desired_uppbndjgkg4} grows linearly with $n$.
    \item If 
    \begin{align}
    &\sum_{i=1}^n C\left(P_{Y|X}^{(i)}\right)
    -
    kR_d(P_S) \geq
    Q^{-1}(\epsilon)\sqrt{\mathsf{B}_n}
    -
    \frac{1}{2}\log n,
    \label{eq:kstar-first-case}
\end{align}
then \eqref{desired_uppbndjgkg4} again holds since $\widetilde{K} > 0$. 
\end{itemize}

Hence, in the remainder of the proof, we assume $k\geq k_0$ and 
\begin{align}
    &\sum_{i=1}^n C\left(P_{Y|X}^{(i)}\right)
    -
    kR_d(P_S)
    <
    Q^{-1}(\epsilon)\sqrt{\mathsf{B}_n}
    -
    \frac{1}{2}\log n.
    \label{eq:kstar-second-case}
\end{align}
Then if a $(d,\epsilon,n,k)$ source-channel code exists, then Theorem \ref{convoghlswhr0vm} implies that 
\begin{align}
    &\sum_{i=1}^n C\left(P_{Y|X}^{(i)}\right)
    -
    kR_d(P_S)
    \nonumber\\
    &\qquad\geq
    Q^{-1}(\epsilon)
    \sqrt{
        kV_d(P_S)
        +
        \sum_{i=1}^n V\left(P_{Y|X}^{(i)}\right)
    }
    -
    \frac{1}{2}\log k
    -
    K.
    \label{eq:implicit-kstar-converse}
\end{align}
It now suffices to show that for sufficiently large $n \geq n_1$,  $(\ref{eq:implicit-kstar-converse})$ implies \eqref{desired_uppbndjgkg4}.

We first note that $(\ref{eq:implicit-kstar-converse})$ implies that the source blocklength $k$ can be at most linear in $n$. Indeed, rearranging
\eqref{eq:implicit-kstar-converse} and using
\eqref{eq:bec-capacity-uniform-upper-bound} and
\eqref{eq:bec-dispersion-uniform-upper-bound} give
\begin{align}
    kR_d(P_S)
    &\leq
    n\log(2^m)
    +
    \left|Q^{-1}(\epsilon)\right|
    \sqrt{
        kV_d(P_S)
        +
        \frac{n}{4}\bigl(\log(2^m)\bigr)^2
    }
    \nonumber\\
    &\quad
    +
    \frac{1}{2}\log k
    +
    K.
    \label{eq:coarse-linear-k-bound}
\end{align}
Since $R_d(P_S) > 0$, there exist constants $\beta_+<\infty$ and $n_2$
such that every feasible $k\geq k_0$ satisfies
\begin{align}
    k
    \leq
    \beta_+n
    \qquad
    \text{for every } n\geq n_2.
    \label{eq:k-upper-linear-bound}
\end{align}
To see this, suppose otherwise. Then there would exist
pairs $(n_j,k_j)$ satisfying \eqref{eq:implicit-kstar-converse} and thus \eqref{eq:coarse-linear-k-bound} such that $k_j/n_j\to\infty$. Dividing
\eqref{eq:coarse-linear-k-bound} by $k_j$ would give
\begin{align*}
    R_d(P_S)
    &\leq
    \frac{n_j\log(2^m)}{k_j}
    +
    \left|Q^{-1}(\epsilon)\right|
    \sqrt{
        \frac{V_d(P_S)}{k_j}
        +
        \frac{
            n_j\bigl(\log(2^m)\bigr)^2
        }{
            4k_j^2
        }
    }
    +
    \frac{\log k_j}{2k_j}
    +
    \frac{K}{k_j},
\end{align*}
whose right-hand side converges to zero, contradicting
$R_d(P_S)>0$.

By \eqref{eq:combined-variance-linear-bounds},
\eqref{eq:kstar-second-case} implies the upper bound
\begin{align}
    \sum_{i=1}^n C\left(P_{Y|X}^{(i)}\right)
    -
    kR_d(P_S)
    <
    \left|Q^{-1}(\epsilon)\right|\sqrt{b_+n}.
    \label{eq:capacity-surplus-upper-bound}
\end{align}
On the other hand, using
\eqref{eq:k-upper-linear-bound} in
\eqref{eq:implicit-kstar-converse} gives
\begin{align}
    &\sum_{i=1}^n C\left(P_{Y|X}^{(i)}\right)
    -
    kR_d(P_S)
    \nonumber\\
    &\qquad\geq
    -
    \left|Q^{-1}(\epsilon)\right|
    \sqrt{
        \beta_+nV_d(P_S)
        +
        \frac{n}{4}\bigl(\log(2^m)\bigr)^2
    }
    \nonumber\\
    &\qquad\quad
    -
    \frac{1}{2}\log(\beta_+n)
    -
    K.
    \label{eq:capacity-surplus-lower-bound}
\end{align}
It follows from
\eqref{eq:capacity-surplus-upper-bound} and
\eqref{eq:capacity-surplus-lower-bound} that there exists a
finite constant $K_1$, independent of $n$ and $k$, such that
for all sufficiently large $n$,
\begin{align}
    \left|
        \sum_{i=1}^n C\left(P_{Y|X}^{(i)}\right)
        -
        kR_d(P_S)
    \right|
    \leq
    K_1\sqrt{n}.
    \label{eq:capacity-surplus-square-root-bound}
\end{align}

The variance appearing in
\eqref{eq:implicit-kstar-converse} is related to
$\mathsf{B}_n$ through the identity
\begin{align}
    &kV_d(P_S)
    +
    \sum_{i=1}^n V\left(P_{Y|X}^{(i)}\right)
    \nonumber\\
    &\qquad=
    \mathsf{B}_n
    -
    \frac{V_d(P_S)}{R_d(P_S)}
    \left(
        \sum_{i=1}^n C\left(P_{Y|X}^{(i)}\right)
        -
        kR_d(P_S)
    \right).
    \label{eq:combined-variance-identity}
\end{align}
By
\eqref{eq:combined-variance-linear-bounds},
\eqref{eq:capacity-surplus-square-root-bound} and
\eqref{eq:combined-variance-identity}, for all sufficiently
large $n$,
\begin{align}
    kV_d(P_S)
    +
    \sum_{i=1}^n V\left(P_{Y|X}^{(i)}\right)
    \geq
    \frac{b_-}{2}n.
    \label{eq:actual-variance-linear-lower-bound}
\end{align}
Therefore, using $\sqrt{x} - \sqrt{y} = \frac{x-y}{\sqrt{x} + \sqrt{y}}$, we obtain 
\begin{align}
    &\left|
        \sqrt{
            kV_d(P_S)
            +
            \sum_{i=1}^n V\left(P_{Y|X}^{(i)}\right)
        }
        -
        \sqrt{\mathsf{B}_n}
    \right|
    \nonumber\\
    &\quad=
    \frac{
        \dfrac{V_d(P_S)}{R_d(P_S)}
        \left|
            \sum_{i=1}^n C\left(P_{Y|X}^{(i)}\right)
            -
            kR_d(P_S)
        \right|
    }{
        \sqrt{
            kV_d(P_S)
            +
            \sum_{i=1}^n V\left(P_{Y|X}^{(i)}\right)
        }
        +
        \sqrt{\mathsf{B}_n}
    }
    \nonumber\\
    &\quad\leq
    \frac{
        \dfrac{V_d(P_S)}{R_d(P_S)}K_1\sqrt{n}
    }{
        \sqrt{\dfrac{b_-}{2}n}
        +
        \sqrt{b_-n}
    }.
    \label{eq:square-root-comparison}
\end{align}
Hence, there exists a finite constant $K_2$ such that
\begin{align}
    \left|
        \sqrt{
            kV_d(P_S)
            +
            \sum_{i=1}^n V\left(P_{Y|X}^{(i)}\right)
        }
        -
        \sqrt{\mathsf{B}_n}
    \right|
    \leq
    K_2.
    \label{eq:square-root-comparison-constant}
\end{align}

Furthermore, \eqref{eq:k-upper-linear-bound} gives
\begin{align}
    \log k
    \leq
    \log n
    +
    \log\left(\max\{1,\beta_+\}\right).
    \label{eq:log-k-log-n-comparison}
\end{align}
Substituting
\eqref{eq:square-root-comparison-constant} and
\eqref{eq:log-k-log-n-comparison} into
\eqref{eq:implicit-kstar-converse}, we obtain
\begin{align}
    &\sum_{i=1}^n C\left(P_{Y|X}^{(i)}\right)
    -
    kR_d(P_S)
    \nonumber\\
    &\qquad\geq
    Q^{-1}(\epsilon)\sqrt{\mathsf{B}_n}
    -
    \frac{1}{2}\log n
    -
    \widetilde{K}_1,
    \label{eq:explicit-capacity-surplus-converse}
\end{align}
where
\begin{align}
    \widetilde{K}_1
    :=
    K
    +
    \left|Q^{-1}(\epsilon)\right|K_2
    +
    \frac{1}{2}
    \log\left(\max\{1,\beta_+\}\right).
    \label{eq:explicit-capacity-surplus-constant}
\end{align}
The proof is finished by noting that \eqref{eq:explicit-capacity-surplus-converse} implies \eqref{desired_uppbndjgkg4} for 
 $\widetilde{K} = \frac{\widetilde{K}_1}{R_d(P_S)}$. 

\section{Proof of Lemma \ref{lemmafirstlowerboundthen} \label{lemmafirstlowerboundthen_proof}}

Let $Q_{X^n, \alpha} = \mathcal{N}(\mathbf{0}, \alpha I_n)$ and $U_{X^n, \alpha}
=
\operatorname{Unif}
\left(
\left\{
x^n : \|x^n\|^2 = n\alpha
\right\}
\right).$ We first write
\begin{align}
    \imath_{\widehat{P}_{X^n|Y^n}, Q_{X^n}^{\operatorname{u}}}
    (Y^n, X^n)
    &=
    \log
    \left(
    \frac{
    (2\pi \widehat{p}(Y^n))^{-n/2}
    \exp\left(
    -\frac{\|Y^n-X^n\|^2}{2\widehat{p}(Y^n)}
    \right)
    }{
    Q(Y^n)
    }
    \right)
    -
    \log \frac{U(Y^n)}{Q(Y^n)},
    \label{firstermandsecondterm}
\end{align}
where
\begin{align}
    U(y^n)
    &\coloneqq
    \int_{\mathbb{R}^n}
    (2\pi \widehat{p}(y^n))^{-n/2}
    \exp\left(
    -\frac{\|y^n-u^n\|^2}{2\widehat{p}(y^n)}
    \right)
    dU_{X^n,\alpha}(u^n), \label{lkej3.t11}\\
    Q(y^n)
    &\coloneqq
    \int_{\mathbb{R}^n}
    (2\pi \widehat{p}(y^n))^{-n/2}
    \exp\left(
    -\frac{\|y^n-u^n\|^2}{2\widehat{p}(y^n)}
    \right)
    dQ_{X^n,\alpha}(u^n), \label{lkej3.t}\\
    \overline{Q}(y^n)
    &\coloneqq
    \int_{\mathbb{R}^n}
    \exp\left(
    -\frac{\|y^n-u^n\|^2}{2\widehat{p}(y^n)}
    \right)
    dQ_{X^n,\alpha}(u^n). \notag 
\end{align}
Since $(X^n,Y^n)
\sim
U_{X^n,\alpha}\times P_{Y^n|X^n},$ we have $\|X^n\|^2=n\alpha$ almost surely and $Y^n=X^n+Z^n,$ where $Z^n\sim\mathcal{N}(\mathbf{0},p^2I_n).$ Hence, the first term in \eqref{firstermandsecondterm} can be
written as
\begin{align*}
    \log
    \frac{
    \exp\left(
    -\frac{\|Y^n-X^n\|^2}{2\widehat{p}(Y^n)}
    \right)
    }{
    \overline{Q}(Y^n)
    }
    &=
    -\frac{\|Y^n-X^n\|^2}{2\widehat{p}(Y^n)}
    -
    \frac{n}{2}
    \log
    \left(
    \frac{\widehat{p}(Y^n)}
    {\alpha+\widehat{p}(Y^n)}
    \right)
    +
    \frac{\|Y^n\|^2}
    {2(\alpha+\widehat{p}(Y^n))}
    \\
    &=
    \frac{n}{2}
    \log
    \left(
    1+\frac{\alpha}{\widehat{p}(Y^n)}
    \right)
    -
    \frac{\|Y^n-X^n\|^2}{2\widehat{p}(Y^n)}
    +
    \frac{\|Y^n\|^2}
    {2(\alpha+\widehat{p}(Y^n))}.
\end{align*}
To prove Lemma~\ref{lemmafirstlowerboundthen}, it suffices to
upper bound
\[
\log\frac{U(Y^n)}{Q(Y^n)}
\]
by a constant on an event having probability at least
$1-O(n^{-1})$. Define
\[
T_n
\coloneqq
\frac{\|Y^n\|^2}{n}-\alpha.
\]
Then
\[
\mathbb{E}[T_n]=p^2,
\qquad
\operatorname{Var}(T_n)
=
\frac{2p^4}{n}
+
\frac{4p^2\alpha}{n}.
\]
In particular, $\widehat{p}(Y^n)\to p^2$ in probability. For any fixed $t>0$ and $\alpha > 0$, define
\begin{align*}
    U_{\alpha,t}(y^n)
    &\coloneqq
    \int_{\mathbb{R}^n}
    (2\pi t)^{-n/2}
    \exp\left(
    -\frac{\|y^n-u^n\|^2}{2t}
    \right)
    dU_{X^n,\alpha}(u^n),\\
    Q_{\alpha,t}(y^n)
    &\coloneqq
    \int_{\mathbb{R}^n}
    (2\pi t)^{-n/2}
    \exp\left(
    -\frac{\|y^n-u^n\|^2}{2t}
    \right)
    dQ_{X^n,\alpha}(u^n).
\end{align*}
Then note that 
\begin{align}
\begin{split}
    U_{\alpha,t}(y^n)
    &=
    \frac{1}{t^{n/2}}
    U_{\frac{\alpha}{t},1}
    \left(
    \frac{y^n}{\sqrt{t}}
    \right),\\
    Q_{\alpha,t}(y^n)
    &=
    \frac{1}{t^{n/2}}
    Q_{\frac{\alpha}{t},1}
    \left(
    \frac{y^n}{\sqrt{t}}
    \right).
\end{split}
\label{molvako}
\end{align}
We now fix
\[
t_0\coloneqq\frac{p^2}{2},
\qquad
\rho_0\coloneqq\frac{\alpha}{t_0}
=\frac{2\alpha}{p^2}.
\]
For $\rho>0$ and $r\geq 0$, define
\begin{align}
    b_\rho(r)
    &\coloneqq
    \frac{c_0}{2}
    \frac{
    1+\sqrt{1+4\rho r}
    }{
    (1+4\rho r)^{1/4}
    },
    \label{brhodef}\\
    f_\rho(r)
    &\coloneqq
    (1+\rho)
    -
    \log\bigl(2(1+\rho)\bigr)
    +
    \frac{\rho r}{1+\rho}
    -
    \sqrt{1+4\rho r}
    +
    \log\left(
    1+\sqrt{1+4\rho r}
    \right),
    \label{frhodef}
\end{align}
where
\[
c_0=27\sqrt{\frac{\pi}{8}}.
\]
For the single fixed value $\rho=\rho_0$,
\cite[(102)--(103)]{molavianjazi} gives the existence of an
integer $n_1=n_1(\rho_0)$ such that, for every $n\geq n_1$
and every $z^n\in\mathbb{R}^n$,
\begin{align}
    \frac{
    U_{\rho_0,1}(z^n)
    }{
    Q_{\rho_0,1}(z^n)
    }
    &\leq
    b_{\rho_0}(r)
    \exp\left(
    -\frac{n}{2}f_{\rho_0}(r)
    \right),
    \qquad
    r=\frac{\|z^n\|^2}{n}.
    \label{eq102used}
\end{align}
Direct differentiation gives
\begin{align*}
    f_\rho'(r)
    &=
    \rho
    \left(
    \frac{1}{1+\rho}
    -
    \frac{2}{1+\sqrt{1+4\rho r}}
    \right).
\end{align*}
It follows that $f_\rho$ is decreasing on $[0,1+\rho]$,
increasing on $[1+\rho,\infty)$, and $f_\rho(1+\rho)=0.$ Consequently, $f_\rho(r)\geq 0$ for every $r \geq 0$. Moreover, as $r\to\infty$,
\begin{align*}
    b_\rho(r)&=O(r^{1/4}),\\
    f_\rho(r)
    &=
    \frac{\rho}{1+\rho}r
    -
    2\sqrt{\rho r}
    +
    O(\log r),
\end{align*}
and hence $f_\rho(r)\to\infty$. Define
\begin{align}
    K_0
    &\coloneqq
    \sup_{r\geq 0}
    b_{\rho_0}(r)
    \exp\left(
    -\frac{n_1}{2}f_{\rho_0}(r)
    \right).
    \label{Kzerodef}
\end{align}
The preceding asymptotic estimates imply that $K_0<\infty$.
Furthermore, since $f_{\rho_0}(r)\geq 0$, for every $n\geq n_1$,
\begin{align}
    \sup_{z^n\in\mathbb{R}^n}
    \frac{
    U_{\rho_0,1}(z^n)
    }{
    Q_{\rho_0,1}(z^n)
    }
    &\leq K_0.
    \label{fixedrhoK}
\end{align}
By \eqref{molvako} and \eqref{fixedrhoK}, for every
$n\geq n_1$,
\begin{align}
    \sup_{y^n\in\mathbb{R}^n}
    \frac{
    U_{\alpha,t_0}(y^n)
    }{
    Q_{\alpha,t_0}(y^n)
    }
    &\leq K_0.
    \label{fixedendpoint}
\end{align}
We now extend \eqref{fixedendpoint} from the single
variance $t_0$ to every $t\geq t_0$. Let
\[
\phi_s^{(n)}(x^n)
\coloneqq
(2\pi s)^{-n/2}
\exp\left(
-\frac{\|x^n\|^2}{2s}
\right),
\qquad s>0.
\]
For $t>t_0$, the Gaussian convolution identity gives
\begin{align*}
    U_{\alpha,t}(y^n)
    &=
    \int_{\mathbb{R}^n}
    U_{\alpha,t_0}(z^n)
    \phi_{t-t_0}^{(n)}(y^n-z^n)
    dz^n,\\
    Q_{\alpha,t}(y^n)
    &=
    \int_{\mathbb{R}^n}
    Q_{\alpha,t_0}(z^n)
    \phi_{t-t_0}^{(n)}(y^n-z^n)
    dz^n.
\end{align*}
Using \eqref{fixedendpoint} and the nonnegativity of the
Gaussian density, we obtain
\begin{align*}
    U_{\alpha,t}(y^n)
    &\leq
    K_0
    \int_{\mathbb{R}^n}
    Q_{\alpha,t_0}(z^n)
    \phi_{t-t_0}^{(n)}(y^n-z^n)
    dz^n\\
    &=
    K_0 Q_{\alpha,t}(y^n).
\end{align*}
Together with \eqref{fixedendpoint}, this proves that, for
every $n\geq n_1$,
\begin{align}
    \sup_{t\geq t_0}
    \sup_{y^n\in\mathbb{R}^n}
    \frac{
    U_{\alpha,t}(y^n)
    }{
    Q_{\alpha,t}(y^n)
    }
    &\leq K_0.
    \label{uniformtbound}
\end{align}
Define the event
\[
A_n
\coloneqq
\left\{
\widehat{p}(Y^n)\geq \frac{p^2}{2}
\right\}.
\]
On the event $A_n$, we have from \eqref{uniformtbound} that 
\begin{align}
    \frac{
    U_{\alpha,\widehat{p}(Y^n)}(Y^n)
    }{
    Q_{\alpha,\widehat{p}(Y^n)}(Y^n)
    }
    &\leq K_0.
    \label{randomtbound}
\end{align}
The above step follows because
\eqref{uniformtbound} holds simultaneously for every
$t\geq t_0$ and every $y^n\in\mathbb{R}^n$, with a common
blocklength threshold $n \geq n_1$. Let $C_1\coloneqq\log K_0$.
It follows from \eqref{randomtbound} that, on $A_n$,
\begin{align*}
    \log\frac{U(Y^n)}{Q(Y^n)}
    &=
    \log
    \frac{
    U_{\alpha,\widehat{p}(Y^n)}(Y^n)
    }{
    Q_{\alpha,\widehat{p}(Y^n)}(Y^n)
    }
    \leq C_1.
\end{align*}
It remains to show that the event $A_n$ has probability at least $1 - O(n^{-1})$. For all sufficiently large $n$ we
have $1/n<p^2/2$. Therefore,
\begin{align*}
    \mathbb{P}(A_n^c)
    &=
    \mathbb{P}
    \left(
    \widehat{p}(Y^n)<\frac{p^2}{2}
    \right)\\
    &\leq
    \mathbb{P}
    \left(
    \frac{\|Y^n\|^2}{n}-\alpha
    <
    \frac{p^2}{2}
    \right)\\
    &\leq
    \mathbb{P}
    \left(
    |T_n-p^2|
    \geq
    \frac{p^2}{2}
    \right)\\
    &\leq
    \frac{
    4\operatorname{Var}(T_n)
    }{
    p^4
    } = O\left(\frac{1}{n} \right).
\end{align*}

\section{Proof of Lemma \ref{mytwosidedbndprop} \label{mytwosidedbndprop_proof}}

We start by rewriting 
\begin{align*}
    \Lambda(X^n, Y^n) &= \frac{n}{2} \log \left(1 +  \frac{\alpha}{\widehat{p}(Y^n)} \right) - \frac{||Z^n||^2}{2 \widehat{p}(Y^n)} + \frac{||Y^n||^2}{2(\alpha + \widehat{p}(Y^n) )}. 
\end{align*}
We use the following equalities in distribution: 
\begin{align*}
    X^n &= \sqrt{n \alpha} \frac{G^n}{||G^n||}, \quad Z^n = p N^n,
\end{align*}
where $G^n, N^n$ are independent random vectors in $\mathbb{R}^n$, and both have the standard normal distribution $\mathcal{N}(\mathbf{0}, I_n)$. We will show that on a high probability event $E_n$ to be specified later, $\Lambda(X^n, Y^n) = n f(U_n, M_n, W_n)$ can be written as a smooth function of the following empirical means of i.i.d. random variables: 
\begin{align*}
    U_n &= \frac{1}{n} \sum_{i=1}^n p^2 N_i^2, \quad M_n = \frac{1}{n} \sum_{i=1}^n p G_i N_i, \quad W_n = \frac{1}{n} \sum_{i=1}^n G_i^2.  
\end{align*}
Define $T_n = \frac{||Y^n||^2}{n} - \alpha$. Let the event $E_n \subset \left \{T_n > \frac{1}{n} \right \}$, where the superset is a high probability event as $n \to \infty$ because $T_n \to p^2 > 0$. Specifically, 
\begin{align*}
    &\mathbb{P} \left(\frac{||Y^n||^2}{n} - \alpha \leq \frac{1}{n} \right)\\
    &= \mathbb{P} \left(\frac{||Y^n||^2}{n} - \alpha - p^2 \leq \frac{1}{n} - p^2 \right)\\
    &\leq \mathbb{P} \left(\left |\frac{||Y^n||^2}{n} - \alpha - p^2 \right| \geq  p^2 - \frac{1}{n} \right)\\
    &\leq O\left(\frac{1}{n} \right). 
\end{align*}
Hence, on the high probability event $E_n$ (yet to be fully specified), we have $\widehat{p}(Y^n)=  T_n$ and  
\begin{align}
    \Lambda(X^n, Y^n) &= \frac{n}{2} \log \left(1 +  \frac{\alpha}{T_n} \right) - \frac{||Z^n||^2}{2 T_n} + \frac{||Y^n||^2}{2(\alpha + T_n )}\\
    &= \frac{n}{2} \log \left(1 +  \frac{\alpha}{T_n} \right) - \frac{p^2 ||N^n||^2}{2 T_n} + \frac{n \alpha + n T_n}{2(\alpha + T_n )}\\
    &= \frac{n}{2} \log \left(1 +  \frac{\alpha}{T_n} \right)  + \frac{n}{2} - \frac{n U_n}{2T_n}, \label{b25cu}
\end{align}
where $T_n$ can further be expressed in terms of $U_n, M_n$ and $W_n$ as 
\begin{align*}
    T_n &= \frac{2p\sqrt{\alpha}}{\sqrt{n} ||G^n||}  (G^n)^T N^n + \frac{p^2}{n} ||N^n||^2\\
    &= U_n + \frac{2 \sqrt{\alpha}M_n}{\sqrt{W_n}}.  
\end{align*}
Therefore, $\Lambda(X^n, Y^n) = nf(U_n, M_n, W_n)$ on the event $E_n$ where we define
\begin{align*}
    f(u, m, w) &= \frac{1}{2} \log \left(1 +  \frac{\alpha}{t(u, m, w)} \right)  + \frac{1}{2} - \frac{u}{2 t(u, m, w)},
\end{align*}
where 
\begin{align*}
    t(u, m, w) &= u + \frac{2 \sqrt{\alpha} m}{\sqrt{w}}. 
\end{align*}
Note that the function $f$ is $C^2$ on 
\begin{align*}
    \mathcal{D} \coloneqq \left \{ (u, m, w): w > 0, t(u, m, w) > 0 \right \}.
\end{align*}
Now define a random vector $\mathbf{V}_i \in \mathbb{R}^3$ as $\mathbf{V}_i = (p^2 N_i^2, p G_i N_i, G_i^2 )$. Then 
\begin{align*}
    f(U_n, M_n, W_n) &= f \left(\frac{1}{n} \sum_{i=1}^n \mathbf{V}_i \right),
\end{align*}
where $\{\mathbf{V}_i \}_{i=1}^\infty$ is a sequence of i.i.d. random vectors. We have 
\begin{align*}
    \overrightarrow{\mu} = \mathbb{E}[\mathbf{V}_1] &= (p^2, 0, 1)^T,\\
    f(\overrightarrow{\mu}) &= \frac{1}{2} \log \left(1 + \frac{\alpha}{p^2} \right),\\
    \nabla f(\overrightarrow{\mu}) &= \left(-\frac{\alpha}{2p^2(\alpha + p^2)}, \frac{\sqrt{\alpha}}{\alpha + p^2}, 0  \right)^T. 
\end{align*}
Fix some constant $A > 0$ to be specified later. Let $\delta_n = A \sqrt{\log(n)/n}$. Choose $n$ large enough so that    
\begin{align*}
    B_n \coloneqq \left \{ (u, m, w) : |u - p^2| \leq \delta_n, |m| \leq \delta_n, |w - 1| \leq \delta_n \right \}
\end{align*}
is contained in $\mathcal{D}$ and $t(u, m, w) > \frac{1}{n}$ for all $(u, m, w) \in B_n$. Since $B_n$ is compact and $f \in C^2(B_n)$, for sufficiently large $n$, there exists 
\begin{align*}
    L \coloneqq  \sup_{\overrightarrow{v} \in B_n} ||\nabla^2 f(\overrightarrow{v})||_{\operatorname{op}} < \infty,
\end{align*}
where $L$ can be chosen to be independent of $n$, and $||\cdot ||_{\operatorname{op}}$ is the operator norm induced by the Euclidean norm on vectors in $\mathbb{R}^3$.

Finally, we define the event $E_n = \left \{\frac{1}{n} \sum_{i=1}^n \mathbf{V}_i  \in B_n \right \}$. On this event, we have 
\begin{align*}
    &\left |f \left(\frac{1}{n} \sum_{i=1}^n \mathbf{V}_i \right) - f(\overrightarrow{\mu}) - \nabla f(\overrightarrow{\mu})^T\left(\frac{1}{n} \sum_{i=1}^n \mathbf{V}_i - \overrightarrow{\mu} \right) \right| \\
    &\leq \frac{L}{2} \left | \left | \frac{1}{n} \sum_{i=1}^n \mathbf{V}_i - \overrightarrow{\mu} \right | \right |^2\\
    &\leq \frac{3LA^2 \log(n)}{2n} = \frac{C_1 \log n}{n}.
\end{align*}
Then note that 
\begin{align*}
  \nabla f(\overrightarrow{\mu})^T\left(\frac{1}{n} \sum_{i=1}^n \mathbf{V}_i - \overrightarrow{\mu} \right) &= \frac{1}{n} \sum_{i=1}^n \left [ \frac{p \sqrt{\alpha}}{\alpha + p^2} G_i N_i - \frac{\alpha}{2(\alpha + p^2)} (N_i^2 - 1) \right ]. 
\end{align*}
Lastly, it remains to show that $\operatorname{Pr}(E_n^c) \leq C_2/n$ for some constant $C_2$. This can be shown by the union bound applied to the three events
\begin{align*}
    &\left \{ \left| \frac{1}{n} \sum_{i=1}^n \left(p^2 N_i^2 - p^2\right) \right| > \delta_n \right \},\\
    &\left \{ \left| \frac{1}{n} \sum_{i=1}^n p N_i G_i \right| > \delta_n \right \}, \text{ and} \\
    &\left \{ \left| \frac{1}{n} \sum_{i=1}^n \left( G_i^2 - 1\right) \right| > \delta_n \right \},
\end{align*}
followed by standard Chernoff bounds with a suitable choice of the constant $A$.

\section{Examples of Sources \label{examples_of_sources}}

The main results in the paper apply to a source distribution $P_S$, distortion level $d$ and distortion measure $\operatorname{d}$ satisfying Assumptions 
\ref{assumpr1_source}\,-\,\ref{assumpr3_source} and Conditions \ref{cond1dee}\,-\,\ref{weird9thmomcond}. In this section, we give examples of source-distortion tuples $(P_S, d, \operatorname{d})$ that satisfy these conditions.  

\subsection{Gaussian source with squared-error distortion measure}

For a Gaussian source $S \sim P_S$ with squared-error distortion measure $\operatorname{d}$, we have 
\begin{align*}
\mathcal S&=\mathcal Z=\mathbb R,\\
    P_S &= \mathcal{N}(0, \sigma_S^2),\\
    0&<\sigma_S^2<\infty,\\    
    \operatorname d(s,z)&=(s-z)^2.  
\end{align*}
For any distortion level \(d\in(0,\infty)\), the
rate-distortion function of the Gaussian source is
\begin{align}
    R_d(P_S) = \begin{cases}
         \frac12\log\frac{\sigma_S^2}{d}, & 0<d<\sigma_S^2,\\[0.5em]
        0, & d\ge \sigma_S^2.
    \end{cases}
\end{align}
Hence, $d_{\min} = 0$ and $d_{\max} = \sigma_S^2$. For $d \in (0, \sigma_S^2)$, the infimum in $(\ref{64})$ is uniquely achieved by the following conditional distribution \cite[10.3.2]{Cover2006}:
\begin{align}
    P_{Z|S}^*(\cdot|s)
    =
    \mathcal N\!\left(
        \left(1-\frac{d}{\sigma_S^2}\right)s,
        d\left(1-\frac{d}{\sigma_S^2}\right)
    \right),
\end{align}
and the reproduction marginal is $P_Z^*=\mathcal N(0,\sigma_S^2-d).$ Equivalently, the rate-distortion-achieving joint law $P_{S,Z}^* = P_S \times P_{Z|S}^*$ admits the backward
representation $S = Z + N$, where $Z\sim \mathcal N(0,\sigma_S^2-d), N\sim\mathcal N(0,d)$ and $Z \ind N$. Assumption \ref{assumpr3_source} can be verified to hold by simply taking $E = \{0 \}$. To verify that Condition \ref{weird9thmomcond} holds, let \(S\sim P_S\) and \(Z\sim P_Z^*\) be independent. Then since $S-Z\sim\mathcal N(0,2\sigma_S^2-d)$, we have 
$\mathbb E\!\left[(\operatorname d(S,Z))^9\right]
    =
    \mathbb E[(S-Z)^{18}]
    <\infty$. Furthermore, for $d \in (0, \sigma_S^2)$, we have 
\begin{align*}
    \lambda^* &= - \partial_d R_d(P_S) = \frac{1}{2d},\\
    \jmath_d(s, P_S, \operatorname{d}) &= \frac12\log\left(\frac{\sigma_S^2}{d}\right) + \frac{1}{2} \left(\frac{s^2}{\sigma_S^2} - 1 \right),\\
    V_d(P_S) &= \frac{1}{2},\\
    \overline{c} &= 1 + \frac{d^2}{\sigma_S^4 - d^2}, 
\end{align*}
where $\overline{c}$ was defined in $(\ref{c28f})$. 

\subsection{Bernoulli source with Hamming distortion measure}

For a Bernoulli source $S \sim P_S$ with Hamming distortion measure $\operatorname{d}$, we have 
\begin{align*}
\mathcal S&=\mathcal Z= \{0,1\},\\
    P_S &= \operatorname{Bern}(p),\\
    0&<p<1,\\    
    \operatorname d(s,z)&= \mathds{1}\{s \neq z\}.  
\end{align*}
For any distortion level \(d\in[0,\infty)\), the
rate-distortion function for a Bernoulli$(p)$ source is given by
\begin{align}
    R_d(P_S) = \begin{cases}
         H_b(p) - H_b(d), & 0 \leq d \leq \min\{p, 1-p \},\\[0.5em]
        0, & d > \min\{p, 1-p \}.
    \end{cases}
\end{align}
Hence, $d_{\min} = 0$ and $d_{\max} = \min\{p, 1-p \}$. For $d \in (0, d_{\max})$, the infimum in $(\ref{64})$ is uniquely achieved by the following conditional distribution \cite[10.3.1]{Cover2006}:

\begin{table}[H]
    \centering
    \renewcommand{\arraystretch}{2.5} 
    \begin{tabular}{c|cc}
        $P_{Z|S}^*(z | s)$ & $z = 0$ & $z = 1$ \\
        \hline
        $s = 0$ & $\frac{(1-d)(1-p-d)}{(1-p)(1-2d)}$ & $\frac{d(p-d)}{(1-p)(1-2d)}$ \\
        $s = 1$ & $\frac{d(1-p-d)}{p(1-2d)}$ & $\frac{(1-d)(p-d)}{p(1-2d)}$. \\
    \end{tabular}
\label{tab:rd_conditional_dist}
\end{table}
The reproduction marginal is $P_Z^*= \left(\frac{1-p-d}{1-2d}, \frac{p - d}{1-2d} \right).$ Equivalently, the rate-distortion-achieving joint law $P_{S,Z}^* = P_S \times P_{Z|S}^*$ admits the backward
representation $S = Z \oplus N$, where $Z\sim P_Z^*, N\sim\operatorname{Bern}(d)$ and $Z \ind N$. Clearly, Assumption \ref{assumpr3_source} and Condition \ref{weird9thmomcond} hold. Furthermore, for $0 < d  < \min \{p, 1-p \}$, we have 
\begin{align*}
    \lambda^* &= - \partial_d R_d(P_S) = \log\left(\frac{1-d}{d} \right),\\
    \jmath_d(s, P_S, \operatorname{d}) &= \log\left(\frac{1}{P_S(s)} \right) - H_b(d),\\
    V_d(P_S) &= p(1-p) \log^2\left(\frac{1-p}{p} \right),\\
    \overline{c} &= 1 + \frac{d(1-d)(2p-1)^2}{(p-d)(1-p-d)}. 
\end{align*}

\ifCLASSOPTIONcaptionsoff
  \newpage
\fi



\bibliographystyle{IEEEtran}

\bibliography{citations}

%








\end{document}